\documentclass[journal]{IEEEtran}
\usepackage{color}
\usepackage{mathrsfs}
\usepackage{amsfonts}
\usepackage{amsthm}
\theoremstyle{remark}
\newtheorem{definition}{Definition}
\newtheorem{problem}{Problem}
\newtheorem{corollary}{Corollary}
\newtheorem{proposition}{Proposition}
\newtheorem{theorem}{Theorem}
\newtheorem{lemma}{Lemma}
\newtheorem{remark}{Remark}
\newtheorem{example}{Example}
\usepackage{algorithm}
\usepackage{cite}
\usepackage{fancyhdr}
\usepackage{bbm}
\usepackage{datetime}

\usepackage{hyperref}
\ifCLASSINFOpdf
   \usepackage[pdftex]{graphicx}
\else
   \usepackage[dvips]{graphicx}
   \graphicspath{{../eps/}}
   \DeclareGraphicsExtensions{.eps}
\fi

\usepackage[cmex10]{amsmath}
\usepackage{algorithmic}
\begin{document}

\title{\vspace{0.4cm}Computing Linear Transformations with Unreliable Components}
\author{
\IEEEauthorblockN{Yaoqing Yang, Pulkit Grover and Soummya Kar}
%
\thanks{A preliminary version of this work~\cite{yang2016computing} was presented in part at the 2016 IEEE International Symposium on Information Theory (ISIT). This work is supported by NSF ECCS-1343324, NSF CCF-1350314 (NSF CAREER) and NSF CNS-1702694 for Pulkit Grover, NSF ECCS-1306128, NSF CCF-1513936, the Bertucci Graduate Fellowship for Yaoqing Yang, and by Systems on Nanoscale Information fabriCs (SONIC), one of the six SRC STARnet Centers, sponsored by MARCO and DARPA.

Y. Yang, P. Grover and S. Kar are with the Department of Electrical and Computer Engineering, Carnegie Mellon University, Pittsburgh, PA, 15213, USA. Email: \{yyaoqing,pgrover,soummyak\}@andrew.cmu.edu}
}
\maketitle
\rfoot{}
\renewcommand{\headrulewidth}{0pt}


\vspace{-0.6in}

\begin{abstract}
We consider the problem of computing a binary linear transformation when all circuit components are unreliable. Two models of unreliable components are considered: probabilistic errors and permanent errors. We introduce the ``ENCODED'' technique that ensures that the error probability of the computation of the linear transformation is kept bounded below a small constant independent of the size of the linear transformation even when all logic gates in the computation are noisy. By deriving a lower bound, we show that in some cases, the computational complexity of the ENCODED technique achieves the optimal scaling in error probability. Further, we examine the gain in energy-efficiency from use of a ``voltage-scaling'' scheme where gate-energy is reduced by lowering the supply voltage. We use a gate energy-reliability model to show that tuning gate-energy appropriately at different stages of the computation (``dynamic'' voltage scaling), in conjunction with ENCODED, can lead to orders of magnitude energy-savings over the classical ``uncoded'' approach. Finally, we also examine the problem of computing a linear transformation when noiseless decoders can be used, providing upper and lower bounds to the problem.
\end{abstract}


\section{Introduction}\label{Intro}
It is widely believed that noise and variation issues in modern low-energy and low-area semiconductor devices call for new design principles of circuits and systems~\cite{Bor_Micro_05,Shan_DTC_08,haque2010hard}. From an energy-viewpoint, an urgent motivation for studying noise in circuits comes from saturation of ``Dennard's scaling'' of energy with smaller technology~\cite{Dennard}. Reducing CMOS transistor size no longer leads to a guaranteed reduction in energy consumption. Many novel devices are being explored to continue reducing energy consumption, e.g.~\cite{Nat_CUP_14}. However, such emerging low-energy technologies generally lack the reliability of CMOS. On the other hand, aggressive design principles, such as ``voltage-scaling''~(which is commonly used in modern circuits), reduce energy consumption~\cite{Pil_ACM_01}, but often at a reliability cost: when the supply voltage is reduced below the transistor's threshold voltage, component variability results in reduced control of component reliability. From an area viewpoint, as transistors become smaller and clock frequencies become higher, noise margin of semiconductor devices is reduced~\cite{Zhao_TR_07}. In fact, voltage variation, crosstalk, timing jitter, thermal noise caused from increased power density and quantum effects can all jeopardize reliability.\footnote{e.g.~in source-drain channels of small CMOS transistors, the number of electrons can be so few that laws of large numbers may not apply~\cite{Mir_Spec_12}, increasing the component variability.} Beyond CMOS, circuits for many emerging technologies, such as those built out of carbon-nanotubes~\cite{Shu_Nature_13}, suffer from reliability problems, such as wire misalignment and metallic carbon-nanotubes~\cite{Pat_TCAD_08}. Thus, for a host of factors, circuit reliability is becoming an increasingly important issue.

While most modern implementations use overwhelmingly reliable transistors, an appealing idea is to deliberately allow errors in computation, and design circuits and systems that ``embrace randomness and statistics, treating them as opportunities rather than problems''~\cite{Shan_DTC_08}. Inspired by the triumph of Shannon theory in dealing with noise in communication channels~\cite{Sha_Bel_48}, von Neumann initialized the study of noise in circuits~\cite{Neu_Aut_56}. He showed that even when circuit components are noisy, it is possible to bias the output towards the correct output using repetition-based schemes. Repeated computations, followed by majority voting, have been used in some applications to make circuits error-tolerant~\cite{Abda_JSSC_13,Han_VLSI_14,Han_DTC_05}. In fact, many functional-block-level or algorithmic error-tolerant designs have been tested on real systems~\cite{Kim_ACMMi_03,Cho_TCAD_12,Abda_JSSC_13}. However, in absence of a comprehensive understanding of the fundamental tradeoffs between redundancy and reliability, these designs have no guarantees on the gap from optimality.
\pagestyle{fancy}
\rfoot{}

Can use of sophisticated codes help? For storage, which can be viewed as computing the identity function, Low-Density Parity-Check (LDPC) codes~\cite{Gal_TIT_62} and Expander codes~\cite{Sip_FCS_96} have been used to correct errors~\cite{Tay_Bel_68,Kuz_PIT_73,Chi_ITW_07,Gun_IAC_08}. Closer in spirit of computation with noisy elements, in~\cite{Tay_Bel_68,Kuz_PIT_73,Chi_ITW_07}, the decoders (though not the encoders) for storage are assumed to be noisy as well. In \cite{huang2015acoco}, adaptive coding is used to correct memory faults for fault-tolerant approximate computing. Decoding with noisy elements has become an area of active research~\cite{Var_TIT_11,Yaz_TC_01,Huang_TC_14,ref_1,ref_2}. In~\cite{Var_TIT_11,Yaz_TC_01,Huang_TC_14,ref_1,ref_2}, noisy decoders performing message-passing algorithms are analyzed \textcolor{black}{using} the density evolution \textcolor{black}{technique}~\cite{Ric_TIT_01_2,Len_TIT_01}. The idea of using noisy decoders, and yet achieving reliable performance, is further extended to noisy discrete-time error-resilient linear systems in~\cite{Had_TIT_05}, where LDPC decoding is utilized to correct state errors after each state transition. Error control coding is also used in fault-tolerant parallel computing~\cite{Spi_FCS_96}, AND-type one-step computing with unreliable components~\cite{Rac_ISIT_08} and applied to error-resilient systems on chips (SoCs)~\cite{Bert_CAD_05}. Unfortunately, all of the above works on using sophisticated codes in noisy computing have one major intellectual and practical shortcoming: while they use noisy gates to perform some computations, they all assume absolute reliability in either the encoding part, or the decoding part, or both.

The perspective of allowing some noiseless gates in noisy computing problems has permeated in the investigation of fundamental limits as well (e.g.~\cite{Simon_ITW_11,Simon_ITW_10,Koch1}), where, assuming that encoding and/or decoding are free, the authors derive fundamental limits on required resources for computation with noisy elements with no assumptions on the computation strategy. Can one choose to ignore costs associated with encoding or decoding? While ignoring these costs is reasonable in long-range noisy communication problems~\cite{Gro_JSAC_11}, where the required transmit energy tends to dominate encoding/decoding computation energy, recent work shows this can yield unrealistically optimistic results in short-range communication~\cite{Gro_JSAC_11,Gro_ISIT_12,Gro_TIT_15,Blake1,Blake2} and noisy computing~\cite{Gro_ISIT_14}, especially in the context of energy. These works derive fundamental limits for simplistic implementation models that account for total energy consumption, including that of encoding and decoding, in communication~\cite{Gro_TIT_15,Gro_ISIT_12,Blake1,Blake2} and computing~\cite{Gro_ISIT_14}.

In this paper, we investigate the problem of reliable\footnote{Note that the notion of reliability here differs from that in Shannon theory. The goal here is to bound the error-probability by a small constant that depends on the gate-error probability, but does \textit{not} depend on the size of the computation.} computation of binary linear transformations using circuits built entirely out of unreliable components, \textit{including} the circuitry for introducing redundancy and correcting errors. In Section~\ref{main_results}, we study the problem of computing linear transformations using homogeneous noisy gates, all of which are drawn from the same faulty gate model. We consider both probabilistic error models (transient gate errors)~\cite{Kar_TDSC_04} and permanent-errors models ({defective} gates)~\cite{Huang_TC_14}. The problem formulation and reliability models are detailed in Section~\ref{System_Model}.

The key to our construction is the ``ENCODED'' technique (\textbf{En}coded \textbf{Co}mputation with \textbf{De}coders Embedde\textbf{D}), in which noisy decoders are embedded inside the noisy encoder to repeatedly suppress errors (Section~\ref{Encoder_Sec}). The entire computation process is partitioned into multiple stages by utilizing the properties of an encoded form of the linear transformation matrix (see Section~\ref{pipeline_str} for details). In each stage, errors are introduced due to gate failures, and then suppressed by embedded noisy decoders~\cite{Yaz_TC_01}, preventing them from accumulating. Intuition on why embedded decoders are useful is provided in Section~\ref{Intuition}.

In Section~\ref{main_results} and~\ref{Encoder_Sec}, we show that using ENCODED with LDPC decoders, an $L\times K$ binary linear transformation can be computed with $\mathcal{O}(L)$ operations per output bit, while the output bit error probability is maintained below a small constant that is independent of $L$ and $K$. In Section~\ref{expander_encoder}, we use expander LDPC codes to achieve worst-case error tolerance using these codes, while still using error-prone decoding circuitry. We show that ENCODED can tolerate defective gates errors as long as the fraction of defective gates is below a small constant. We also obtain a stronger result on the computational complexity when the block error probability, instead of bit error probability, is specified: by deriving a fundamental lower bound (when the linear transform has full row rank), we show that the computational complexity per bit matches the lower bound in the scaling of the target error probability, as the required block error probability approaches zero. Interestingly, in the derivation of this lower bound, we allow the circuit to use noiseless gates to perform decoding operations. In Section~\ref{simulation_sec}, we use simulations to show that using exactly the same types of noisy gates (even with the same fan-in), the achieved bit error ratio and the number of iterations of ENCODED are both smaller than those of repetition-based schemes. Since computing energy is closely related to the number of operations, this shows an energy advantage of our ENCODED technique as well.

In Section~\ref{vs}, we go a step further and systematically study the effect of tunable supply voltage (``dynamic'' voltage scaling) on the total energy consumption by modeling energy-reliability tradeoffs at gate-level. For dynamic scaling, the gates are no longer homogeneous. We introduce a two phase algorithm in which the first phase is similar to ENCODED with homogeneous gates, but in the second phase, the voltage (and hence gate-energy) is tuned appropriately, which leads to orders of magnitude energy savings when compared with ``static''~voltage scaling (where the supply voltage is kept constant through the entire computation process). For example, when the required output bit error probability is $p_\text{tar}$, for polynomial decay of gate error probability $\epsilon$ with gate energy $E$ (i.e., $\epsilon=\frac{1}{E^c}$), the energy consumption per output bit is $\mathcal{O} \left( \frac{N}{K}\max\left\{L,\left(\frac{1}{p_\text{tar}}\right)^{\frac{1}{c}}\right\} \right)$ with dynamic voltage scaling, while it is $\Theta(\frac{NL}{K}(\frac{1}{p_\text{tar}})^{\frac{1}{c}})$ for the static case (we note that energy for ENCODED with static voltage scaling is still smaller than ``uncoded'' with static voltage scaling, {which is $\Omega(L(\frac{L}{p_\text{tar}})^{\frac{1}{c}})$}). Finally, in Section~\ref{nless_dec}, for deriving a lower bound as well as to connect with much of the existing literature, we allow the circuit to use noiseless gates for decoding. We derive (asymptotically) matching upper and lower bounds on required number of gates to attain a target error-probability.

\subsection{Related Work}
In spirit, our scheme is similar to von Neumann's  repetition-based construction~\cite{Neu_Aut_56} where an error-correction stage follows each computation stage to keep errors suppressed. Subsequent works~\cite{Dob_PPI_77,Pip_FOC_85,Pip_TIT_91} focus on minimizing the number of redundant gates while making error probability below a small constant. The difference from our work here is that these works do not allow (noiseless) precomputation based on the the knowledge of the required function, which our scheme (ENCODED) explicitly relies on. Therefore, our results are applicable when the same function needs to be computed multiple times for (possibly) different inputs, and thus the one-time cost of a precomputation is worth paying for. Thus, we do not include the preprocessing costs of the linear transformation matrix in the computational complexity calculation, and we assume all preprocessing can be done offline in a noise-free fashion\footnote{This difference in problem formulation is also why some of our achievable results on computational complexity might appear to beat the lower bounds of~\cite{Dob_PPI_77,Pip_FOC_85,Pip_TIT_91}.}.

We note that the algorithm introduced by Hadjicostis in~\cite{Had_TAC_03}, which is applied to finite-state linear systems, is similar to ours in that he also uses a matrix encoding scheme. However, \cite{Had_TAC_03} assumes that encoding and decoding procedures are noiseless, which we do not assume. In~\cite{laurenciu2016error}, Cucu Laurenciu, et al. designed a fault-tolerant computing scheme that embeds the encoding of an error control code into the logical functionality of the circuit. Decoding units are allowed to be noisy as well. However, the computation size is small, and theoretical guarantees are not provided. In~\cite[Theorem 4.4]{Pip_FOC_85}, Pippenger designed an algorithm to compute a binary linear transformation with noisy gates. The algorithm requires gates with fan-in $2^{23}$ and a gate-error probability of $35\cdot2^{-50}$. While the fan-in values are unrealistically high, the gate-error probability is also low enough that most practical computations can be executed correctly using ``uncoded'' strategies, possibly the reason why it has not received significant attention within circuits community. At a technical level, unlike the multi-stage computing scheme used in our work, Pippenger uses exhaustive enumeration of all linear combinations with length $\frac{1}{3}\log L$ for computing, where $L$ is the number of rows in the binary linear transformation. Lastly, we note here that Pippenger's scheme only works for the case when the number of columns $K$ in the binary linear transformation matrix and the code length $N$ of the utilized LDPC code satisfies $K=\Theta(N^3)$, while our algorithm works in a more practical scenario where $K=\Theta(N)$.

This work builds on our earlier work~\cite{Yang_All_14}, in which the problem of reliable communication with a noisy encoder is studied. In~\cite{Yang_All_14}, noisy decoders are embedded in the noisy encoder to repeatedly suppress errors. The noisy encoding problem is a special case of computing noisy linear transformation when the linear transformation matrix is the generator matrix of an error-correcting code. In~\cite{dupraz2016practical}, an augmented encoding approach was introduced to protect the encoder from hardware faults using extra parity bits. In~\cite{Hac_ISIT_13}, which considers a similar problem, errors are modelled as erasures on the encoding Tanner graph. 

Outside information theory, fault-tolerant linear transformations and related matrix operations have been studied extensively in algorithm-based fault tolerance~\cite{Huang_TC_84,Wang_TC_94,Ding_ISPA_11,Anf_TC_88,Chen_HPCN_09}. The main difference in our model is that faults happen at the circuit-level, e.g., in AND gates and XOR gates. Instead, in \cite{Huang_TC_84,Wang_TC_94,Ding_ISPA_11,Anf_TC_88,Chen_HPCN_09}, each functional block, e.g. a vector inner product, fails with a constant probability. If errors are considered at gate level, the error probability of a vector inner product will approach $1/2$~\cite{Had_TIT_05} as vector size grows, and one may not be able to use these schemes. Fault-detection algorithms on circuits and systems with unreliable computation units have also been studied extensively~\cite{Abda_JSSC_13,Radha_Asil_13,Bowman_ADAC_07,Choi_TSP_07,Had_CDC_01}. However, these algorithms assume that the detection units are reliable, which we do not assume. Moreover, using error control coding, we can combine the error detection and correction in the same processing unit.

\section{System Model and Problem Formulation}\label{System_Model}

\subsection{Circuit Model}\label{Model}

We first introduce unreliable gate models and circuit models that we will use in this paper. We consider two types of unreliable gates: probabilistic gates and defective gates.
\begin{definition}\label{ng1}(Gate Model I ($D,\epsilon$))
The gates in this model are probabilistically unreliable in that they compute a deterministic boolean function $g$ with additional noise $z_g$
\begin{equation}\label{noisy_gate}
  y=g(u_1,u_2,...,u_{d_g})\oplus z_g,
\end{equation}
where $d_g$ denotes the number of inputs and is bounded above by a constant $D>3$, $\oplus$ denotes the XOR-operation and $z_g$ is a boolean random variable which takes the value $1$ with probability \textcolor{black}{$\epsilon$ which is assumed to be smaller than $\frac{1}{2}$}. The event $z_g=1$ means the gate $g$ fails and flips the correct output. Furthermore, in this model, all gates fail independently of each other and the failure events during multiple uses of a single gate are also independent of each other. We allow different kinds of gates (e.g. XOR, majority, etc.) to fail with different probabilities. However, different gates of the same kind are assumed to fail with the same error probability\footnote{A weaker assumption is that different gates fail independently, but with different probabilities all smaller than $\epsilon$, which is called $\epsilon$-approximate~\cite{Pip_FOC_85}. The ENCODED technique also works for this model. Also note that our model is limited in the sense that the error probability $\epsilon$ does not depend on the gate input. This may not be realistic because the gate error probability can also depend on the input and even the previous gate outputs, which is also noted in \cite{nahlus2014energy}. However, the assumption that $\epsilon$ does not depend on the gate input can be relaxed by assuming that $\epsilon$ is the maximum error probability over all different input instances.}.
\end{definition}
This model is similar to the one studied in~\cite{Pip_TIT_91} and the failure event is often referred to as a transient fault. Our next model abstracts defective gates that suffer from permanent failures.
\begin{definition}\label{ng2}(Gate Model II ($D,n,\alpha$))
In a set of $n$ gates, each gate is either perfect or defective. A perfect gate always yields a correct output function
\begin{equation}\label{noisy_gate_2}
  y=g(u_1,u_2,...,u_{d_g}),
\end{equation}
where $d_g$ denotes the number of inputs and is bounded above by a constant $D> 3$. A defective gate outputs a deterministic boolean function of the correct output $\tilde{y}=f(g(\cdot))$. This function can be either $f(x)=\bar{x}$ (NOT function), $f(x)=0$ or $f(x)=1$. The fraction of defective gates in the set of $n$ gates is denoted by $\alpha$. We assume that measurement techniques cannot be used to distinguish between defective gates and perfect gates\footnote{Defective gates may result from component aging after being sold, and examining each gate in circuitry is in practice extremely hard. The storage units are easier to examine~\cite{Goor_DTC_93}, but replacing faulty memory cells requires replacing an entire row or column in the memory cell array~\cite{Kim_TC_98}.}.
\end{definition}
\begin{figure}
  \centering
  \includegraphics[scale=0.27]{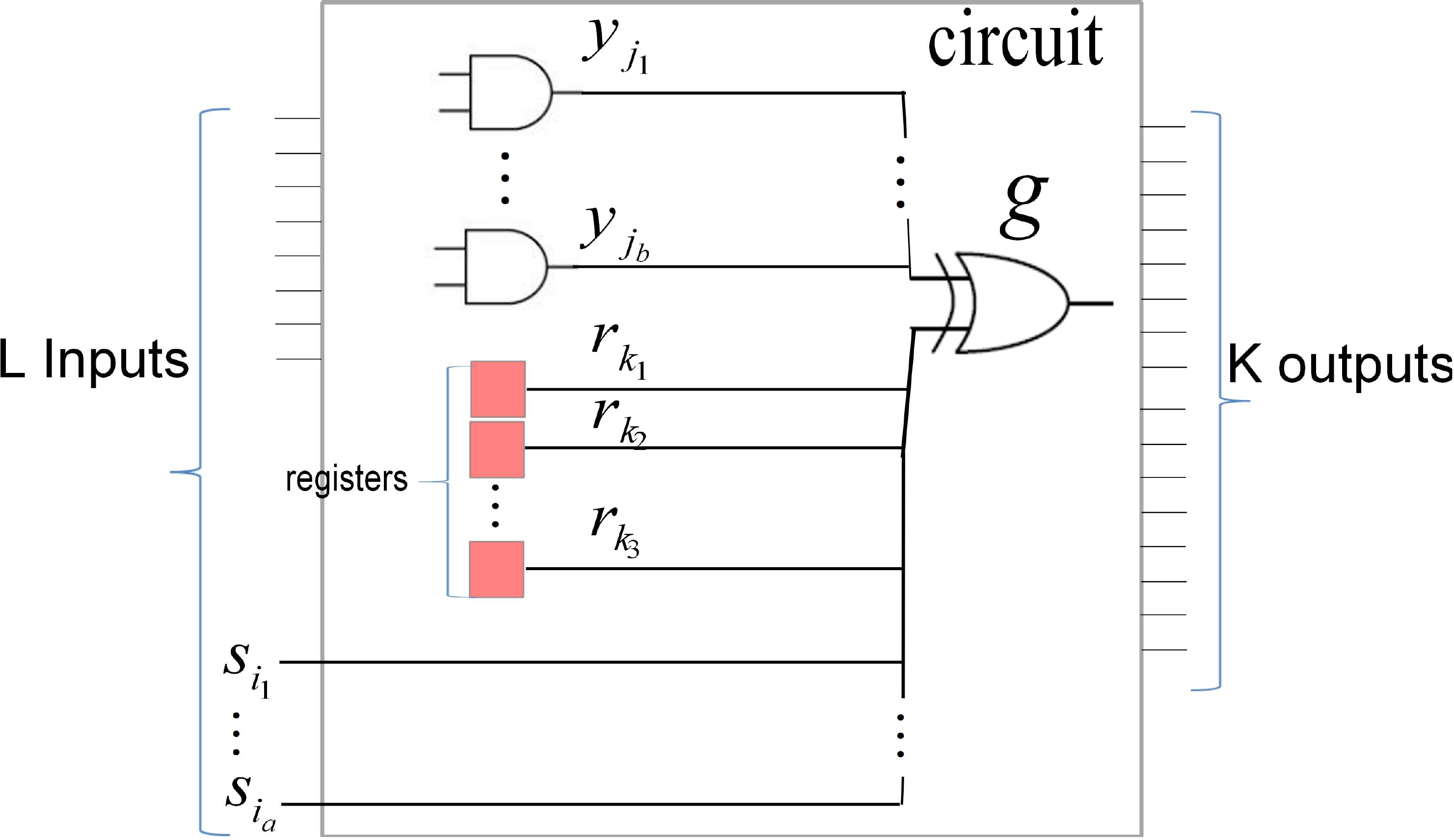}\\
  \caption{This figure shows an unreliable gate $g$ (Gate Model I or II) in a noisy circuit defined in Definition~\ref{NC}. A noisy circuit is constituted by many unreliable gates, which form the set $\mathcal{G}$.}\label{noisy_gate_fig}
\end{figure}
From the definition, a defective gate may repeatedly output the value 1 no matter what the input is, which is often referred to as a \textquotedblleft stuck-at error\textquotedblright. This might happen, for example, when a circuit wire gets shorted.

\begin{remark}
In fact, we can generalize all the results in this paper on the permanent error model (Gate Model II) to arbitrary but bounded error model, in which errors occur in a worst-case fashion but no more than a fixed fraction. The latter error model has been used in worst-case analyses in coding theory and arbitrarily varying channels~\cite{AhlswedeAVCs}. However, for consistency with the existing literature on error-prone decoding with permanent errors~\cite{Huang_TC_14,Chi_ITW_07}, we limit our exposition to these errors.
\end{remark}

The computation in a noisy circuit is assumed to proceed in discrete steps for which it is helpful to have circuits that have storage components.
\begin{definition}(Register) A register is an error-free storage unit that outputs the stored binary value. A register has one input. At the end of a time slot, the stored value in a register is changed to its input value if this register is chosen to be \emph{updated}.
\end{definition}

\begin{remark}
We assume that registers are noise-free only for clarity of exposition. It is relatively straightforward to incorporate in our analysis the case when registers fail probabilistically. A small increase in error probability of gates can absorb the error probability of registers. A similar change allows us to incorporate permanent errors in the registers as well.
\end{remark}

\begin{definition}\label{NC}(Noisy Circuit Model ($\mathcal{G},\mathcal{R}$))
A noisy circuit is a network of binary inputs $\textbf{s}=(s_1, s_2,...s_L)$, unreliable gates $\mathcal{G}=\{g_1,g_2,...,g_\mathscr{S}\}$ and registers $\mathcal{R}=\{r_1,r_2,...,r_\mathscr{T}\}$. Each unreliable gate $g\in \mathcal{G}$ can have inputs that are elements of \textbf{s}, or outputs of other gates, or from outputs of registers. That is, the inputs to an unreliable gate $g$ are $s_{i_1},\dots,s_{i_{a}},y_{j_1},\dots,y_{j_b},r_{k_1},\dots,r_{k_c}$, where $a+b+c=d_g$, the total number of inputs to this gate. Each register $r\in \mathcal{R}$ can have its single input from the circuit inputs $\textbf{s}$, outputs of unreliable gates or outputs of other registers. For simplicity, wires in a noisy circuit are assumed to be noiseless.
\end{definition}
\begin{definition}\label{NC_model}(Noisy Computation Model ($L,K,\mathscr{N}_\mathrm{comp}$))
A computing scheme $\mathcal{F}$ employs a noisy circuit to compute a set of binary outputs $\textbf{r}=(r_1, r_2,...r_K)$ according to a set of binary inputs $\textbf{s}=(s_1, s_2,...s_L)$ in multiple stages. At each stage, a subset of all unreliable gates $\mathcal{G}$ are \textit{activated} to perform a computation and a subset of all registers $\mathcal{R}$ are updated. At the completion of the final stage, the computation outputs are stored in a subset of $\mathcal{R}$. The number of \textit{activated} unreliable gates in the $t$-th stage is denoted by $\mathscr{N}_\mathrm{comp}^t$. Denote by $\mathscr{N}_\mathrm{comp}$ the total number of unreliable operations (one unreliable operation means one activation of a single unreliable gate) executed in the noisy computation scheme, which is obtained by
\begin{equation}
  \mathscr{N}_\mathrm{comp}=\mathop\sum_{t=1}^{T} \mathscr{N}_\mathrm{comp}^t,
\end{equation}
where $T$ is the total number of stages, which is predetermined.
\end{definition}
The noisy computation model is the same as a sequential circuit with a clock. The number of stages $T$ is the number of time slots that we use to compute the linear transform. In each time slot $t$, the circuit computes an intermediate function $f_t(x)$ using the computation units on the circuit, and the result $f_t(x)$ is stored in the registers for the computation in the next time slot $t+1$. The overall number of stages $T$ is predetermined (fixed before the computation starts).
\begin{remark}\label{feasible}
A computing scheme should be \emph{feasible}, that is, in each time slot, all the gates that provide inputs to an activated gate, or a register to be updated, should be activated.
\end{remark}
In this paper, we will only consider noisy circuits that are either composed entirely of probabilistic gates defined in Gate Model I or entirely of unreliable gates in Gate Model II. Note that if we consider probabilistic gates, the noisy circuit can be transformed into an equivalent circuit that does not have registers. This is because, since the probabilistic gate failures are (assumed to be) independent over operations, we can replicate each gate in the original circuit multiple times such that each gate in the equivalent circuit is only activated once. This circuit transformation is used in the proof of Theorem~\ref{low_bound_thm}.

\subsection{Problem Statement}\label{p_s}
The problem considered in this paper is that of computing a binary linear transformation $\mathbf{r}=\mathbf{s}\cdot\mathbf{A}$ using a noisy circuit, where the input vector $\textbf{s}=(s_1, s_2,...s_L)$, the output vector $\textbf{r}=(r_1, r_2,...r_K)$ and the $L$-by-$K$ (linear transformation) matrix $\mathbf{A}$ are all composed of binary entries. We consider the problem of designing a feasible (see Remark~\ref{feasible}) computing scheme $\mathcal{F}$ for computing $\mathbf{r}=\mathbf{s}\cdot\mathbf{A}$ with respect to Definition~\ref{NC_model}. Suppose the correct output is $\textbf{r}$. Denote by $\hat{\textbf{r}}=(\hat r_1, \hat r_2,...\hat r_K)$ the (random) output vector of the designed computing scheme $\mathcal{F}$. Note that the number of operations $\mathscr{N}_\mathrm{comp}$ has been defined in Definition~\ref{NC_model}. The computational complexity per bit $\mathscr{N}_\mathrm{per\text{-}bit}$ is defined as the total number of operations per output bit in the computing scheme. That is
\begin{equation}\label{communication_complexity}
  \mathscr{N}_\mathrm{per\text{-}bit}=\mathscr{N}_\mathrm{comp}/K.
\end{equation}

For gates from Gate Model I (Definition~\ref{ng1}), we are interested in the usual metrics of bit-error probability $P_e^{\text{bit}}=\frac{1}{K}\mathop\sum\limits_{k=1}^K\Pr ({{\hat{r}}_{k}}\ne {{r}_{k}})$ and block-error probability $P_e^{\text{blk}}=\Pr (\hat{\mathbf{r}}\ne \mathbf{r})$, averaged over uniformly distributed inputs $\mathbf{s}$ and noise realizations. In addition, in the spirit of ``excess distortion'' formulation in information theory~\cite{MartonExcessDistortion}, we are also interested in keeping the fraction of (output) errors bounded with high probability. This could be of interest, e.g., in approximate computing problems. To that end, we define another metric, $\delta_e^{\text{frac}}$, the ``bit-error fraction,'' which is simply the Hamming distortion between the computed output and the correct output (per output bit). That is, $\delta_e^{\text{frac}}=\max_{\mathbf{s}}\frac{1}{K}\mathop\sum\limits_{k=1}^K \mathbbm{1}_{\{{{\hat{r}}_{k}}\ne {{r}_{k}}\}}$, where $\mathbbm{1}_{\{\cdot\}}$ is the indicator function. The bit-error fraction depends on the noise, which is random in Gate Model I. Thus, we will constrain it probabilistically (see \textit{Problem~\ref{pro_2}}\footnote{\textcolor{black}{We will show that the bit-error fraction is constrained probabilistically (see \textit{Problem~\ref{pro_2}}) for all input vector $\mathbf{s}$.}}). The resulting problems are stated as follows:
\begin{problem}\label{pro_1}
\begin{equation}\label{op_problem_1}
\min_{\mathcal{F}}{\;\;\;}\mathscr{N}_\mathrm{per\text{-}bit},{\;}\text{s.t.}{\;}P_e<p_{\text{tar}},
\end{equation}
where $p_{\text{tar}}>0$ is the target bit error probability, and $P_e$ could be $P_e^{\text{bit}}$ or $P_e^{\text{blk}}$.
\end{problem}
\begin{problem}\label{pro_2}
\begin{equation}\label{op_problem_2}
\min_{\mathcal{F}}{\;\;\;}\mathscr{N}_\mathrm{per\text{-}bit},{\;}\text{s.t.}{\;}\Pr(\delta_e^{\text{frac}}<p_{\text{tar}})>1-\delta,
\end{equation}
where $p_{\text{tar}}>0$ is the target block error fraction and $\delta$ is a small constant.
\end{problem}

When we consider the Gate Model II (Definition~\ref{ng2}), since all gates are deterministic functions, we are interested in the worst-case fraction of errors $\delta_e^{\text{frac}}$. Thus, the optimization problem can be stated as follows:
\begin{problem}\label{pro_3}
\begin{equation}\label{op_problem_3}
\min_{\mathcal{F}}{\;\;\;}\mathscr{N}_\mathrm{per\text{-}bit},{\;}\text{s.t.}{\;} \max_{\mathbf{s},\mathcal{S}_\text{def}^i\;\text{s.t.}|\mathcal{S}_\text{def}^i|<\alpha_i n_{\mathcal{F},i},\forall i\in W }\delta_e^{\text{frac}}<p_{\text{tar}},
\end{equation}
where $\mathbf{s}$ is the input vector, $\mathcal{S}_\text{def}^i$ is the set of defective gates of type $i$, $W$ is the set of indices of different types of noisy gates (such as AND gates, XOR gates and majority gates), $\alpha_i$ is the error fraction of the gates of type $i$, $n_{\mathcal{F},i}$ is the total number of gates of type $i$ in the implementation of $\mathcal{F}$, and $p_{\text{tar}}>0$ is the target fraction of errors. {Note that $n_{\mathcal{F},i}$ is chosen by the designer as a part of choosing $\mathcal{F}$, while the error-fraction $\alpha_i$ is assumed to be known to the designer in advance.}
\end{problem}

Throughout this paper, we rely on the family of Bachmann-Landau notation~\cite{knuth} (i.e.~``big-O'' notation). For any two functions $f(x)$ and $g(x)$ defined on some subset of $\mathbb{R}$, asymptotically (as $x \rightarrow \infty$), $f(x) = \mathcal{O}(g(x))$
	if $|f(x)| \leq c_2 |g(x)|$; $f(x) = \Omega(g(x))$ if
	$|f(x)| \geq c_1 |g(x)|$; and $f(x) = \Theta(g(x))$
	if $c_3 |g(x)| \leq |f(x)| \leq c_4 |g(x)|$ for some
	positive real-valued constants $c_1, c_2, c_3, c_4$.

\subsection{Technical Preliminaries}\label{code_section}
First we state a lemma that we will use frequently.
\begin{lemma}[\cite{Gal_TIT_62}, pp. 41, Lemma 4.1]\label{idp_odd}
Suppose $X_i,i=1,\dots,L,$ are independent Bernoulli random variables and $\Pr(X_i=1)=p_i,\forall i$. Then
\begin{equation}
  \Pr(\sum_{i=1}^L X_i=1)=\frac{1}{2}\left[1-\mathop\prod\limits_{i=1}^L (1-2p_i)\right],
\end{equation}
where the summation is over $\mathbb{F}_2$, i.e., $1+1=0$.
\end{lemma}

We will use error control coding to facilitate the computation of the binary linear transformation. Here, we introduce some notations related to  the codes that we will use. We will use a regular LDPC code~\cite{Gal_TIT_62,Ric_TIT_01_2} with code length $N$, dimension $K$ and a $K\times N$ generator matrix $\mathbf{G}$ written as
\begin{equation}\label{encoding_matrix}
\mathbf{G}=\left[\begin{matrix}
\leftarrow&\mathbf{g}_1&\rightarrow\\
\leftarrow&\mathbf{g}_2&\rightarrow\\
\leftarrow&...&\rightarrow\\
\leftarrow&\mathbf{g}_K&\rightarrow
\end{matrix}\right].
\end{equation}
where each row $\mathbf{g}_k$ is a length-$N$ codeword. In the LDPC Tanner graph, denote the degree of a variable node $v$ by $d_v$ and the degree of a parity check node $c$ by $d_c$. The embedded decoders use either the Gallager-B decoding algorithm which is a $1$-bit hard-decision based decoding algorithm proposed in~\cite{Gal_TIT_62} and is included for completeness in Appendix~\ref{GB_details}, or the parallel bit flipping (PBF) algorithm, which is also a hard-decision algorithm proposed in \cite{Sip_FCS_96}. In particular, we use the modified parallel bit flipping algorithm defined in \cite{burshtein2008error}.
\begin{definition} The PBF algorithm is defined as follows
\begin{itemize}
  \item Flip each variable node that is connected to more than $\frac{d_v}{2}$ unsatisfied parity check nodes;
  \item Set the value of each variable node connected to exactly $d_v/2$ unsatisfied parity-check nodes to 0(or 1) with probability $1/2$;
  \item Update all parity check nodes;
  \item Repeat the first three steps for $c_e\log N$ times, where $c_e$ is a constant.
\end{itemize}
\end{definition}
The PBF algorithm can be used to correct a constant fraction of errors after $\Theta(\log N)$ decoding iterations when the computing components in the decoder are noiseless and the error fraction is small enough. However, since we will consider noisy decoders, we will build on a more refined result, which concerns a single decoding iteration of the algorithm (see the following requirement (A.3) and Lemma~\ref{expander_fraction}).

In our main results, we may require the utilized LDPC code to satisfy some of (not all) the following conditions.
\begin{itemize}
  \item \textbf{(A.1) Degree Bound}:  \textcolor{black}{The variable node degree $d_v$ and the parity check node degree $d_c$ are both less than or equal to $D$}, so that each majority or XOR-operation \textcolor{black}{(in the Gallager-B decoding algorithm)} can be carried out by a single unreliable gate. Moreover, we assume that the variable node degree $d_v\ge 4, \forall v$.
  \item \textbf{(A.2) Large Girth}: The girth $l_g=\Theta(\log N)$. An LDPC code with the following girth lower bound is obtained in~\cite{Gal_TIT_62,Len_TIT_01}:
    \begin{equation}\label{girth_lb}
         l_g>\frac{2\log N}{\log((d_v-1)(d_c-1))}-2c_g,
    \end{equation}
    where $c_g=1-\frac{\log\frac{d_cd_v-d_c-d_v}{2d_c}}{\log((d_v-1)(d_c-1))}$ is a constant that does not depend on $N$.
  \item \textbf{(A.3) Worst-case Error Correcting}:One iteration of the PBF algorithm using a noiseless decoder can bring down the number of errors in the codeword from $\alpha_0N$ to $(1-\theta)\alpha_0N$ for two constants $\alpha_0,\theta\in(0,1)$, for any possible patterns of $\alpha_0N$ errors.
\end{itemize}
The requirement in (A.2) can be met by using codes introduced in~\cite{Len_TIT_01} or using the PEG construction proposed in \cite{ref_3}. The requirement in (A.3) can be met either by using $(d_v,d_c)$-regular random code ensembles and using the analysis in \cite{burshtein2008error}, or by using regular Expanders \cite{Sip_FCS_96}. \textcolor{black}{In particular, in Appendix~\ref{worst_case}, we show that almost all codes in the $(9,18)$-regular code ensemble of sufficiently large length $N$ can reduce the number of errors by $\theta=15\%$ after one iteration of the PBF algorithm, if the fraction of errors is upper-bounded by $\alpha_0\le 5.1\cdot 10^{-4}$. We also show that at least $4.86\%$ of the $(9,18)$-regular codes of length $N=50,000$ can reduce the number of errors by $\theta=15\%$ after one iteration of the PBF algorithm, if the number of errors satisfies $\alpha_0 N\le 20$, which is equivalent to $\alpha_0\le 0.0004$.}


\section{ENCODED: \textbf{En}coded \textbf{Co}mputation with \textbf{De}coders Embedde\textbf{D}}\label{main_results}
In this section, we present the main scheme that we use for noisy computation of linear transformations. We call this scheme ``ENCODED'' (\textbf{En}coded \textbf{Co}mputation with \textbf{De}coders Embedde\textbf{D}). We aim to provide an overview of ENCODED in this section. Then, this scheme will be modified to a tree-structured scheme, ENCODED-T, in Section~\ref{Construction} and further modified to a bit-flipping-based technique ENCODED-F in Section~\ref{expander_encoder}.

\subsection{ENCODED: A Multi-stage Error-Resilient Computation Scheme}\label{pipeline_str}
Instead of computing a binary linear transformation $\mathbf{r}=\mathbf{s}\cdot\mathbf{A}$ without using any redundancy, we will compute
\begin{equation}\label{coded_computing}
  \mathbf{x}=\mathbf{r}\cdot \mathbf{G}=\mathbf{s}\cdot\mathbf{A}\mathbf{G},
\end{equation}
where $\mathbf{G}=[\mathbf{I},\mathbf{P}]=[\mathbf{g}_1; \mathbf{g}_2;...; \mathbf{g}_K]$ is the $K\times N$ generator matrix of the chosen systematic LDPC code. The matrix product $\mathbf{A}\mathbf{G}$ is assumed to be computed offline in a noise-free fashion.
\begin{figure}
  \centering
  \includegraphics[scale=0.26]{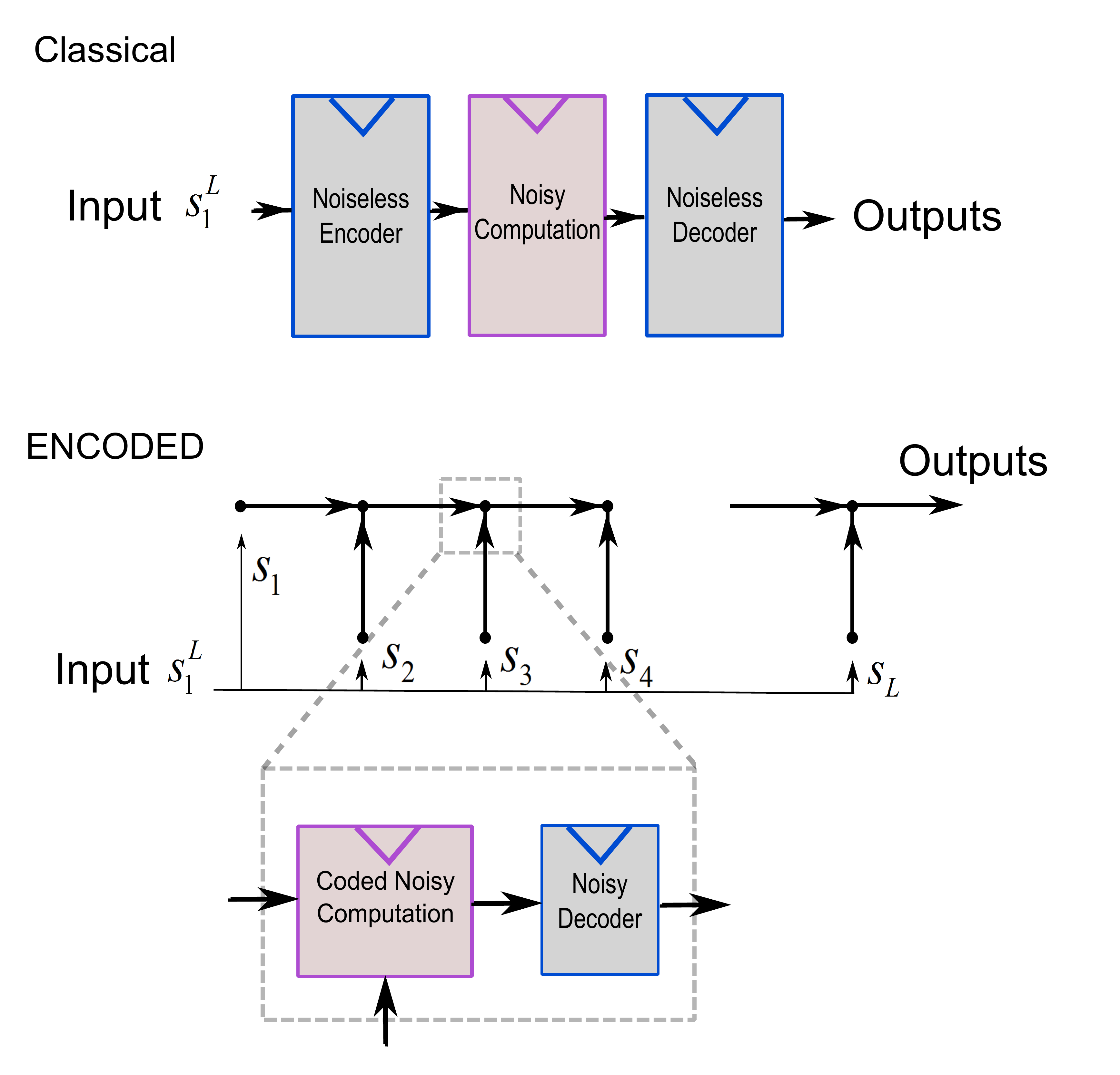}\\
  \caption{An illustration of the conceptual difference between classical noisy computing schemes and the ENCODED technique.}\label{coded_computing_fig}
\end{figure}
An important observation is that since all rows in the matrix product $\mathbf{A}\mathbf{G}$ are linear combinations of the rows in the generator matrix $\mathbf{G}$, the rows of $\mathbf{A}\mathbf{G}$ are codewords as well. That is,
\begin{equation}\label{quiv_encoding_matrix}
\tilde{\mathbf{G}}=\mathbf{A}\mathbf{G}=\left[\begin{matrix}
\leftarrow&\tilde{\mathbf{g}}_1&\rightarrow\\
\leftarrow&\tilde{\mathbf{g}}_2&\rightarrow\\
\leftarrow&...&\rightarrow\\
\leftarrow&\tilde{\mathbf{g}}_L&\rightarrow
\end{matrix}\right]
\end{equation}
where each row $\tilde{\mathbf{g}}_l,l=1,\dots,L$ is a codeword. Then, if the computation were noiseless, the correct computation result $\mathbf{r}=\mathbf{s}\cdot\mathbf{A}$ could be obtained from the combined result
\begin{equation}\label{combine_result_1}
\begin{split}
  \mathbf{x}=[\mathbf{r},\mathbf{r}\cdot\mathbf{P}]=\mathbf{r}\cdot \mathbf{G}.
\end{split}
\end{equation}
Since $\mathbf{r}\cdot \mathbf{G}=\mathbf{s}\cdot\mathbf{A}\mathbf{G}=\mathbf{s}\cdot\tilde{\mathbf{G}}$,
\begin{equation}\label{combine_result}
\begin{split}
  \mathbf{x}=\mathbf{s}\cdot\tilde{\mathbf{G}}=\mathop\sum\limits_{l=1}^L s_l \tilde{\mathbf{g}}_l.
\end{split}
\end{equation}
In the following sections, we will explain how error control coding can be used to reliably compute $\mathbf{x}=\mathop\sum\limits_{l=1}^L s_l \tilde{\mathbf{g}}_l$. The basic idea is as follows: we break the computation into $L$ stages, so that the noiseless intermediate result after the $l$-th stage would be $\mathbf{x}^{(l)}=\mathop\sum\limits_{j=1}^l s_j \tilde{\mathbf{g}}_j$. When gates are noise-free, $\mathbf{x}^{(l)}$ is a codeword. When gates are noisy, during the $l$-th stage, we first compute $\mathbf{x}^{(l-1)}+s_l \tilde{\mathbf{g}}_l$ using noisy AND gates (binary multiplication) and noisy XOR gates (binary addition) and then correct errors (with high probability) using an LDPC decoder or an expander decoder to get $\mathbf{x}^{(l)}$. During the entire computing process, AND gates and XOR gates introduce errors, while the noisy decoders suppress errors. Finally, it will be proved in Theorem~\ref{Main_thm} and Theorem~\ref{expander_probabilistic} that error probability is maintained below a small constant. We summarize the ENCODED technique in Algorithm~\ref{alg_0}.
\textcolor{black}{
\begin{algorithm}\caption{ENCODED (\textbf{En}coded \textbf{Co}mputation with \textbf{De}coders Embedde\textbf{D})}\label{alg_0}
\textbf{INPUT}: A binary vector $\textbf{s}=(s_1, s_2,...s_L)$.\\
\textbf{OUTPUT}: A binary vector $\textbf{x}=(x_1, x_2,...x_N)$.\\
\textbf{INITIALIZE} \\
Compute $\tilde{\mathbf{G}}=\mathbf{A}\mathbf{G}=[\tilde{\mathbf{g}}_1; \tilde{\mathbf{g}}_2;...; \tilde{\mathbf{g}}_L]$. Store an all-zero vector $\mathbf{x}^{(0)}$ in an $N$-bit register.\\
\textbf{FOR} $l$ from $1$ to $L$
\begin{itemize}
  \item Use $N$ unreliable AND gates to multiply $s_l$ with $\tilde{\mathbf{g}}_l$, the $l$-th row in $\tilde{\mathbf{G}}$, add this result to $\mathbf{x}^{(l-1)}$ using $N$ unreliable XOR gates, and store the result in the $N$-bit register.\footnote.
  \item Use an unreliable decoder to correct errors and get $\mathbf{x}^{(l)}$.
\end{itemize}
\textbf{END}\\
Output $\mathbf{x}^{(L)}$ as the output $\mathbf{x}$.
\end{algorithm}}
Compared to many classical results~\cite{Simon_ITW_11,Spi_FCS_96,Rac_ISIT_08,Huang_TC_84} on applying error control coding to noisy computing, instead of computing after encoding, the proposed scheme combines encoding and computing into a joint module (see Fig.~\ref{coded_computing_fig}). Because there is no separation between computing and encoding, in some sense, we \emph{encode the computation}, rather than encoding the message. We briefly discuss the intuition underlying the ENCODED technique in Section~\ref{Intuition}. We note that, we change the \textbf{FOR}-loop in Alg.~\ref{alg_0} to a tree-structure (ENCODED-T) in Section~\ref{Construction} in order to reduce error accumulation as explained in Remark~\ref{tree_depth} in Section~\ref{Construction}.
\footnotetext[7]{These operations are assumed to be performed noiselessly, as discussed earlier.}

\subsection{Intuition underlying the Embedded Decoders}\label{Intuition}
The basic idea of our proposed computing scheme is to split the computation into a multistage computation of $\mathbf{x}=\mathop\sum\limits_{l=1}^L s_l \tilde{\mathbf{g}}_l$, and use embedded decoders inside the noisy circuit to repeatedly suppress errors as the computation proceeds. Since the noisy circuit can only be constructed using unreliable gates, the embedded decoders are also constituted by unreliable gates.

Why is such a multistage computation helpful? For instance, if ``uncoded''~matrix multiplication $\mathbf{r}=\mathbf{s}\mathbf{A}$ is carried out, each output bit is computed using an inner product, and $\mathcal{O}(L)$ unreliable AND and XOR-operations are required. Without repeated suppression, each output bit is erroneous with probability $\frac{1}{2}$ as $L\to \infty$. Intermediate and repeated error suppression alleviates this  \textit{error accumulation} problem. Can one use a feedback structure, as is used in Turbo codes~\cite{Ber_ICC_93} and ARA codes~\cite{Abb_TCom_07} (these codes often have a feedback structure~\cite{Cost_PHE_04} for encoding), to keep errors suppressed, instead of the LDPC codes used here? A feedback structure can be detrimental since errors persist in the feedback loop and propagate to the future, which can make the final bit error probability large. This observation motivated us to use LDPC codes.

Also note that due to the \textquoteleft last-gate\textquoteright~effect in noisy circuits, error probability cannot approach zero. Thus, our goal is not to eliminate errors, but to suppress them so that the error probability (or the error fraction) is kept bounded below a target value that depends on the error probability of the last gate.

\section{Main Results on Computing a Binary Linear Transformation with Embedded Decoders}\label{Encoder_Sec}
\textcolor{black}{In this section, we show that a linear transformation can be computed \textquoteleft reliably\textquoteright~(in accordance with the goals of Problems \ref{pro_1}-\ref{pro_3} in Section~\ref{p_s}) even in presence of noise, using error control coding. We provide three results, one each for formulations in Problem~\ref{pro_1} to Problem~\ref{pro_3}. These results are obtained using two variants of ENCODED, which we call ENCODED-T and ENCODED-F. ENCODED-T uses Gallager-B decoding algorithm for error suppression, while ENCODED-F uses the PBF algorithm. The implementation details of ENCODED-T and ENCODED-F are respectively provided in Section~\ref{Construction} and Section~\ref{expander_encoder}. We also compare resource requirements of this coding-based computation with repetition-based computation using simulations (in Section \ref{simulation_sec}.)}

\subsection{ENCODED-T: A Scheme for Reliable Computation of Linear Transformations under Gate Model I}\label{Construction}
\begin{figure*}
  \centering
  \includegraphics[scale=0.5]{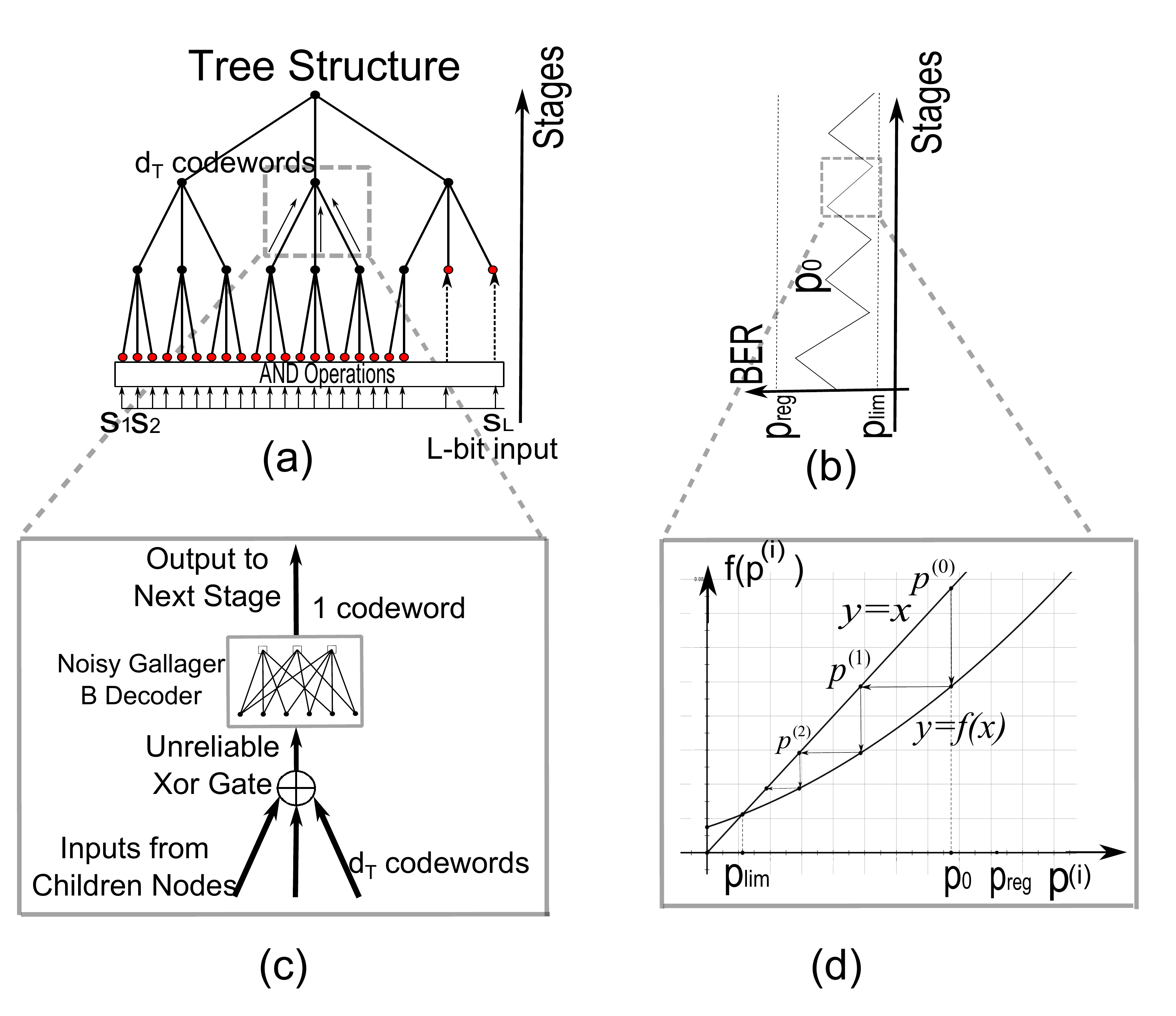}\\
  \caption{(a) shows the tree structure of the noisy computing scheme. During the computing process, the bit error probability is bounded between two constants $p_{\text{reg}}$ and $p_{\text{lim}}$ shown in (b). (c) shows a compute-and-correct structure. The bit error probability evolution in one embedded decoder is shown in (d).}\label{Encoder_tree}
\end{figure*}
The utilized unreliable gates in the computing scheme are AND gates, XOR gates and majority gates that are defined in Gate Model I, with error probabilities $p_\text{and}$, $p_\text{xor}$ and $p_\text{maj}$ respectively. We change the \textbf{FOR}-loop of ENCODED in Section~\ref{pipeline_str} slightly and use a $D$-branch tree with depth $M$ instead. We call this computing scheme ``ENCODED-T'' (ENCODED scheme with a Tree structure) and is conceptually illustrated in Fig.~\ref{Encoder_tree}(a). We use this tree structure because it reduces the number of stages of the computing scheme from $L$ in the \textbf{FOR}-loop in Alg.~\ref{alg_0} to $\Theta(\log L)$. This reduces the extent of information mixing caused by message-passing decoding (in comparison with ENCODED's sequential structure), which introduces correlation among messages and makes the density evolution analysis difficult. This issue will be detailed in Remark~\ref{tree_depth}.

The message $\textbf{s}=(s_1,...,s_L)$ is input from the leaf nodes. The output $\textbf{x}=\textbf{s}\cdot\tilde{\mathbf{G}}=(x_1,...,x_N)$ is computed from bottom to top and finally obtained at the root. Note that the tree structure is not necessarily a complete tree. Specifically, the tree is complete from the first level to the $(M-1)$-th level, i.e., the level just above the bottom level, and the number of nodes in the bottom level is chosen such that the number of leaf nodes is $L$. An illustration of a non-complete $3$-branch tree with $L=22$ leaf nodes (which are colored red) is shown in Fig.~\ref{Encoder_tree}(a). From Fig.~\ref{Encoder_tree}(a), we see that by removing the leaf children-nodes of each complete non-leaf node (a non-leaf node with $d_T$ leaf children-nodes), we effectively reduce the number of leaf-nodes by $(d_T-1)$, because the $d_T$ removed leaf children-nodes (red nodes) is replaced by one non-leaf node that will turn into a leaf node. Therefore, it is easy to see that the total number of non-leaf nodes is $\left\lceil\frac{L-1}{d_T-1}\right\rceil$.

Each of the $L$ leaf nodes has one of the $L$ rows of the matrix $\tilde{\mathbf{G}}$, i.e., $\tilde{\mathbf{g}}_1$ to $\tilde{\mathbf{g}}_L$ stored in it. At the start of the computing process, the $l$-th node of the first $L$ nodes calculates $s_l\cdot\tilde{\mathbf{g}}_l$ using $N$ unreliable AND gates and stores it as an intermediate result in a register. In the upper levels, each non-leaf node performs a component-wise XOR-operation of the $d_T$ intermediate results from $d_T$ children. Observe that if no gate errors occur, the root gets the the binary sum of all $s_l\cdot\mathbf{g}_l,i=1,\dots,L$, which is the correct codeword $\textbf{x}=\textbf{s}\cdot\tilde{\mathbf{G}}$.
\begin{algorithm}\caption{ENCODED-T}\label{alg1}
\textbf{INPUT}: A binary vector $\textbf{s}=(s_1, s_2,...s_L)$.\\
\textbf{OUTPUT}: A binary vector $\textbf{x}=(x_1, x_2,...x_N)$.\\
\textbf{NOTATION}: Denote by $v_m^k$ the $k$-th node in the $m$-th level of the tree structure. Denote the stored intermediate result in node $v$ by $\mathbf{y}_v$ and $v_m^l$ by $\mathbf{y}_m^l$. Denote the $d_T$ children-nodes of $v_m^l$ by $\mathcal{D}(v_m^l)$.\\
\textbf{INITIALIZE} \\
\textcolor{black}{For $1\le l\le L$, use noisy AND gates to create $d_v$ copies of $s_l\cdot\mathbf{g}_l$ and store them as an $E$-bit vector $\tilde{\mathbf{y}}_{M}^l$ in the register of $v_M^l$, or in the ($M-1$)-th level with the appropriate index. All of these $E$-bit vectors are stored as the first layer of intermediate results.}\\
\textbf{FOR} $m$ from $M-1$ to 1
\begin{itemize}
  \item Each non-leaf node $v_m^k$ calculates the XOR of the outputs of its $d_T$ (or less if the node is not complete) children-nodes and writes the result in its own $E$-bit register (computations could be noisy):
\begin{equation}\label{enct}
\tilde{\mathbf{y}}_{m}^k=\bigoplus_{v\in \mathcal{D}(v_{m}^k)} \tilde{\mathbf{y}}_v,
\end{equation}
  \item Each node $v_m^k$ performs one iteration of the message-passing decoding.
\end{itemize}
\textbf{END}\\
\textcolor{black}{Change the $E$-bit vector $\tilde{\mathbf{y}}_{1}^1$ back to the $N$-bit codeword $\mathbf{y}_{1}^1$ by randomly selecting one copy. Output $\mathbf{y}_{1}^1$ as the output $\mathbf{y}$}.
\end{algorithm}

In order to deal with errors caused by unreliable gates, each non-leaf tree node is a compute-and-correct unit shown in Fig.~\ref{Encoder_tree}(c). Unreliable XOR gates are used to perform the component-wise XOR-operation of the intermediate results. A noisy Gallager-B decoder \textcolor{black}{(see Appendix~\ref{GB_details}) is used to correct errors in the associated register after the XOR-operation. Note that the number of bits transmitted from variable nodes to parity check nodes during each Gallager-B decoding iteration is $E=d_vN$, where $d_v$ is the variable node degree of the Tanner graph and $E$ is the total number of edges. Therefore, at a tree node, the register stores $E=d_vN$ bits as intermediate results instead of $N$ bits as in Algorithm~\ref{alg_0}}\footnote{We have to store $E$ bits instead of an $N$-bit codeword, because we need the i.i.d. assumption of messages in the density evolution analysis. Note that by storing these $E$ bits, the corresponding $N$-bit codeword can be retrieved either using a noisy majority vote or a random selection.}. These $E$ bits of messages can be viewed as $d_v$ copies of the corresponding $N$-bit codeword, with indices from 1 to $d_v$. (The technique of replicating $N$ bits into $E=d_v N$ bits was first used in \cite{Tay_Bel_68,Kuz_PIT_73} for the correction circuit in fault-tolerant storage, which is known as the Taylor-Kuznetsov memory. A related technique was used in~\cite[Pg. 216]{Had_TIT_05} to construct fault-tolerant linear systems. See \cite{Vas_ITW_06} for details.) The XOR-operations at the beginning of the next stage are also done on codeword copies with the same index.

Before sending the output to the parent-node, each node performs $C$ iterations of the message-passing decoding \textcolor{black}{on the $E$-bit intermediate result obtained by the XOR-operations} with the embedded decoder. \textcolor{black}{Note that there are $C$ iterations of decoding at each non-leaf node of the tree structure shown in Fig.~\ref{Encoder_tree}. We will use the index $i=1,\dots,C,$ to number the decoding iterations done at a single tree node, which is different from the index $m$ used to number the level in the tree structure.} However, we will show that it suffices to use $C=1$ iteration (see Lemma~\ref{lemma7}) to bring the \textcolor{black}{bit error probability of the $E$-bit intermediate result} back to $p_\text{maj}+\frac{1}{d_T}p_\text{thr}$. In the noisy decoder, the bit error probability of the intermediate result, assuming the decoding neighborhood is cycle-free for $i+1$ iterations \textcolor{black}{(in fact, the number of levels of the decoding neighborhood grows by $2C$ at each tree node, see Remark~\ref{tree_depth} for details)}, follows the density evolution $P_e^{(i+1)}=f(P_e^{(i)})$ and the explicit expression of function $f(\cdot)$ is given in Lemma~\ref{irregular_DE} in Appendix~\ref{Upper_Bound_Analysis}. This evolution is illustrated in Fig.~\ref{Encoder_tree}(d). The bit error probability follows the directed path shown in Fig.~\ref{Encoder_tree}(d), and asymptotically approaches the fixed point of the density evolution function as number of iterations increases if decoding neighborhoods continue to satisfy the cycle-free assumption (which no fixed good code does). \textcolor{black}{However, the expression of $f(\cdot)$ is complicated, so we provide an upper bound $f(P_e^{(i)})<f_0(P_e^{(i)})$ in Lemma~\ref{Irr_DE_ge_lmm} for further analysis.}

During the computing process, the XOR-operations introduce errors, while the Gallager-B decoding process suppresses them. The bit error probability \textcolor{black}{of the intermediate result is reduced repeatedly at different stages} and is ensured to be bounded within the interval $(p_\text{maj},p_\text{reg})$ (as illustrated in Fig.~\ref{Encoder_tree}(b)), where the parameter $p_\text{reg}$ is defined as
\begin{equation}\label{p_reg}
  p_{\text{reg}}= {{p}_{\text{xor}}}+(d_T+1)p_\text{thr}.
\end{equation}

\begin{figure*}
  \centering
  \includegraphics[scale=0.38]{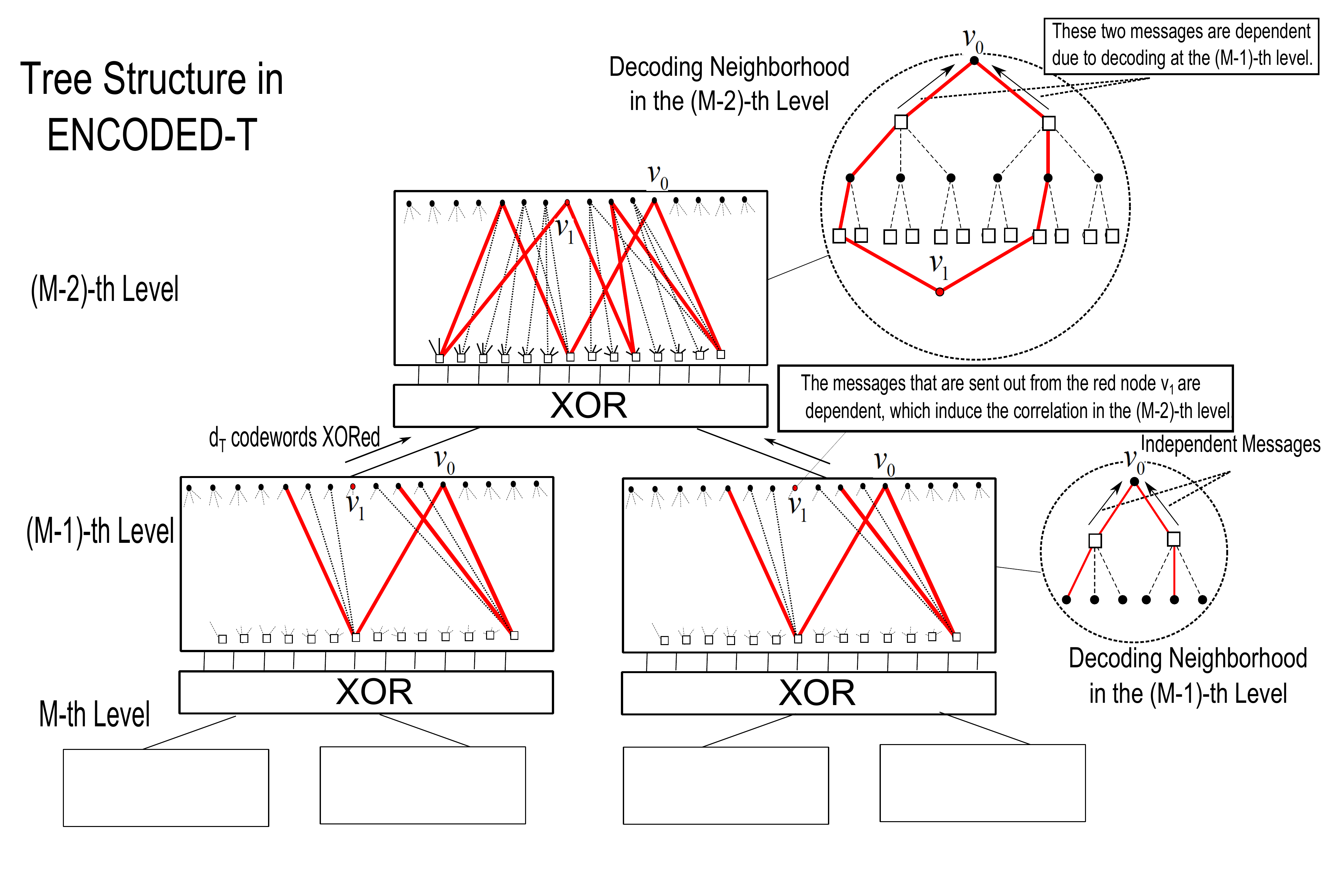}\\
  \caption{This figure shows the decoding neighborhoods in two adjacent levels of the tree structure in ENCODED-T. The black nodes in the decoding neighborhood denote variable nodes, while the white ones denote parity check nodes. The decoding neighborhood in the ($M-2$)-th level is not cycle-free due to a short cycle of length $8$. Note that the tree structure in ENCODED-T and the decoding neighborhoods are two completely different things and should not be confused.}\label{Two_Trees}
\end{figure*}
\begin{remark}\label{tree_depth}
An astute reader might notice that as the computation process proceeds, the message-passing algorithm introduces increasing correlation in messages as they are passed up the tree. This introduces a difficulty in analyzing the error probability using density-evolution techniques, because density evolution relies on the cycle-free assumption on the decoding neighborhoods\footnote{In fact, the analysis of message-passing type algorithms implemented on general graphs with cycles are the key question to many real-world problems, and convergence results are only known in some limited cases \cite{du2013network,Had_TIT_05}.}. \textcolor{black}{As shown in Fig.~\ref{Two_Trees}, the decoding neighborhoods of the same level in the ENCODED-T tree have the same neighborhood structure (because all nodes use the same LDPC code and the number of decoding iterations is a constant across the same level), while the decoding neighborhood on the upper level, the ($M-2$)-th level (see Fig.~\ref{Two_Trees}), grows by a depth of $2C$ compared to the adjacent lower level, the ($M-1$)-th level. Note that the XOR-operations before the LDPC decoding iterations do not change the nodes in the decoding neighborhood (again, because the same code is used at all nodes). Due to small girth of the LDPC code chosen in the example in Fig.~\ref{Two_Trees} (girth = $8$), the decoding neighborhood is not cycle free, and hence the messages sent to a variable node can be correlated. In Fig.~\ref{Two_Trees}, the correlation between messages at the root of the decoding neighborhood is caused by the correlation between messages sent by the red-colored node $\nu_1$.} We circumvent this issue\footnote{Another possible way could be to prove convergence to cycle-free neighborhoods through randomized constructions, as is done in~\cite{Var_TIT_11,Ric_TIT_01_2}. } by choosing a code of girth large enough so that no correlations are introduced even over multiple levels of the tree. The tree structure of ENCODED-T reduces the number of levels exponentially (in comparison with ENCODED), thereby reducing the girth requirement.

To understand this, as in~\cite{Ric_TIT_01_2,Len_TIT_01}, we define the \textit{number of independent iterations} of an LDPC code as the maximum number of iterations of message-passing decoding such that the messages sent to any variable node or parity check node are independent (i.e., number of iterations after which the code reaches its girth). We also use the phrase `\textit{overall} number of decoding iterations' in ENCODED-T to denote the sum (over levels) of maximum number of iterations over all nodes at each level of the ENCODED-T tree. For instance, if $C$ decoding iterations are executed at each tree node, then the ``overall'' number of iterations is $CM$, where $M$ is the number of levels in the tree. Intuitively, this quantifies the extent of information mixing in the Tanner graph by the message-passing algorithm. We need to ensure that the overall number of decoding iterations is smaller than the number of \textit{independent iterations} of the LDPC code used. For a fixed $C$, the tree structure of ENCODED-T requires the number of independent iterations, and hence also the  code-girth, to be $\Theta(\log N/\log d_T)$.

The importance of the ENCODED-tree structure also becomes clear: the \textbf{FOR}-loop in ENCODED (Alg.~\ref{alg_0}) has exponentially larger number of levels, and requires a large girth of $\Theta(L)$ for $L$ stages of the algorithm to maintain independence of messages in a decoding neighborhood at the root node.

The number of independent iterations of message-passing decoding of an LDPC code can be made as large as $\Theta(\frac{\log N}{\log (d_v-1)(d_c-1)})$ (see assumption (A.2)). Thus, for $d_T$ and $N$ large enough, the overall decoding neighborhoods are cycle free, and density evolution analysis is valid. More details are provided in Appendix~\ref{Upper_Bound_Analysis}.

In Section~\ref{expander_encoder}, we provide another way to handle this correlation issue that uses codes designed for correcting worst-case errors. For these codes, correlations do not matter as long as the number of errors does not exceed a certain threshold.
\end{remark}

\subsection{Analysis of ENCODED-T}\label{sketch_proof}
\vspace{0.2cm}
In what follows, we prove Theorem~\ref{Main_thm}, which characterizes the performance of Algorithm~\ref{alg1}.

\begin{theorem}[Error Suppression Using Gallager-B decoding for Problem 1]\label{Main_thm}
Using unreliable AND gates, majority gates and \textcolor{black}{XOR} gates from Gate Model I with respective gate error probabilities $p_\text{and}$, $p_\text{xor}$ and $p_\text{maj}$ that satisfy\footnote{In this result, we obtain different conditions on the error probabilities of different types of gates as sufficiency conditions, instead of a uniform bound on all gates. We do so to offer maximum flexibility to the circuit designer.}
\begin{equation}\label{p0satis}
\begin{split}
{{p}_{\text{thr}}}&:= {{\left( \begin{matrix}
   {{d}_{v}}-1  \\
   \lfloor \frac{{{d}_{v}}-1}{2}\rfloor   \\
\end{matrix} \right)}^{-\frac{1}{d-1}}}{{d}_{c}}^{-\frac{d}{d-1}}d_T^{-\frac{1}{d-1}}{{(d_T+1)}^{-\frac{d}{d-1}}},\\
{{p}_{\text{maj}}}&\le p_\text{thr},\\
{{p}_{\text{xor}}}&\le \frac{d_T+1}{{{d}_{c}}}{{p}_{\text{thr}}},\\
{{p}_{\text{and}}}&\le\frac{d_T+1}{d_T}{{p}_{\text{thr}}}.
\end{split}
\end{equation}
where $d=\left\lceil \frac{{{d}_{v}}-1}{2}\right\rceil$, and the tree width $d_T$ satisfies
\begin{equation}\label{D_satisfy}
  \left\lceil\frac{\log L}{\log d_T}\right\rceil\le \frac{\log N}{2\log(d_v-1)(d_c-1)},
\end{equation}
the binary linear transformation $\mathbf{r}=\mathbf{s}\cdot\mathbf{A}$ can be computed with output bit error probability $P_e^{\text{bit}}\le p_{\text{maj}}+p_\text{thr}\cdot\frac{1}{d_T}$ using the ENCODED technique provided in Alg.~\ref{alg1} (see Section~\ref{Construction}) that uses an LDPC code which satisfies assumptions (A.1) and (A.2). Further, the number of operations per bit satisfies
\begin{equation}\label{complexity_a}
  \mathscr{N}_\mathrm{per\text{-}bit}\le \frac{3E}{K}\cdot{\left\lceil\frac{L-1}{d_T-1}\right\rceil}+LE/K=\Theta\left(\frac{LN}{K}\right).
\end{equation}
where $E$ is the number of edges in the Tanner graph, $K$ is the number of outputs and $N$ is the code length of the utilized LDPC code. \end{theorem}

\begin{remark}
Note that $P_e^\text{bit}$ remains bounded even as $L$ is increased. Thus, we can compute linear transforms with large size and still obtain good error performance.
\end{remark}

\begin{proof}[Proof of Theorem~\ref{Main_thm}]
We construct the computation tree of the ENCODED-T technique (see Section~\ref{Construction} and Section~\ref{sketch_proof}) for computing $\mathbf{x}=\mathbf{s}\cdot\mathbf{A}\mathbf{G}=\mathbf{s}\cdot\tilde{\mathbf{G}}$ as described in Section~\ref{Construction}. In all, there are $M$ levels, and in the $m$-th level the number of nodes is chosen such that the overall number of leaf-nodes is $L$. Each leaf node has the codeword $\tilde{\mathbf{g}}_l$ stored in it at the beginning of computation, where $\tilde{\mathbf{g}}_l$ is the $l$-th row of $\tilde{\mathbf{G}}$, an $L$-by-$N$ matrix (see~\eqref{quiv_encoding_matrix}). Note that $M$ satisfies $d_T^{M-2}<L\le d_T^{M-1}$, or
\begin{equation}\label{treebigger}
M=\left\lceil\frac{\log L}{\log d_T}+1\right\rceil.
\end{equation}
To ensure that the number of leaf nodes is $L$, the number of non-leaf nodes $S_\text{nl}$ is
\begin{equation}\label{treesize}
S_\text{nl}=\left\lceil\frac{L-1}{d_T-1}\right\rceil.
\end{equation}
\subsubsection{Error probability analysis}
Define $p_0$ to be the bit error probability at the start of the first round of the Gallager-B decoding carried out in the $(M-1)$-th level of the computing tree in Fig~\ref{Encoder_tree}(b). Since at this time, each bit is computed by XORing $d_T$ bits (see Fig~\ref{Encoder_tree}(a)) where each bit is calculated from an AND-operation (see Fig~\ref{Encoder_tree}(b)), we know from Lemma~\ref{idp_odd} that
\begin{equation}
\begin{split}
  p_0=&\frac{1}{2}\left[1-(1-2p_{\text{xor}})(1-2p_{\text{and}})^{d_T}\right]\\
  \overset{(a)}{<}&\frac{1}{2}\left[1-(1-2p_{\text{xor}})(1-2d_Tp_{\text{and}})\right]
  <p_{\text{xor}}+d_T p_{\text{and}},
\end{split}
\end{equation}
where step (a) is from the inequality $(1-x)^{d_T}>1-d_Tx$, $x\in (0,1)$. Using definition of $p_{\text{reg}}$ in~\eqref{p_reg} and condition~\eqref{p0satis}
\begin{equation}\label{p_zero2}
  p_0<p_{\text{xor}}+(d_T+1) p_{\text{thr}}=p_{\text{reg}}.
\end{equation}
In the following, we will prove that as long as the error probabilities of noisy gates satisfy~\eqref{p0satis}, the noisy Gallager-B decoder can make the bit error probability fall below $p_\text{maj}+\frac{1}{d_T}p_\text{thr}$ after one iteration of decoding.
\begin{lemma}\label{lemma7}
Suppose all noise parameters satisfy~\eqref{p0satis}. Then, as long as the bit error probability before decoding is smaller than $p_\text{reg}$ defined in~\eqref{p_reg}, after one iteration of decoding, the bit error probability $p_e^\text{dec}$ satisfies
\begin{equation}\label{Converge}
  p_\text{maj}<p_e^\text{dec}\le p_\text{maj}+\frac{1}{d_T}p_\text{thr}.
\end{equation}
\end{lemma}
\begin{proof}
See Appendix~\ref{Upper_Bound_Analysis}.
\end{proof}
\begin{remark}\label{remark41}Lemma~\ref{lemma7} describes the behavior shown in Fig.~\ref{Encoder_tree} (d). For noiseless LDPC decoders~\cite{Len_TIT_01}\cite{Ric_TIT_01}\cite{Ric_TIT_01_2}, the fixed point of the density evolution function is zero when the code is operated for a channel that has small enough noise. However, for noisy density evolution, it can be shown that the fixed point \textcolor{black}{$p_{\text{lim}}=p_{\text{maj}}+\mathcal{O}(p_{\text{maj}}^2+p_{\text{xor}}^2)\approx p_{\text{maj}}$}~\cite[Thm 5]{Yaz_TC_01}.
\end{remark}
Therefore, after one iteration of Gallager-B decoding, the bit error probability reduces from $p_0$ to something less than $p_\text{maj}+\frac{1}{d_T}p_{\text{thr}}$. Then, the corrected codeword is sent to the parent-node in the next level of the tree structure. After that, the XOR-operation of $d_T$ codewords from children-nodes are carried out. From Lemma~\ref{idp_odd}, we know that the error probability after this XOR-operation is upper bounded by
\begin{equation}
\begin{split}
  &\frac{1}{2}\left[1-(1-2p_{\text{xor}})\left[1-2(p_\text{maj}+\frac{1}{d_T}p_\text{thr})\right]^{d_T}\right]\\
  <&p_{\text{xor}}+d_Tp_\text{maj}+p_\text{thr} \overset{(a)}{<}p_{\text{xor}}+(d_T+1)p_\text{thr} \overset{(b)}= p_\text{reg},
\end{split}
\end{equation}
where $(a)$ follows from~\eqref{p0satis} and $(b)$ follows from~\eqref{p_reg}. Therefore, we can reuse the result in Lemma~\ref{lemma7} and carry out a new Gallager-B decoding iteration. To summarize, the bit error probability starts at $p_0$ and then oscillates between $p_{\text{reg}}$ and $p_{\text{maj}}$ during the entire computation process. This behavior is qualitatively illustrated in Fig.~\ref{Encoder_tree}(b) and numerically illustrated through simulation results in Fig.~\ref{simulation_figure}.
\subsubsection{Computational complexity analysis}
Each compute-and-correct unit is constituted by $E$ $D$-fan-in noisy XOR gates, an $E$-bit register and a Gallager-B decoder, where $E$ is the number of edges in the LDPC bipartite graph.

The operations required in one iteration of decoding are $E$ XOR-operations and $E$ majority-operations, because on each edge that connects a variable node $v$ and a parity-check node $c$, two messages $m_{c\to v}^{(i)}$ and $m_{v\to c}^{(i)}$ are calculated respectively using an XOR-operation and a majority-operation. Thus, the number of operations to carry out during one iteration of decoding is $2E$. Since the number of non-leaf tree nodes is $S_\text{nl}=\left\lceil\frac{L-1}{d_T-1}\right\rceil$ (see~\eqref{treesize}), the number of operations required in all non-leaf nodes is at most $(2E+E)\left\lceil\frac{L-1}{d_T-1}\right\rceil=\mathcal{O}(EL)$. The operations required in the first layer (input-layer) of the computing tree are $LE$ AND-operations. The computing process outputs $K$ bits, so the number of operations required per bit is $\frac{3E}{K}\cdot{\left\lceil\frac{L-1}{d_T-1}\right\rceil}+LE/K=\mathcal{O}(LE/K)$. Since the number of edges $E=d_vN=\Theta(N)$, we know that the number of operations per bit is in the order of $\mathcal{O}(LN/K)$.
\end{proof}
\subsection{ENCODED-F: A scheme for Reliable Computation of Linear Transformations under Gate Model II}\label{expander_encoder}
In this Section, we consider unreliable gates defined in Definition~\ref{ng2}, which are either perfect or defective. We construct a computing scheme using a decoding scheme different from that of ENCODED-T that is still able to attain a small error fraction in the final output. This computing scheme operates in exactly the same manner as ENCODED (see Alg.~\ref{alg_0}). However, the embedded noisy decoder is a PBF decoder. This computation scheme is named ENCODED-F (ENCODED using the flipping algorithm). We modify ENCODED as follows: we partition the entire computing process into $\left\lceil\frac{L}{d_s-1}\right\rceil$ stages, where $d_s$ is called the \emph{group size}. First, we store an all-zero codeword in the $N$-bit register. In the $l$-th stage, we first use $(d_s-1)N$ AND gates to obtain the $d_s-1$ scalar-vector multiplications $s_i\cdot \tilde{\mathbf{g}}_i$ for $i\in\{(d_s-1)(l-1)+1,(d_s-1)(l-1)+2,\dots,(d_s-1)l\}$, where $\tilde{\mathbf{g}}_i$ is the $i$-th row of the combined matrix $\tilde{\mathbf{G}}=\mathbf{A}\mathbf{G}$. Then, we use $N$ XOR gates to add the $d_s-1$ results to the $N$-bit register. The parameter $d_s$ is chosen so that $d_s\le D$, the maximum input to each noisy gate. After that, we use one iteration of the PBF algorithm (see Section~\ref{main_results}) to correct errors. We use $P$ XOR gates and $N$ majority gates in one iteration of the PBF algorithm.

We note here that the tree structure in ENCODED-T (see Alg.~\ref{alg1}) could be used in the ENCODED-F. However, we still use the \textbf{FOR}-loop structure in Alg.~\ref{alg_0}, because the resulting maximum tolerable gate error probability is smaller using the \textbf{FOR}-loop structure. This is for two reasons (i) The tree structure in ENCODED-T was motivated by induced correlations in messages as they are passed up. However, correlations do not matter when decoding using the PBF algorithm; (ii) Interestingly, there is a benefit to using the \textbf{FOR}-loop structure: in ENCODED-T, the error probability increases by a factor of $d_T$ at every level {due to XOR-ing of $d_T$ inputs}. Since these XOR operations are not needed in ENCODED-F, to keep error probability suppressed, the constraints on the gate error probability will thus be more severe in ENCODED-T than in ENCODED-F.

In what follows, we prove Theorem~\ref{expander_thm}, which quantifies the error-tolerance of ENCODED-F. The basic tool used to prove Theorem~\ref{expander_thm} is a modified version of the worst-case error correcting result in the requirement (A.3), which provides the worst-case error correcting capability of regular LDPC codes using one iteration of noisy PBF decoding.

\begin{theorem}[Error Suppression Using the PBF algorithm for Problem 3]\label{expander_thm}
Using unreliable AND gates, XOR gates and majority gates from Gate Model II ($D,n,\alpha$) with respective error fractions $\alpha_{\text{and}}$, $\alpha_{\text{xor}}$ and $\alpha_{\text{maj}}$ respectively, and using an $(d_v,d_c)$-regular LDPC code which satisfies (A.3) to implement ENCODED-F with group size $d_s$, as long as
\begin{align}\label{expander_condition1}
  (d_s-1)\alpha_{\text{and}}+[D(1-R)+1]\alpha_{\text{xor}}+\alpha_{\text{maj}}<\theta\alpha_0,
\end{align}
the binary linear transformation $\mathbf{r}=\mathbf{s}\mathbf{A}$ can be computed using $N$ AND gates, $(N + P)$ XOR gates, and $N$ majority gates, and the number of operations per bit is at most $\frac{2N+P}{K}\left\lceil \frac{L}{d_s-1}\right\rceil+\frac{NL}{K}=\Theta(\frac{LN}{K})$. Further, the error fraction of the final output is at most $\alpha_0$.
\end{theorem}

\begin{proof}[Proof of Theorem~\ref{expander_thm}]
We use induction on the stage index $l$ to derive an upper bound on the number of errors. In the first stage, $N(d_s-1)$ AND gates and $N$ XOR gates introduce at most $N[(d_s-1)\alpha_{\text{and}}+\alpha_{\text{xor}}]$ errors, which is upper bounded by
\begin{equation}
  (d_s-1)\alpha_{\text{and}}+\alpha_{\text{xor}}<\theta\alpha_0\label{expander_condition2},
\end{equation}
which can be obtained by combining~\eqref{expander_condition1}. Suppose in the $(l-1)$-th stage, after adding a set of ($d_s-1$) codewords $s_{i}\cdot\tilde{\mathbf{g}}_{i},i\in \{(d_s-1)(l-2)+1,(d_s-1)(l-2)+2,\dots,(d_s-1)(l-1)\}$ to the $N$-bit register, the number of errors is strictly less than $N\alpha_0$.

Then, according to condition (A.3), if no computation errors occur during execution of one iteration of PBF algorithm, the fraction of errors can be reduced to $N\alpha_0(1-\theta)$. Whenever an XOR gate flips the corresponding parity check value during the PBF algorithm, it affects at most $d_c$ majority gates. In total, there are $P$ XOR gates used in one iteration of the PBF algorithm, so there are at most $\alpha_\text{xor} Pd_c$ errors due to XOR errors in the PBF algorithm. There are at most $\alpha_\text{maj}N$ errors due to majority gate failures. After this iteration of bit flipping, another set of ($d_s-1$) codewords $s_i\mathbf{g}_{i}$ is added to the $N$-bit register with $N(d_s-1)$ AND gates and $N$ XOR gates. These two operations introduce $(\alpha_\text{xor}+\alpha_\text{and}(d_s-1))N$ errors. Therefore, the total error fraction before the next PBF algorithm is upper bounded (using the union bound) by
\begin{equation}
  \begin{split}\label{after_bit_flip_1}
 \alpha_\text{PBF}\le N\alpha_0(1-\theta)
 +&[d_c(1-R)+1]\alpha_{\text{xor}} \\
 +&\alpha_{\text{maj}}+ (d_s-1)\alpha_{\text{and}},
\end{split}
\end{equation}
where $R$ is the code rate ($R=\frac{N-P}{N}$).
As long as~\eqref{expander_condition1} holds and $d_c\le D$, before the next bit flipping, it holds that
\begin{equation}\label{before_bit_flip_2}
\begin{split}
  \alpha_\text{PBF}\le N\alpha_0(1-\theta)+\theta N\alpha_0=\alpha_0 N.
\end{split}
\end{equation}
Therefore, the induction can proceed.

In each stage, we need $N+P$ XOR-operations and $N$ majority-operations. During the entire computation, we need $NL$ AND-operations. Therefore, the computational complexity per output bit, which is the total number of operations in $\left\lceil L/(d_s-1)\right\rceil$ stages divided by $K$ bits, is $(2N+P)\left\lceil L/(d_s-1)\right\rceil/K+\frac{NL}{K}$.
\end{proof}

ENCODED-F can be applied to Gate Model I as well, which is characterized in the following theorem.

\begin{theorem}[Error Suppression Using PBF algorithm for Problem 2]\label{expander_probabilistic}
Using unreliable AND gates, majority gates and \textcolor{black}{XOR} gates from Gate Model I with respective gate error probabilities $p_\text{and}$, $p_\text{xor}$ and $p_\text{maj}$, and using an $(d_v,d_c)$-regular LDPC code which satisfies (A.3) to implement ENCODED-F with group size $d_s$, as long as
\begin{equation}
\begin{split}
   &\max\{p_\text{and},p_\text{xor},p_\text{maj}\}\\
   &<\lambda:=\frac{\theta\alpha_0/2}{(d_s-1)+\left[d_c(1-R)+1 \right]+1}, \label{expander_prob_condition}
\end{split}
\end{equation}
the binary linear transformation $\mathbf{r}=\mathbf{s}\cdot\mathbf{A}$ can be computed using $\frac{2N+P}{K}\left\lceil \frac{L}{d_s-1}\right\rceil+\frac{NL}{K}=\Theta(\frac{LN}{K})$ operations per bit. Further, the final error fraction $\delta_e^{\text{frac}}$ satisfies
\begin{equation}
  \Pr(\delta_e^{\text{frac}}<\alpha_0)>1-P_{e}^\text{blk},
\end{equation}
where the probability $P_{e}^\text{blk}$ satisfies
\begin{equation}\label{expander_error_prob}
  P_{e}^\text{blk}<3L\exp\left( -\lambda^*N \right),
\end{equation}
where
\begin{align}\label{expander_error_exp}
   &{{\lambda }^{*}}=D(2\lambda\|\lambda)=\left( 2\log 2-1 \right)\lambda+\mathcal{O}({{\lambda}^{2}}).
\end{align}
\end{theorem}

However, in probabilistic settings, the number of errors at any stage could exceed $N\alpha_0$. In what follows, we use large deviation analysis to show that the probability of exceeding $N\alpha_0$ is small. First, we review the large deviation result for binomial distribution~\cite[page 502, Example 3]{Che_AMS_52}.
\begin{lemma}\label{Binomial_deviation}
Let $X_i,i=1,\dots,N$ be $N$ i.i.d. binary random variables with $\Pr [{{X}_{i}}=1]=p$. Then
\begin{equation}\label{lg_d}
  \Pr \left[\frac{1}{N}\sum\limits_{i=1}^{N}{{{X}_{i}}}>(p+\lambda )\right]<\exp \left[ -D(p+\lambda \|p)N \right],
\end{equation}
where $D(p+\lambda \|p)=(p+\lambda )\log_e \frac{p+\lambda }{p}+(1-p-\lambda )\log_e \frac{1-p-\lambda }{1-p}$. Further, if $p<\lambda$,
\begin{equation}\label{lg_d2}
  \Pr \left[\frac{1}{N}\sum\limits_{i=1}^{N}{{{X}_{i}}}>(p+\lambda )\right]<\exp \left[ -D(2\lambda \|\lambda)N\right].
\end{equation}
\end{lemma}
\begin{proof}
The inequality~\eqref{lg_d} is the large deviation bound for binomial distribution and is presented in~\cite[page 502, Example 3]{Che_AMS_52}. Note that $D(p+\lambda\|p)$ is monotone non-increasing for $p\in(0,\lambda)$. When $p<\lambda$, we have $D(p+\lambda\|p)>D(2\lambda\|\lambda)$. Therefore, \eqref{lg_d} holds for $p<\lambda$.
\end{proof}
Then, Theorem~\ref{expander_probabilistic} follows from Theorem~\ref{expander_thm} and Lemma~\ref{Binomial_deviation}.

\begin{proof}[Proof of Theorem~\ref{expander_probabilistic}]
Using Theorem~\ref{expander_thm}, we know that if the error fraction in all stages is bounded by the inequality~\eqref{expander_condition1}, the final error fraction is at most $N\alpha_0$.

From Lemma~\ref{Binomial_deviation}, we know that
\begin{equation}
\begin{split}
   \Pr (\alpha_\text{and}>p_\text{and}+\lambda)&<\exp \left[ -D(p_\text{and}+\lambda\|p_\text{and})N \right] \\ &<\exp[-D(2\lambda\|\lambda)N],
\end{split}
\end{equation}
\begin{equation}
\begin{split}
  \Pr (\alpha_\text{xor}>p_\text{xor}+\lambda)&<\exp \left[ -D(p_\text{xor}+\lambda\|p_\text{xor})N \right]\\
  &<\exp[-D(2\lambda\|\lambda)N],
\end{split}
\end{equation}
\begin{equation}
\begin{split}
 \Pr (\alpha_\text{maj}>p_\text{maj}+\lambda)&<\exp \left[ -D(p_\text{maj}+\lambda\|p_\text{maj})N \right]\\
 &<\exp[-D(2\lambda\|\lambda)N].
\end{split}
\end{equation}
Setting $\lambda=\frac{\theta\alpha_0/2}{(d_s-1)+[d_c(1-R)+1]+1}$ as in the condition~\eqref{expander_prob_condition}, we have
\[\begin{split}
&(d_s-1){{p}_\text{and}}+\left[ d_c(1-R)+1 \right]{{p}_\text{xor}}+{{p}_\text{maj}}\\
&<(d_s-1)\lambda+\left[ d_c(1-R)+1 \right]\lambda+\lambda=\frac{\theta\alpha_0}{2}.\end{split}\]
Therefore,
\[\begin{split}
   &\Pr  \left( (d_s-1){\alpha_\text{and}}+\left[ d_c(1-R)+1 \right]{\alpha_\text{xor}}+{\alpha_\text{maj}}>\theta\alpha_0 \right) \\
  <&\Pr  \left( (d_s-1){\alpha_\text{and}}+\left[ d_c(1-R)+1 \right]{\alpha_\text{xor}}+{\alpha_\text{maj}}\right.\\
  &\left.\quad\quad>(d_s-1){{p}_\text{and}}+\left[ d_c(1-R)+1 \right]{{p}_\text{xor}}+{{p}_\text{maj}}+\frac{\theta\alpha_0}{2}\right) \\
  <&\Pr  ( (d_s-1){\alpha_\text{and}}>(d_s-1){{p}_\text{and}}+(d_s-1)\lambda)\\
  &+\Pr(\left[ d_c(1-R)+1 \right]{\alpha_\text{xor}}
  >\left[ d_c(1-R)+1 \right]{{p}_\text{xor}}\\
  &\qquad+\left[ d_c(1-R)+1 \right]\lambda)\\
  &+\Pr({\alpha_\text{maj}}>{p}_\text{maj}+\lambda)\\
  =&\Pr  \left(\alpha_\text{and}>p_\text{and} +\lambda\right) + \Pr \left( \alpha_\text{xor} >p_\text{xor}+\lambda\right)\\
  &+ \Pr \left( \alpha_\text{maj}>p_\text{maj}+\lambda\right)\\
 <&3\exp(-ND(2\lambda\|\lambda)),
\end{split}\]
where
\begin{equation}
\begin{split}
  & D(2\lambda\|\lambda)=2\lambda\log 2+(1-2\lambda)\log \frac{1-2\lambda}{1-\lambda} \\
 & =2\lambda\log 2-(1-2\lambda)\log \left( 1+\frac{\lambda}{1-2\lambda} \right) \\
 & =2\lambda\log 2-(1-2\lambda)\left( \frac{\lambda}{1-2\lambda}+\mathcal{O}({{\lambda^2}}) \right)\\
 & =\left( 2\log 2-1 \right)\lambda+\mathcal{O}({{\lambda^2}}).
\end{split}
\end{equation}
Since $(d_s-1){\alpha_\text{and}}+{\alpha_\text{maj}}<(d_s-1){\alpha_\text{and}}+\left[ d_c(1-R)+1 \right]{\alpha_\text{xor}}+{\alpha_\text{maj}}$, we also have
\[\Pr \left( (d_s-1){\alpha_\text{and}}+{\alpha_\text{maj}}> \theta\alpha_0\right)<3\exp \left( -D(2\lambda\|\lambda) N \right).\]
Therefore, using the union bound for the $L$ stages, the total error probability is upper bounded by
$P_{e}^\text{blk}<3L\exp\left( -D(2\lambda\|\lambda)N \right)$.
\end{proof}

\begin{remark}
The analysis of the PBF algorithm (see Appendix~\ref{worst_case}) still requires randomized code constructions. Another method to analyze the bit flipping algorithm is to use Expander codes (also see Appendix~\ref{worst_case}). However, hardware-friendly expander codes tend to be hard to construct and use in practice, while many hardware-friendly LDPC codes have been designed\cite{Kou_TIT_01,Mou_SPM_04,Tho_INP_03,Fos_TIT_04}. In fact, the ENCODED-T and ENCODED-F both have some advantages over the other, so none of them universally outperforms the other. On the one hand, ENCODED-T works for all regular LDPC codes, but it requires a tree-structure for its theoretical analysis to work. On the other hand, ENCODED-F does not require a tree-structure, and requires less redundancy than ENCODED-T (it does not need to maintain $d_v$ copies of the computation), but it requires the LDPC code to satisfy certain properties. More concretely, it requires that each single iteration of the simple Bit Flipping algorithm corrects a constant fraction of errors (say, 75\%) for any combination of less than some constant number of errors (say, 20 errors). This property is hard to verify in practice and only existence result can be obtained. This further makes it hard to say which one is better than the other. This is the reason that we keep both the density evolution analysis for general LDPC codes in Section~\ref{sketch_proof} under the assumption of large girth, and the PBF analysis for ENCODED-F.
\end{remark}

\textcolor{black}{The following converse result holds for all computation schemes. Although this converse result does not match with any of the achievable results listed above, it matches with an achievable result when a ``noiseless decoder'' is available (details will be provided in Section~\ref{nless_dec}) in the scaling of the target error probability $p_\text{tar}$. Thus, we believe the converse result captures the required computational complexity for the beginning stages of the linear transform computation.}

\begin{theorem}[Converse result]\label{low_bound_thm}
For Gate Model I with error probability $\epsilon$, maximum fan-in $D$, and linear transformation $\mathbf{r}=\mathbf{s}\cdot\mathbf{A}$ \textcolor{black}{with $\mathbf{A}$ having full row rank}, in order to achieve $P_e^{\text{blk}}$ smaller than $p_{\text{tar}}$, the number of operations required per bit is lower bounded as $\mathscr{N}_\mathrm{per\text{-}bit}\ge\frac{L\log 1/p_{\text{tar}}}{KD\log D/\epsilon}=\Omega(\frac{L\log 1/p_{\text{tar}}}{K\log 1/\epsilon})$.
\end{theorem}
\begin{proof}
See Appendix~\ref{Lower_Bound}.
\end{proof}

\subsection{Simulation Results for ENCODED Techniques}\label{simulation_sec}
We present simulation results in Fig.~\ref{simulation_figure}, which shows the variation of bit error probability during the process of implementing Algorithm~\ref{alg1}. In the simulation, we generate random binary $\mathbf{A}$ matrices where each entry takes value one with probability $1/2$. The x-axis is the computing steps from the bottom to the top of the noisy computing tree structure. The y-axis is the bit error probability. As we have mentioned, during the entire computing process, computation introduces errors and decoding suppresses errors. Thus, the bit error probability oscillates between two limits. This is exactly the expected behavior as shown in Fig.~\ref{Encoder_tree}(b).

This simulation uses a randomly generated $(6,12)$-regular LDPC code of length 1200. The systematic generator matrix $\mathbf{G}$ is computed by solving the equation $\mathbf{G}\mathbf{H}^\top=\mathbf{0}$ in the binary field using Gaussian elimination, where $\mathbf{H}$ is the parity check matrix. The tree in Algorithm~\ref{alg1} is set to be a two-branch tree, i.e., $d_T=2$. The failure probability values of different unreliable gates are set to be the same as the threshold value computed using the condition \eqref{p0satis} in Theorem 1: $p_\text{maj}=0.001$, $p_\text{xor}=0.00026$ and $p_\text{and}=0.002$. We still assume that each operation in the decoding process is done by a single unreliable gate, and all gates fail independently of each other. Notice that the error probability lower limit is just above $p_\text{maj}=0.001$, which is consistent with our analysis in Section~\ref{sketch_proof}. The bit error probability after each decoding iteration should be confined between $p_\text{maj}$ and $p_\text{maj}+\frac{1}{d_T}p_\text{thr}$ in theory (see Lemma~\ref{lemma7}).

It is interesting to note that the computation scheme works at fairly practical values of node degrees ($d_v=6$) and blocklengths ($N=1200$). The target error probabilities are typically much smaller, but so are gate-error probabilities. \textcolor{black}{Moreover, the scheme works even though the choice of the tree-width $d_T$ do \textit{not} satisfy the constraints~\eqref{D_satisfy}.} These suggest that the bounds in Theorem~\ref{Main_thm} are conservative. The moderate blocklength of the code suggests that the scheme could be applied in practice, but a deeper investigation is needed which is beyond the scope of this work.

We also use simulations to compare ENCODED and repetition-based schemes. In particular, we provide a comparison between ENCODED-F and a particular repetition-based scheme called ``distributed voting scheme'' \cite{Had_TIT_05}, that is designed for $p_\text{maj}>0$. This method repeats not only the computation part, but also the majority voting part of the repetition-based circuit. The illustration of the distributed voting scheme is shown in Fig.~\ref{Distributed_majority}. In this way, we can compare the (repetition-coding based) distributed voting scheme with ENCODED that both use noisy gates.

\begin{figure}
  \centering
  \includegraphics[scale=0.55]{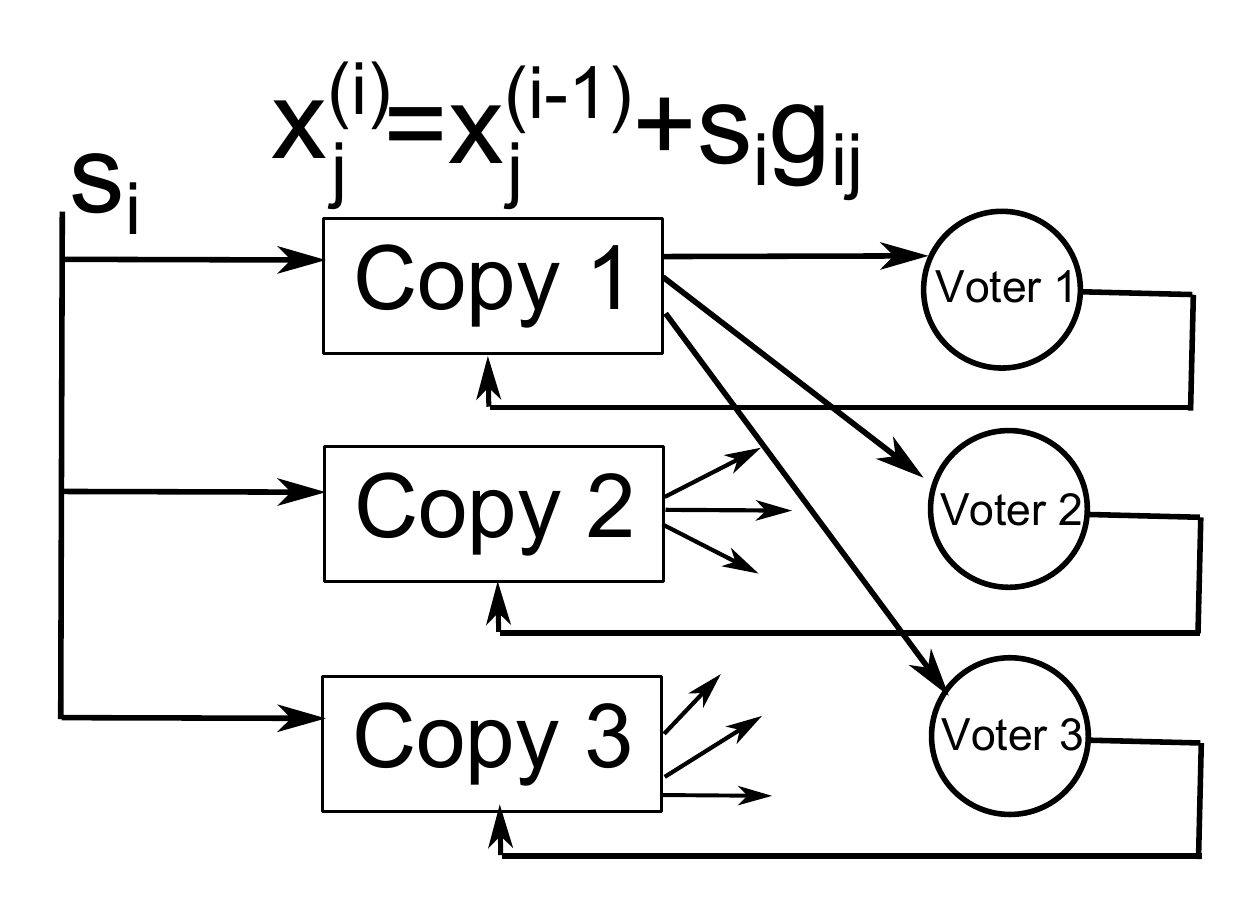}\\
  \caption{This is the illustration of the 3-time distributed voting scheme for computing an inner product $\mathbf{s}\cdot\mathbf{a}_j=s_1a_{1j}+s_2a_{2j}+\dots s_La_{Lj}$, where $\mathbf{s}$ is the input to the linear transform $\mathbf{s}\mathbf{A}$, and $\mathbf{a}_j$ is the $j$-th column of $\mathbf{A}$. The computation is divided into $L$ stages. In the $i$-th stage, the distributed voting scheme computes $x_j^{(i)}=x_j^{(i-1)}+s_ig_{ij}$ for three times using three sets of AND-gates and XOR-gates, uses three noisy majority-gates (which are called voters in \cite{Had_TIT_05}) to compute three copies of the majority votes. Then, the output of each majority value is sent to the corresponding copy for the computation in the next stage.}\label{Distributed_majority}
\end{figure}

The performance comparison is shown in Fig.~\ref{NLT_Comparison_nontree}. In the distributed majority scheme, we use three-time repetition or four-time repetition. For ENCODED-F, we set $d_v=4$, $d_c=8$, $d_s=8$, $K=2000$, $L=2100$, $N=4000$. We set $p_\text{and}=0.000125$, $p_\text{maj}=0.0005$ and $p_\text{xor}=0.001$. We set these error parameters because we assume that the error probability of each gate is proportional to its fan-in number (we use 2-input AND-gates, 4-input MAJ-gates and 8-input XOR gates). Note that the number of compute-and-correct stages in ENCODED-F should be $\left\lceil\frac{L}{d_s-1}\right\rceil=300$. In one compute-and-correct stage, we need $N$ XOR-operations of fan-in $d_s=8$ for binary addition, $P$ XOR-operations of fan-in $d_c=8$ for parity computation and $N$ MAJ-operations of fan-in $d_v=4$ for majority computation. In all $300$ stages, we also need $NL$ AND-operations of fan-in $2$. Therefore the number of operations per output bit for ENCODED-F is
\begin{equation}
  \mathscr{N}^{\text{ENC}}_\mathrm{XOR-8, per\text{-}bit}=\frac{N+P}{K}\left\lceil \frac{L}{d_s-1}\right\rceil
  =\frac{3}{7}L,
\end{equation}
\begin{equation}
  \mathscr{N}^{\text{ENC}}_\mathrm{MAJ-4,per\text{-}bit}=\frac{N}{K}\left\lceil \frac{L}{d_s-1}\right\rceil
  =\frac{2}{7}L,
\end{equation}
\begin{equation}
  \mathscr{N}^{\text{ENC}}_\mathrm{AND-2,per\text{-}bit}=\frac{NL}{K}=2L.
\end{equation}
In the distributed majority voting scheme with repetition time $t_m$ ($t_m$ can be 3 or 4 when the majority gate with fan-in $4$ is used), the number of operations per output bit is
\begin{equation}
  \mathscr{N}^{\text{Rep}}_\mathrm{XOR-8, per\text{-}bit}=t_m\left\lceil \frac{L}{d_s-1}\right\rceil
  =\frac{t_m}{7}L,
\end{equation}
\begin{equation}
  \mathscr{N}^{\text{Rep}}_\mathrm{MAJ-t_m,per\text{-}bit}=t_m\left\lceil \frac{L}{d_s-1}\right\rceil
  =\frac{t_m}{7}L,
\end{equation}
\begin{equation}
  \mathscr{N}^{\text{Rep}}_\mathrm{AND-2,per\text{-}bit}=t_m L.
\end{equation}
Therefore, when the repetition time $t_m$ is 3 or 4, the number of operations per output bit for ENCODED-F is always smaller than the number of operations per output bit for the distributed majority voting scheme.
\begin{figure}
  \centering
  \includegraphics[scale=0.5]{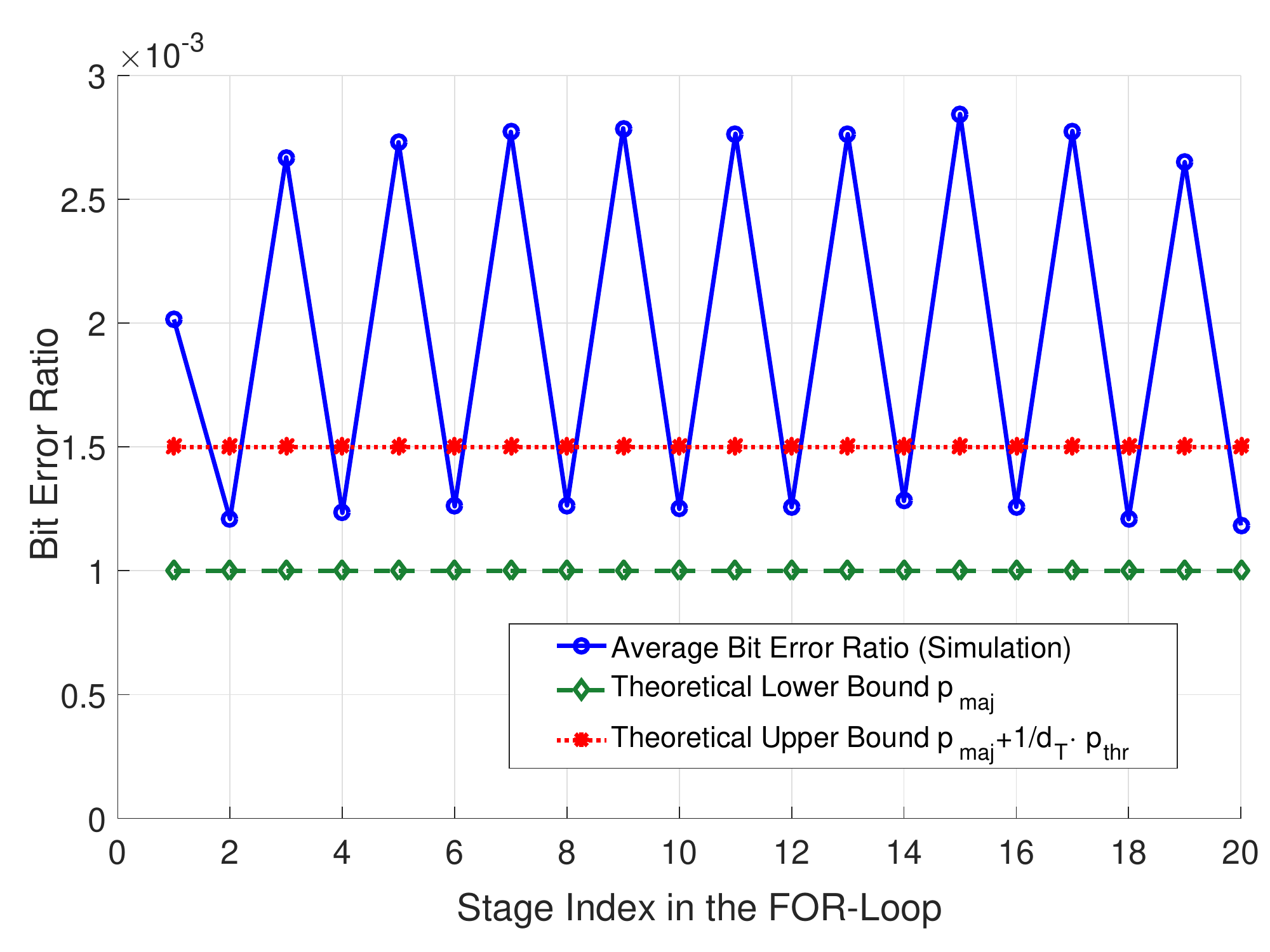}\\
  \caption{Simulated performance of ENCODED-T using a (6,12) regular LDPC with branch width $d_T=2$, and Gallager-B threshold value $b=3$. The gate error probabilities are set to $p_\text{maj}=0.001$, $p_\text{xor}=0.00026$ and $p_\text{and}=0.002$. The code length $N=1200$, the size of the linear transform satisfies $L=K=600$. The theoretical upper bound and lower bound are obtained in Theorem 1. Note that the bounds on the bit error ratio are for the error probability after the decoding stages, so they only apply for the even stages.}\label{simulation_figure}
\end{figure}
\begin{figure}
  \centering
  \includegraphics[scale=0.45]{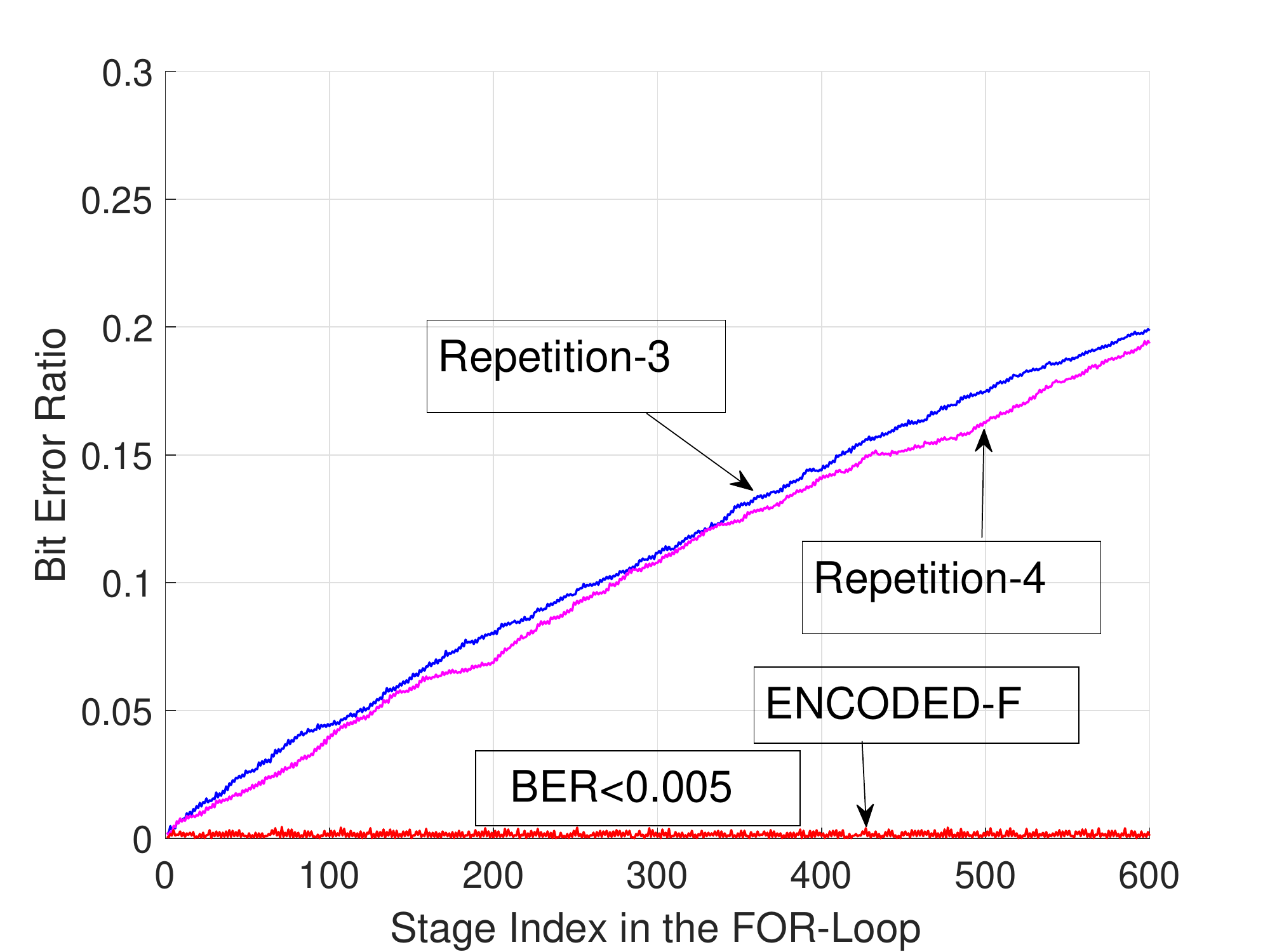}\\
  \caption{In this figure, a simulation result of ENCODED-F using a (4,8) regular LDPC with $d_s=8$ is shown. The code length $N=4000$, the size of the linear transform satisfies $L=2100$ and $K=2000$. A comparison with the distributed majority voting schemes with repetition time 3 and 4 is also shown. The gate error probabilities are set to $p_\text{and}=0.000125$, $p_\text{maj}=0.0005$ and $p_\text{xor}=0.001$ in both ENCODED-F and the distributed majority voting scheme.}\label{NLT_Comparison_nontree}
\end{figure}
\subsection{Theoretical Comparison with Repetition Coding}\label{sec:compare}

In this section, we provably show the advantage of ENCODED through theoretical analysis. We also provide a result in an online version \cite{yang2015computing} of this paper which shows ENCODED beats repetition-based techniques in scaling sense.

In this paper, although we obtain results on the number of operations for ENCODED in Theorem~\ref{Main_thm}-\ref{expander_probabilistic}, the results are biased for the comparison between ENCODED and repetition-based schemes, because the number of operations do not take into account the gate fan-in. Therefore, to compare the complexity of operations with different fan-in, we define a new concept called ``effective number of operations''.  We assume that the ``effective number of operations'' for an operation with fan-in $c$ is $\mathscr{N}_\text{c-fan-in}=c$ (the analysis for a different $\mathscr{N}_\text{c-fan-in}$ can be done similarly). We show that if we consider the problem of \textbf{find a binary linear transform scheme that achieves target error probability $p_\text{tar}=5.1\cdot 10^{-4}$ using only noisy gates with $\max(p_\text{xor},p_\text{maj},p_\text{and})<1.3\cdot 10^{-6}$, the effective number of operations of ENCODED-F is smaller than that of distributed majority voting, provided that the size of the linear transform satisfies $N=2K>9.85\cdot 10^7$, and $L>\frac{p_\text{tar}}{p_\text{and}}$}. We choose these parameters only to show that ENCODED can provably beat repetition-based schemes in situations when the parameters are not absurdly large, and hence the theoretical analysis here has potential to provide practical insight. Here $\max(p_\text{xor},p_\text{maj},p_\text{and})$ is interpreted as the maximum error probability over all types of different gates, which allows the same type of gates (i.e., MAJ-gates) with different fan-in to have different error probabilities.

\subsubsection{Counting the effective number of operations}

First, we compare the effective number of operations in both schemes. For ENCODED-F, we use a $(9,18)$ code. To make the comparison fair, we allow the distributed majority voting scheme to group several stages into one stage as well. Recall that we use $d_s$ to denote the number of stages that ENCODED-F groups into one stage. Therefore, we use $d_s'$ to denote the number of stages that distributed majority voting groups into one stage. In general, $d_s\neq d_s'$.

We compare ENCODED-F using (9,18) LDPC codes with distributed majority voting with three-time repetition. We show when
\begin{equation}
d_s>14,
\end{equation}
the ``effective'' number of operations in ENCODED-F is less than that of distributed majority voting. \textbf{Note that $d_s'$ can be arbitrary}.

The number of compute-and-correct stages in ENCODED-F is $\left\lceil\frac{L}{d_s-1}\right\rceil$, and that of distributed majority voting is $\left\lceil\frac{L}{d_s'-1}\right\rceil$. For the ease of analysis, assume $L$ is a multiple of both $d_s-1$ and $d_s'-1$. For ENCODED-F, in each compute-and-correct stage, we need $N$ XOR-operations of fan-in $d_s$ for binary addition, $P$ XOR-operations of fan-in $d_c$ and $N$ MAJ-operations of fan-in $d_v$ for LDPC decoding. In all compute-and-correct stages, the overall number of AND-operations of fan-in 2 is $NL$. Then,
\begin{equation}\label{eqn:enc_xords}
  \mathscr{N}^{\text{ENC}}_\mathrm{XOR-d_s, per\text{-}bit}=\frac{N}{K}\left\lceil \frac{L}{d_s-1}\right\rceil
  =\frac{2}{d_s-1}L,
\end{equation}
\begin{equation}\label{eqn:enc_xordc}
  \mathscr{N}^{\text{ENC}}_\mathrm{XOR-d_c, per\text{-}bit}=\frac{P}{K}\left\lceil \frac{L}{d_s-1}\right\rceil
  =\frac{1}{d_s-1}L,
\end{equation}
\begin{equation}\label{eqn:enc_majdc}
  \mathscr{N}^{\text{ENC}}_\mathrm{MAJ-d_v,per\text{-}bit}=\frac{N}{K}\left\lceil \frac{L}{d_s-1}\right\rceil
  =\frac{2}{d_s-1}L,
\end{equation}
\begin{equation}\label{eqn:enc_and2}
  \mathscr{N}^{\text{ENC}}_\mathrm{AND-2,per\text{-}bit}=\frac{NL}{K}=2L.
\end{equation}

In the distributed majority voting scheme with repetition time 3, the number of operations per output bit is
\begin{equation}\label{eqn:rep_xords}
  \mathscr{N}^{\text{Rep}}_\mathrm{XOR-d_s', per\text{-}bit}=3\left\lceil \frac{L}{d_s'-1}\right\rceil
  =\frac{3}{d_s'-1}L,
\end{equation}
\begin{equation}\label{eqn:rep_maj3}
  \mathscr{N}^{\text{Rep}}_\mathrm{MAJ-3,per\text{-}bit}=3\left\lceil \frac{L}{d_s'-1}\right\rceil
  =\frac{3}{d_s'-1}L,
\end{equation}
\begin{equation}\label{eqn:rep_and2}
  \mathscr{N}^{\text{Rep}}_\mathrm{AND-2,per\text{-}bit}=3 L.
\end{equation}

Therefore, from \eqref{eqn:enc_xords}-\eqref{eqn:enc_and2}, for ENCODED-F with $d_c=18$ and $d_v=9$, the effective number of operations is
\begin{equation}
\begin{split}
\mathscr{N}^\text{ENC}_\text{eff}&=d_s\cdot \frac{2}{d_s-1}L+ d_c\cdot \frac{1}{d_s-1}L+d_v\cdot \frac{2}{d_s-1}L+2\cdot 2L\\
&=\frac{36+2d_s}{d_s-1}L+4L=\frac{38}{d_s-1}L+6L.
\end{split}
\end{equation}
From \eqref{eqn:rep_xords} to \eqref{eqn:rep_and2}, for distributed majority voting, the effective number of operations is
\begin{equation}
\begin{split}
\mathscr{N}^\text{rep}_\text{eff}&=d_s'\cdot \frac{3}{d_s'-1}L+3\cdot \frac{3}{d_s'-1}L+2\cdot 3L\\
&=\frac{3d_s'+9}{d_s'-1}L+6L=\frac{12}{d_s'-1}L+9L>9L.
\end{split}
\end{equation}
Therefore, when $d_s>14$,
\begin{equation}
\mathscr{N}^\text{ENC}_\text{eff}<\frac{38}{13}L+6L<3L+6L=9L<\mathscr{N}^\text{rep}_\text{eff}.
\end{equation}

\subsubsection{Analyzing the probability of error}

Now, we analyze the error probability of ENCODED-F for $d_s=14$, $d_c=18$ and $d_v=9$. From Lemma \ref{Bur_lemma} in Appendix~\ref{worst_case}, using almost all codes in a ($d_v,d_c$)-regular LDPC random code ensemble with $d_v>4$ and $N$ large enough, after one iteration of the PBF algorithm, one can reduce the number of errors by at least $\theta \alpha_0 N$ for any $\alpha_0 N$ worst-case errors if $\alpha_0$ and $\theta$ are small enough. That is, using a $(d_v,d_c)$-regular LDPC code, the number of errors after one iteration of noiseless PBF algorithm will be smaller than $\alpha_0\cdot(1-\theta)$. Recall that this is the condition (A.3) that we require on the utilized LDPC code. \textcolor{black}{In Example 1 in Appendix D, for the $(9,18)$-regular LDPC code, we computed numerically the threshold value of $\alpha_0$ for $\theta=0.15$ and obtained $\alpha_0=5.1\cdot 10^{-4}$. We also obtained finite-length bounds which state that there exist $(9,18)$-regular LDPC codes with length $N=50,000$ that can reduce the number of errors by $15$\% for an arbitrary pattern of at most 20 errors, which corresponds to the case when $\alpha_0=4\cdot 10^{-4}$ and $\theta=0.15$.}

From Theorem 3, using the (9,18) code, when the maximum gate error probability $\epsilon=\max(p_\text{xor},p_\text{maj},p_\text{and})$ satisfies the condition
\begin{equation}
\begin{split}
\epsilon<\lambda=&\frac{\theta\alpha_0/2}{(d_s-1)+\left[d_c(1-R)+1 \right]+1}\\
=&\frac{\theta\alpha_0/2}{(14-1)+\left[ 18(1-\frac{1}{2})+1 \right]+1}=\frac{\theta\alpha_0}{54},
\end{split}
\end{equation}
ENCODED-F has bounded final error fraction with high probability, which is
\begin{equation}
 1-P_{e}^\text{blk}>1-3L\exp\left( -D(2\lambda\|\lambda)N \right),
\end{equation}
where $\lambda=\frac{\theta\alpha_0}{54}$ and $D(2\lambda\|\lambda)=\left( 2\log 2-1 \right)\lambda+\mathcal{O}({{\lambda}^{2}})$.

\textcolor{black}{In particular, if we choose $\epsilon=\frac{1}{60}\theta\alpha_0<\frac{1}{54}\theta\alpha_0=\lambda$, the final error fraction satisfies $\delta_e^{\text{frac}}<\alpha_0=\frac{60}{\theta}\cdot\epsilon$ with probability $1-3L\exp\left( -D(2\lambda\|\lambda)N \right)$. As we have mentioned, for $\theta=0.15$, we obtain $\alpha_0=5.1\cdot 10^{-4}$. Therefore, when the gate error probabilities satisfy $\max(p_\text{xor},p_\text{maj},p_\text{and})=\epsilon=\frac{1}{60}\theta\alpha_0=0.0043\alpha_0=1.3\cdot 10^{-6}$, the obtained error probability is smaller than $\alpha_0=60/\theta_0\cdot\epsilon=400\epsilon=5.1\cdot 10^{-4}$ with probability $1-3L\exp\left( -D(2\lambda\|\lambda)N \right)$, which is approximately 1 with reasonably large $N$, which can be guaranteed\footnote{We believe that further optimization in code design can provide techniques for error suppression for even smaller value of $N$.} if $N>\frac{50}{D(2\lambda\|\lambda)}\approx \frac{50}{(2\log 2-1)\cdot\lambda}=\frac{50}{(2\log 2-1)\cdot \frac{1}{60}\theta\alpha_0}=\frac{5.02\cdot 10^4}{\alpha_0}=9.85\cdot10^7$. }

Therefore, if we consider the problem ``find a binary linear transform scheme that achieves target error probability $p_\text{tar}=\alpha_0=5.1\cdot 10^{-4}$ using only noisy gates with $\max(p_\text{xor},p_\text{maj},p_\text{and})<1.3\cdot 10^{-6}$'', ENCODED-F has smaller ``effective number of operations'' than that of distributed majority voting. Additionally, one-time repetition or two-time repetition cannot obtain $p_\text{tar}=\alpha_0=5.1\cdot 10^{-4}$ when $L$ is reasonably large so that $\frac{1}{2}[1-(1-2p_\text{and})^L]\approx Lp_\text{and}>\alpha_0$. Thus, we conclude that ENCODED-F beats repetition-based schemes under this circumstance. Here, we acknowledge that the problem parameters (such as $\max(p_\text{xor},p_\text{maj},p_\text{and})<1.3\cdot 10^{-6}$ and $N>9.85\cdot10^7$) are chosen to show that the theoretical analysis works even when the parameter sizes are not extremely large, and thus the theoretical analysis technique has the potential to provide practical insight.

\section{Computing a Linear Transformation Reliably and Energy-Efficiently with Voltage Scaling}\label{vs}
In this section, we consider unreliable gates with tunable failure probability~\cite{Ernst_Micro_03} when {supply voltage, and hence energy consumed by gates, can be adjusted to attain a desired gate-reliability}. To model this property within Gate Model I in~\eqref{noisy_gate}, {we assume} that the added noise $z_g\sim \text{Bernoulli}(\epsilon_g(E_g))$, in which $\epsilon_g(E_v)$ is a function that depends on the supply energy $E_v$. We assume that $E_v$ is identical for all gates at any stage of the computation, while it can vary across stages. Intuitively, $\epsilon_g(\cdot)$ should be a monotonically decreasing function, since the error probability should be smaller if more energy is used. Suppose the energy-reliability tradeoff functions of AND-gates, XOR-gates and majority-gates are $\epsilon_\text{and}(\cdot)$, $\epsilon_\text{xor}(\cdot)$ and $\epsilon_\text{maj}(\cdot)$ respectively. Then, the failure probability of these three types of gates are $p_\text{and}=\epsilon_\text{and}(E_v)$, $p_\text{xor}=\epsilon_\text{xor}(E_v)$ and $p_\text{maj}=\epsilon_\text{maj}(E_v)$.

\subsection{Uncoded Matrix Multiplication vs ENCODED-T}\label{vs_static}
In this section, we compare the required energy for ENCODED-T with that for \textquoteleft uncoded\textquoteright~matrix multiplication $\mathbf{r}=\mathbf{s}\mathbf{A}$, where the circuit voltage is maintained high to ensure overall error probability is smaller than target error probability. The uncoded matrix multiplication is how almost all circuits today operate. Here, we only provide a scaling sense comparison to show the advantage of ENCODED techniques.

\begin{proposition}
When $\max\{\epsilon_\text{and}(E_v),\epsilon_\text{xor}(E_v)\}<\frac{1}{2L-2}$, to achieve bit error probability $P_e^\text{bit}<p_\text{tar}$, the energy consumption per output bit is $\Omega(L\cdot \max\{\epsilon_\text{and}^{-1}(\frac{2p_\text{tar}}{L}),\epsilon_\text{xor}^{-1}(\frac{2p_\text{tar}}{L})\})$ for the uncoded matrix multiplication scheme, while that for ENCODED-T is { $\mathcal{O}\left(\frac{LN}{K}\max\{\epsilon_\text{maj}^{-1}(\frac{1}{2}p_\text{tar}),\epsilon_\text{xor}^{-1}(\frac{1}{2}p_\text{tar}),\epsilon_\text{and}^{-1}(\frac{1}{2}p_\text{tar})\}\right)$ and $\Omega\left(\frac{LN}{K}\epsilon_\text{maj}^{-1}(p_\text{tar})\right)$.}
\end{proposition}
\begin{proof}
\textbf{``Uncoded'' scheme}: To compute each output bit using straightforward dot-product-based multiplication, one needs to compute a dot product of the message $\mathbf{s}$ with one column in the matrix $\mathbf{A}$, which needs $2L$-$1$ unreliable operations ($L$ AND-operations and $L-1$ XOR-operations). From Lemma~\ref{idp_odd}, we know that the bit error probability is
\begin{equation}
  P_e^\text{bit}=\frac{1}{2}[1-(1-2p_\text{and})^L(1-2p_\text{xor})^{L-1}].
\end{equation}
Since
\begin{equation}
  {(1-2p_\text{and} )}^{L}<1-2Lp_\text{and}+2L(L-1)p_\text{and}^2\overset{(a)}{<}1-Lp_\text{and},
\end{equation}
where step (a) follows from $p_\text{and}<\frac{1}{2L-2}$, and
\begin{equation}
\begin{split}
  {(1-2p_\text{xor} )}^{L-1}&<1-2(L-1)p_\text{xor}+2(L-1)(L-2)p_\text{xor}^2\\
  &\overset{(b)}{<}1-Lp_\text{xor},
\end{split}
\end{equation}
where step (b) follows from $p_\text{xor}<\frac{1}{2L-2}$, we get
\begin{equation}\label{vs_rep_ber}
\begin{split}
  P_e^\text{bit}>&\frac{1}{2}\left[1-(1-Lp_\text{and})(1-Lp_\text{xor})\right]\\
  =&\frac{L}{2}p_\text{and}+\frac{L}{2}p_\text{xor}-\frac{L^2}{2}p_\text{and}p_\text{xor}
  \overset{(c)}{>}\frac{L}{2}\max\{p_\text{and},p_\text{xor}\}.
\end{split}
\end{equation}
where step (c) follows from $\max\{p_\text{and},p_\text{xor}\}<\frac{1}{2L-2}<\frac{1}{L}$.
Thus, to attain a target bit error probability~$p_\text{tar}$, it must hold that $\max\{\epsilon_\text{and}(E_v),\epsilon_\text{xor}(E_v)\}<\frac{2p_\text{tar}}{L}$. Therefore, the total energy required for each output bit is $\Omega(L\cdot \max\{\epsilon_\text{and}^{-1}(\frac{2p_\text{tar}}{L}),\epsilon_\text{xor}^{-1}(\frac{2p_\text{tar}}{L})\})$.

From Theorem~\ref{Main_thm}, we know that in the ENCODED-T technique, $\mathcal{N}_\text{per-bit}=\Theta(L)$ is sufficient to achieve bit error probability smaller or equal to $p_\text{maj}+\frac{1}{d_T}p_\text{thr}$. \textcolor{black}{From~\eqref{p0satis}, it is reasonable to make $p_\text{xor}=p_\text{and}=p_\text{maj}=p_\text{thr}=\frac{1}{2}p_\text{tar}$, in which case $p_\text{maj}+\frac{1}{d_T}p_\text{thr}<2p_\text{maj}=p_\text{tar}$. Since there are $\Theta(\frac{LN}{K})$ AND-, XOR- and majority-operations in the ENCODED-T technique (see the \emph{Computational Complexity Analysis} part in the proof of Theorem~\ref{Main_thm}), the total energy required for each output bit is $\mathcal{O}\left(\frac{LN}{K}\left(\epsilon_\text{maj}^{-1}(p_\text{maj})+\epsilon_\text{xor}^{-1}(p_\text{xor})+\epsilon_\text{and}^{-1}(p_\text{and})\right)\right) =\mathcal{O}\left(\frac{LN}{K}\max\{\epsilon_\text{maj}^{-1}(\frac{1}{2}p_\text{tar}),\epsilon_\text{xor}^{-1}(\frac{1}{2}p_\text{tar}),\epsilon_\text{and}^{-1}(\frac{1}{2}p_\text{tar})\}\right)$. Furthermore, $p_\text{maj}<p_\text{tar}$ due to the `last-gate effect'. Therefore, the total energy required for each bit is at least $\Omega\left(\frac{LN}{K}\epsilon_\text{maj}^{-1}(p_\text{maj})\right)=\Omega\left(\frac{LN}{K}\epsilon_\text{maj}^{-1}(p_\text{tar})\right)$.}
\end{proof}

\begin{remark}
We show an illustrative example when $\epsilon_\text{and}(\cdot)=\epsilon_\text{xor}(\cdot)=\epsilon_\text{maj}(\cdot)=\epsilon_g(\cdot)$. Because $\epsilon_g^{-1}(u)$ typically decreases monotonically in $u$, we consider three specific cases: exponential decay, polynomial decay and sub-exponential decay. For exponential decay, we assume $\epsilon_g(u)=\exp(-c u),c>0$. Therefore, the total energy required for each output bit for the \textquoteleft uncoded\textquoteright~matrix multiplication is $\Omega(L\log\frac{L}{p_\text{tar}})$, while that for ENCODED-T is $\Theta(\frac{LN}{K}\log\frac{1}{p_\text{tar}})$. For polynomial decay, $\epsilon_g(u)=(\frac{1}{u})^{c},c>0$, the total energy required for each output bit for the \textquoteleft uncoded\textquoteright~matrix multiplication is $\Omega(L(\frac{L}{p_\text{tar}})^{\frac{1}{c}})$, while that for ENCODED-T is $\Theta(\frac{LN}{K}(\frac{1}{p_\text{tar}})^{\frac{1}{c}})$. For sub-exponential decay, we assume $\epsilon_g(u)=\exp(-c \sqrt{u}),c>0$. By sub-exponential we mean the delay is slower than exponential but faster than polynomial. The sub-exponential decay model is inspired and obtained from \cite{Patil_spin} on spintronic devices \cite{butler2012switching,kim2015spin,manipatruni2012modeling}. Therefore, the total energy required for each output bit for the \textquoteleft uncoded\textquoteright~matrix multiplication is $\Omega\left(L(\log\frac{L}{p_\text{tar}})^2\right)$, while that for ENCODED-T is $\Theta\left(\frac{LN}{K}(\log\frac{1}{p_\text{tar}})^2\right)$.
In all cases, if $K=R N$ for some constant \textquoteleft rate\textquoteright~$R$, the scaling of the required energy consumption of ENCODED-T is smaller than uncoded.
\end{remark}

\textcolor{black}{In the next subsection, we will show that using \textquoteleft dynamic\textquoteright~voltage scaling, we can achieve even lower energy by using a two-phase computation scheme called ENCODED-V. For example, when $\epsilon_g(u)=(\frac{1}{u})^{c},c>0$, the energy consumption per output bit is $\mathcal{O} \left(\frac{N}{K}\max\left\{L,\left(\frac{1}{p_\text{tar}}\right)^{\frac{1}{c}}\right\} \right)$.}

\subsection{ENCODED-V: Low-energy Linear Transformations Using Dynamic Voltage Scaling}
In this part, we modify the ENCODED-F technique in Section~\ref{expander_encoder} with \textquoteleft dynamic\textquoteright~voltage scaling to obtain arbitrarily small output error fraction. The gate model here is Model I. The original ENCODED-F technique has $\left\lceil L/(d_s-1)\right\rceil$ stages, where in each stage, a noisy decoder of the utilized LDPC code is used to carry out one (noisy) iteration of PBF decoding. In the original ENCODED-F technique, we assumed that gate failure probability is constant (and equal for all gates) throughout the duration of the computation process. Here, we partition the entire ENCODED-F technique into two phases. In the first phase, we use constant supply energy, while in the second phase, we increase the supply energy as the computation proceeds, so that the gate failure probability decreases during the computation process, in order to achieve the required output error fraction with high probability.

For ease of presentation, we consider the case when $d_s=2$, i.e., we only add $d_s-1=1$ codeword to the $N$-bit storage at each stage. The extension to general $d_s$ is straightforward. We partition the entire ENCODED-F so that there are $L-L_\text{vs}$ stages in the first phase and $L_\text{vs}$ stages in the second phase, where $L_\text{vs}$ is defined as
\begin{equation}\label{vs_A}
  {{L}_{\text{vs}}}=\left\lceil \frac{\log \frac{1}{{{p}_{\text{tar}}}}+\log {{\alpha }_{0}}}{\log \frac{1}{1-\frac{1}{2}\theta }} \right\rceil ,
\end{equation}
where $p_{\text{tar}}$ is the required final output error fraction. In the $i$-th stage of the last $L_\text{vs}$ stages, we assume that the supply energy is increased to some value to ensure that
\begin{equation}\label{vs_condition}
  [{{d}_{c}}(1-R)+1]p_{\text{xor}}^{(i+1)}+p_{\text{maj}}^{(i+1)}+p_\text{and}^{(i+1)}\le\frac{1}{4}\theta\alpha_0{\left(1-\frac{1}{2}\theta\right)^{i}}.
\end{equation}
We call this (dynamic) voltage-scaling scheme the ENCODED-V technique.
\begin{theorem}\label{thm4}
(Using dynamic voltage scaling for Problem~\ref{pro_2})
Using unreliable AND gates, majority gates and \textcolor{black}{XOR} gates defined from Gate Model I ($D,\epsilon$) with maximum fan-in $D$ and error probability $p_\text{and}$, $p_\text{xor}$ and $p_\text{maj}$, and using a regular LDPC code that satisfies assumption (A.3), the binary linear transformation $\mathbf{r}=\mathbf{s}\cdot\mathbf{A}$ can be computed using the ENCODED-F technique with dynamic voltage scaling, with per-bit energy consumption
\begin{equation}\label{vs_energy}
\begin{split}
  &E_\text{per-bit}= \\
   & \frac{L-L_\text{vs}}{K}\left[ N\epsilon_{\text{and}}^{-1}\left( {{p}_{\text{and}}} \right)+N\epsilon_{\text{maj}}^{-1}\left( {{p}_{\text{maj}}} \right)+(N+P)\epsilon_{\text{xor}}^{-1}\left( {{p}_{\text{xor}}} \right) \right] \\
  &+\frac{N}{K}\sum\limits_{i=1}^{L_\text{vs}}{\epsilon_{\text{and}}^{-1}\left( p_{\text{and}}^{(i)} \right)} +\frac{N}{K}\sum\limits_{i=1}^{L_\text{vs}}{\epsilon_{\text{maj}}^{-1}\left( p_{\text{maj}}^{(i)} \right)}\\
  &+\frac{N+P}{K}\sum\limits_{i=1}^{L_\text{vs}}{\epsilon_{\text{xor}}^{-1}\left( p_{\text{xor}}^{(i)} \right)},
 \end{split}
\end{equation}
where $L_\text{vs}$, which is a function of $p_\text{tar}$, is defined in~\eqref{vs_A}. Further, the output error fraction is below $p_\text{tar}$ with probability at least $1-P_{e}^\text{blk}$, where the probability $P_{e}^\text{blk}$ satisfies
\begin{equation}\label{vs_error_prob}
  P_{e}^\text{blk}<3(L-L_\text{vs})\exp\left( -\lambda^*N \right)+3\mathop\sum\limits_{i=1}^{L_\text{vs}}\exp \left( -{{\widetilde{\lambda }}^{(i+1)}}N \right),
\end{equation}
where
\begin{equation}\label{vs_condition_2}
  \begin{array}{*{35}{l}}
   {} & {{\lambda }^{*}}=D(2\lambda\|\lambda)=\left( 2\log 2-1 \right)\lambda+\mathcal{O}({{\lambda}^{2}}),\\
   {} & {{\widetilde{\lambda }}^{(i+1)}}=D(2\lambda^{(i+1)}\|\lambda^{(i+1)})\\
   {} & \qquad\quad=\left( 2\log 2-1 \right)\lambda^{(i+1)}+\mathcal{O}({{(\lambda^{(i+1)})}^{2}}),  \\
   {} & \lambda=\frac{\theta\alpha_0/2}{[d_c(1-R)+1]+2},\\
   {} & \lambda^{(i+1)}=\frac{\theta {{\alpha }_{0}}{{\left( 1-\frac{1}{2}\theta  \right)}^{i}}/4}{\left[ {{d}_{c}}(1-R)+1 \right]+2}.\\
\end{array}
\end{equation}
\end{theorem}
\begin{proof}
See Appendix~\ref{vs_app}.
\end{proof}
As the analysis in Section~\ref{vs_static}, we consider three specific cases of energy-reliability tradeoff: exponential decay model $\epsilon_\text{and}(u)=\epsilon_\text{xor}(u)=\epsilon_\text{maj}(u)=\exp(-c u),c>0$, polynomial decay model $\epsilon_\text{and}(u)=\epsilon_\text{xor}(u)=\epsilon_\text{maj}(u)=(\frac{1}{u})^{c},c>0$ or sub-exponential decay model $\epsilon_\text{and}(u)=\epsilon_\text{xor}(u)=\epsilon_\text{maj}(u)=\exp(-c \sqrt{u}),c>0$. We evaluate the total energy consumption per output bit in these two cases under a specific choice of supply energy that ensures the condition~\eqref{vs_condition}.

\begin{corollary}\label{vs_coro}
(Using dynamic voltage scaling for Problem~\ref{pro_1})
Using a $(d_v,d_c)$ regular LDPC code that satisfies assumption (A.3) (with parameters $\alpha_0$ and $\theta$) and has length $N>\frac{1}{{{\theta }^{*}}}\log \left( \frac{6L}{{{p}_{\text{tar}}}} \right)$, where
\begin{equation}
\begin{split}
  &{{\theta }^{*}}=\min \left\{\lambda^*,\right.\\
  &\left.D\left(2\frac{\theta {p_\text{tar}{\left( 1-\frac{1}{2}\theta  \right)}}/4}{\left[ {{d}_{c}}(1-R)+1 \right]+2}\left\|\frac{\theta {p_\text{tar}{\left( 1-\frac{1}{2}\theta  \right)}}/4}{\left[ {{d}_{c}}(1-R)+1 \right]+2}\right.\right)\right\},
\end{split}
\end{equation}
and $\lambda^*$ is defined in~\eqref{vs_condition_2}, the ENCODED-V technique can achieve output bit error probability $p_\text{tar}$ with total energy consumption pet bit $E_\text{per-bit}$: When the energy-reliability tradeoff function $\epsilon_\text{and}(u)=\epsilon_\text{xor}(u)=\epsilon_\text{maj}(u)=(\frac{1}{u})^{c},c>0$, $E_\text{per-bit}=\mathcal{O} \left( \frac{N}{K}\max\left\{L,\left(\frac{1}{p_\text{tar}}\right)^{\frac{1}{c}}\right\} \right)$; when the energy-reliability tradeoff function $\epsilon_\text{and}(u)=\epsilon_\text{xor}(u)=\epsilon_\text{maj}(u)=\exp(-c u),c>0$, $E_\text{per-bit}=\mathcal{O}\left(\frac{N}{K}\max\{L,\log^2\frac{1}{p_\text{tar}}\}\right)$; when the energy-reliability tradeoff function $\epsilon_\text{and}(u)=\epsilon_\text{xor}(u)=\epsilon_\text{maj}(u)=\exp(-c \sqrt{u}),c>0$,
$E_\text{per-bit}=\mathcal{O}\left(\frac{N}{K}\max\{L,\log^3\frac{1}{p_\text{tar}}\}\right)$.
\end{corollary}
\begin{proof}
See Appendix~\ref{vs_app_2}.
\end{proof}
We use Table I to show the energy-reliability tradeoff of ``uncoded'' matrix multiplication, ENCODED-T and ENCODED-V.

\begin{center}
\begin{table*}
\caption{\label{tab:energy}This table shows the energy-reliability tradeoffs of different computing schemes under different gate error probability models.}
    \centering
    \begin{tabular}{ | l | l | l | p{5cm} |}
    \hline
     & uncoded & ENCODED-T & ENCODED-V \\ \hline
     $\epsilon=\exp(-c u)$& $\Omega\left(L\log\frac{L}{p_\text{tar}}\right)$ & $\Theta\left(\frac{LN}{K}\log\frac{1}{p_\text{tar}}\right)$ & $\mathcal{O}\left(\frac{N}{K}\max\{L,\log^2\frac{1}{p_\text{tar}}\}\right)$ \\ \hline
    $\epsilon=(\frac{1}{u})^{c}$ & $\Omega\left(L(\frac{L}{p_\text{tar}})^{\frac{1}{c}}\right)$ & $\Theta\left(\frac{LN}{K}(\frac{1}{p_\text{tar}})^{\frac{1}{c}}\right)$ & $\mathcal{O} \left( \frac{N}{K}\max\left\{L,\left(\frac{1}{p_\text{tar}}\right)^{\frac{1}{c}}\right\} \right)$\\ \hline
    $\epsilon=\exp(-c \sqrt{u})$ & $\Omega\left(L\log^2\frac{L}{p_\text{tar}}\right)$ & $\Theta\left(\frac{LN}{K}\log^2\frac{1}{p_\text{tar}}\right)$ & $\mathcal{O}\left(\frac{N}{K}\max\{L,\log^3\frac{1}{p_\text{tar}}\}\right)$ \\
    \hline
    \end{tabular}
\end{table*}
\end{center}
\section{When a Noiseless Decoder Is Available}\label{nless_dec}
The conclusion in Theorem~\ref{expander_probabilistic} can be further tightened if we use a noiseless PBF decoder after the noisy computation. \textcolor{black}{Although the assumption that the last step of the entire computation process is fault-free is not valid under our Gate Model I or Gate Model II, it is often adopted in existing literature on computing with noisy components~\cite{Spi_FCS_96,Had_TIT_05,Chi_ITW_06}.}

\begin{theorem}[What if we have a noiseless decoder]\label{thm2}
Suppose the unreliable gates are drawn from Gate Model I $(D,\epsilon)$. Further assume that a noiseless PBF decoder is available. Then, the linear transformation $\mathbf{r}=\mathbf{s}\cdot\mathbf{A}$ that outputs $K$ bits can be computed with $P_e^\text{blk}<p_\text{tar}$ using $\frac{1}{\lambda^*}\log \frac{3L}{{{p}_\text{tar}}}=\Theta(\log\frac{1}{p_\text{tar}})$ unreliable operations per output bit and extra $\Theta(\log\log\frac{1}{p_\text{tar}})$ noiseless operations per output bit, where the parameter $\lambda^*$ is defined in~\eqref{expander_error_exp} in Theorem~\ref{expander_probabilistic}.
\end{theorem}

\begin{proof}
We use the ENCODED-F technique to do noisy linear transformations. That is, instead of using  Gallager-B decoding algorithm to correct errors, we use the PBF algorithm. Theorem~\ref{expander_probabilistic} \textcolor{black}{shows} that \textcolor{black}{the final error fraction can be upper bounded by a small constant $\alpha_0$} with high probability as long as~\eqref{expander_prob_condition} holds. The total number of operations per bit is $\frac{(3N+P)L}{K}\le \frac{4NL}{K}=\Theta (\frac{NL}{K})$.

If we require the error probability $p_\text{tar}$ to be arbitrarily small, we have to use a noiseless decoder to correct residual errors in the final output. We can use the noiseless decoder to carry out $\Theta(\log N)$ iterations of noiseless PBF algorithms to correct all errors, which \textcolor{black}{introduces an additional} $\Theta(\log N)$ operations per bit.

The overall error probability is the same as~\eqref{expander_error_prob}. To ensure that $P_e^{\text{blk}}$ is smaller than ${{p}_\text{tar}}$, it suffices (see~\eqref{expander_error_prob}) to let
\[3L\exp\left( -\lambda^*N \right)<{{p}_\text{tar}}.\]
This is satisfied when
\[N\ge \frac{1}{\lambda^*}\log \frac{3L}{{{p}_\text{tar}}}=\Theta (\log\frac{L}{{{p}_\text{tar}}}).\]
Thus, we need $\frac{4L}{K\lambda^*}\log\frac{3L}{p_\text{tar}}=\Theta(\frac{L}{K}\log\frac{L}{p_\text{tar}})$ unreliable operations per bit and extra $\Theta(\log N)=\Theta(\log\log\frac{L}{p_\text{tar}})$ noiseless operations per bit.
\end{proof}

\begin{remark}
As discussed in Remark~\ref{remark41}, the output error probability \textcolor{black}{is at least} $p_\text{maj}$, the error probability of a majority gate, due to the `last-gate effect'. Therefore, in order to achieve arbitrarily small error probability, the noiseless operations in Theorem~\ref{thm2} are necessary. 

In fact, the bound in Theorem~\ref{low_bound_thm} is a lower bound on the number of operations that are used at the entrance stage, \textcolor{black}{i.e., operations that have one of the $L$ inputs $(s_1, s_2,...s_L)$ as an argument}, of the computation scheme. Therefore, Theorem~\ref{thm2} and Theorem~\ref{low_bound_thm} together assert that the number of noisy operations scales as $\Theta(\log\frac{1}{p_\text{tar}})$ under the setting of Problem~\ref{pro_1}, if the `last-gate effect' can be addressed using a few noiseless operations which scales as $\Theta(\log\log\frac{1}{p_\text{tar}})$.
\end{remark}

\section{Conclusions and Future Work}\label{conclusion}
Can reliable computation be performed using gates that are all equally unreliable? As we discussed, the error probability is lower bounded by the last gate's error probability $\epsilon$. We provide LDPC codes-based strategies (called ENCODED) that attain error probability close to $\epsilon$ (which we bound by $2\epsilon$). Further, we show that these strategies outperform repetition-based strategies that are commonly used today.

The key idea that ENCODED relies on is to repeatedly suppress errors in computation process by, in a sense, encoding the computation matrix of the linear transformation, instead of encoding inputs  (as is done in traditional communication). Using ENCODED, both probabilistic errors and worst-case errors can be kept suppressed.

Inspired by voltage-scaling techniques commonly used to reduce power in circuit design, we also analyzed possible gains attainable using `static' and `dynamic' voltage scaling in conjunction with our ENCODED technique. It would be meaningful to experimentally model the power-reliability tradeoffs of voltage scaling to give more insights to the designer. On the modeling side, it would also be important to include energy consumed in wiring~\cite{Gro_ISIT_12,Blake1,Blake2} (which can be a significant chunk of the total energy in decoding circuits~\cite{KarthikSips}) in these models, and observe if predicted gains due to coding are significantly reduced. Perhaps wiring energy will also motivate design of novel coding techniques that attempt to correct errors with local information as much as possible.

There are many coding-theoretic problems that fall out naturally. What are practical codes that can be used to reduce computational errors? Are there benefits to applying more recent discoveries, such as spatially coupled LDPC codes~\cite{Kud_ISIT_12}, instead of expander codes?

More broadly, the problem of computation with noisy gates is of considerable practical and intellectual interest. It is widely accepted that biological systems operate with noisy computational elements, and yet provide good performance at low energy. In engineered systems, with saturation of Dennard's scaling and Moore's law, new device technologies are being used to design circuits that are invariably error-prone. A comprehensive understanding of reliability-resource tradeoffs in error-correction coding in computing could give these novel technologies (e.g. carbon nanotubes and mechanical switches) a better chance to compete with established ones (i.e., CMOS). To that end, it will be key to identify what causes faults in these novel technologies so that they can be modeled and analyzed, and appropriate codes be designed for them. Intellectually, it is interesting (and widely acknowledged) that the remarkable gains that coding brings to communications, especially at long-range, are not easy to obtain in computational settings. The theoretical reasoning for this thus far rests on simplistic models and has rather loose bounds~\cite{Gro_ISIT_14}. Improved strategies and improved outer bounds will go a long way in characterizing how large these gains can be.

\subsection{Connections with coded computing with stragglers and ``exascale computing''}
We note an important connection between ENCODED and the recent works on coded computation in presence of ``stragglers'' \cite{lee2016speeding,YaoqingISTC,dutta2016short,duttaISIT2017,PedarsaniHeterogeneous}. These works focus on ``processor-level'' (rather than gate-level) noise, e.g. it is assumed in~\cite{lee2016speeding} that the product of input $\mathbf{s}$ with each column of $\mathbf{A}$ is ``erasure-prone'' with some probability (which depends on models of time required for computation). The formulation there is not applicable in two ways that are crucial for modern ``exascale'' computing systems: (i) there is an increasing trend in distributed systems community to consider ``soft-errors'' that are undetectable~\cite{cappello2014toward}, whereas~\cite{lee2016speeding,YaoqingISTC,dutta2016short,duttaISIT2017,PedarsaniHeterogeneous} largely focus on erasures; and more importantly, (ii) there is an increasing need for understanding scalability when the number of (fixed memory) processors increases \textit{for a fixed total problem size} (to understand the limits of gains with parallelization of a problem). This is called ``strong scaling''~\cite[Chapter 9]{kaminsky2016big}, whereas ``weak scaling'' allows for increasing problem size and number of processors, while keeping memory of each processor fixed. The works~\cite{lee2016speeding,YaoqingISTC,dutta2016short,duttaISIT2017,PedarsaniHeterogeneous} only examine a fixed number of processors with \textit{increasing} memory of each processor with problem size, which is, strictly speaking, ``weaker than weak'' scaling.

For the specific problem of matrix-vector multiplication, strong scaling allows adding more processors than the number of rows and/or columns of the matrix to increase parallelization. For example, when computing $\mathbf{s}\times \mathbf{A}$, the matrix $\mathbf{A}$ is often split into both horizontal and vertical pieces (as used in the algorithm ``SUMMA''~\cite{van1997summa}), ENCODED can be adapted to this split to suppress error propagation (horizontal decomposition is similar to \eqref{quiv_encoding_matrix}, and vertical decomposition is imposed by using limited memory gates). In ENCODED-T, the tree-structure helps introduce increased parallelism, speeding up the computation, while keeping errors in check through repeated error suppression. However, the algorithms in \cite{lee2016speeding,YaoqingISTC,dutta2016short,duttaISIT2017,PedarsaniHeterogeneous} do not easily adapt to horizontal division. If one naively uses strategies in~\cite{lee2016speeding,YaoqingISTC,dutta2016short,duttaISIT2017,PedarsaniHeterogeneous} for strong scaling with soft errors, errors will accumulate and cause the resulting output to be far from the correct output.

One limitation of the technique proposed here is that it is limited to finite fields instead of real number coding. It is important to extend ENCODED to real number codes, and a preliminary attempt is made in \cite{yang2016fault} on iterative algorithms for logistic regression, where LDPC-type real-number coding techniques (inspired from~\cite{zhang2012verification}) are used for error-correction over reals.

\section{Acknowledgements}
We thank the SONIC center members, especially Ameya Patil, Naresh Shanbhag, Andrew Bean, and Andy Singer, for discussions that motivated this paper, and pointers to spintronics models and connections with practice. We also thank Shawn Blanton and Franz Franchetti for useful discussions regarding practical aspects of the problem. We thank David Burshtein for discussions regarding the analysis of the flipping algorithm. We thank Tze Meng Low for useful discussions on strong scaling and the reference to SUMMA. We also thank Zhuo Jiacheng (Carlson) and Paul Griffioen for their careful reading of the paper, pointing out errors and typos, and helpful comments.

\appendices
\section{Details of Gallager-B decoding algorithm}\label{GB_details}
Assume a variable node $v$ is connected to $d_v$ parity check nodes in $\mathcal{N}_v$ and a parity check node $c$ is connected to $d_c$ variable nodes in $\mathcal{N}_c$. Suppose the received bits are $\textbf{r}=(r_1,...r_N)$. The decoding algorithm we use is the Gallager-B algorithm:
\begin{itemize}
  \item From variable node to check node:
  \begin{itemize}
    \item Iteration 0: $m_{v\to c}^{(0)}=r_v$ is transmitted from $v$ to every check node $c\in \mathcal{N}_v$.
    \item Iteration $i$: $m_{v\to c}^{(i)}$ is transmitted from $v$ to $c\in \mathcal{N}_v$,
    \begin{equation}\label{chnode}  m_{v\to c}^{(i)} = \left\{ {\begin{array}{*{20}{c}}
        {x,}\\
        {z,}
        \end{array}\begin{array}{*{20}{c}}
        \text{if $|{c'\in \mathcal{N}_v\setminus c: m_{c' \to v}^{(i-1)}= x}|\ge b$,}\\
        \text{otherwise,}
        \end{array}} \right.\end{equation}
        where $b=\left\lfloor {\frac{{{d_v} + 1}}{2}} \right\rfloor$ and $z$ is a randomly generated bit.
  \end{itemize}
  \item From check node to variable node:
  \begin{itemize}
    \item Iteration $i$: $m_{c\to v}^{(i)}$ is transmitted from check node $c$ to variable node $v\in \mathcal{N}_c$,
    \begin{equation}\label{varnode}  m_{c\to v}^{(i)}=\mathop \oplus \limits_{v'\in \mathcal{N}_c\slash v}m_{v'\to c}^{(i-1)}.\end{equation}
  \end{itemize}
\end{itemize}
\begin{remark}\label{Yaz_remark}
Note that the updating rule $m_{v\to c}^{(i)}=z$ in~\eqref{chnode}, which is used to break ties, is different from the original rule $m_{v\to c}^{(i)}=y_v$ in~\cite{Gal_TIT_62}, in which $y_v$  is the channel output associated with the variable node $v$. This is because the problem that we consider is a computing problem, instead of a communication problem, and hence there are no channel outputs. Note that the analysis is also done for the modified updating rule. Although the modified updating rule is theoretically sound, we acknowledge that the cost of generating a random bit may be higher than that of the majority rule.
\end{remark}

\section{Proof of Lemma~\ref{lemma7}}\label{Upper_Bound_Analysis}
In this section, we prove that bit error probability can be made below a small constant $p_\text{maj}+\frac{1}{d_T}p_\text{thr}$ after one iteration of Gallager-B decoding. We use density evolution to analyze the change of error probability before and after decoding.
\subsection{Density Evolution Analysis}
We examine the $m$-th level in the tree structure of the ENCODED-T. After the outputs from the $(m+1)$-th level \textcolor{black}{are obtained}, they are forwarded to the $m$-th level of the tree structure to perform a \textcolor{black}{component-wise XOR-operation}. The results of this XOR-operation are stored in the \textcolor{black}{$E$}-bit registers at a node $\mathbf{v}_{m}^l$ (a compute-and-correct unit) in the $m$-th level and is decoded using $C$ iterations of the Gallager-B algorithm.

Now, we focus on the $C$ iterations of Gallager-B decoding done at the node $\mathbf{v}_{m}^l$. For simplicity, we write the \textcolor{black}{message-passing result} after the $i$-th iteration as \textcolor{black}{$\tilde{\mathbf{x}}^{(i)}=(x_{v\rightarrow c}^{(i)})$, which is the vector constituted by the messages sent from variable nodes to parity check nodes}. The initial input $\tilde{\mathbf{x}}^{(0)}$ is the \textcolor{black}{output} of the \textcolor{black}{unreliable XOR-gates in the node $\mathbf{v}_{m}^l$}. Denote the correct \textcolor{black}{message-passing bits} by \textcolor{black}{$\tilde{\mathbf{w}}^{(i)}=(w_{v\rightarrow c}^{(i)})$, i.e., if no computing errors are introduced in the entire computation process, in contrast to just iteration $i$. We write $p_{v\rightarrow c}^{(i)}$ as the bit error probability of $x_{v\rightarrow c}^{(i)}$, that is,
\[p_{v\rightarrow c}^{(i)}=\Pr(x_{v\rightarrow c}^{(i)}\neq w_{v\rightarrow c}^{(i)})\]
We want to calculate the evolution of $p_{v\rightarrow c}^{(i)}$ with $i$.}

From~\cite{Ric_TIT_01,Var_TIT_11} we know that in density-evolution analysis, the bit error probability does not depend on the transmitted codeword, based on the check node and variable node symmetry of the message-passing algorithm, and the channel symmetry and the message noise symmetry~\cite[Def.~5]{Var_TIT_11}. In our problem, the channel symmetry comes from the fact that the AND gates flip different outputs with the same probability. The message wire symmetry comes from the fact that the majority gates and XOR gates flip outputs with the same probability. Note that we do not need the source symmetry, and hence we can use the proof of Theorem 1 in~\cite{Var_TIT_11} to show that the bit error probability $P_e^\text{bit}$ does not depend on the correct codeword at the node $\mathbf{v}_{m}^l$. Therefore, we can assume without loss of generality that the correct input (and hence output) in the linear computation $\mathbf{s}\cdot\tilde{\mathbf{G}}$ is an all-zero codeword and hence
\textcolor{black}{
\begin{equation}
p_{v\rightarrow c}^{(i)}=\Pr(x_{v\rightarrow c}^{(i)}\neq 0).
\end{equation}}
\textcolor{black}{From assumption (A.3)} we know that when the number of levels in the tree structure in the ENCODED-T technique is smaller or equal to $\frac{\log N}{2\log(d_v-1)(d_c-1)}$, we can assume that the decoding neighborhood for each variable node is cycle-free and all bits entering a majority-gate or an XOR-gate are independent of each other. In our case, choosing $C=1$ iterations at each level, we can indeed ensure that the constraint on number of decoding iterations holds, since the tree structure in the ENCODED-T technique makes the total number of decoding iterations equal to $C\cdot (M-1)=\left\lceil\frac{\log L}{\log d_T}\right\rceil$, (see~\eqref{treebigger}, we use $M-1$ because only non-leaf nodes have embedded decoders), and the tree-width $d_T$ can be set large enough so that~\eqref{D_satisfy} is satisfied.

Therefore, based on symmetry and independence, we can use the analysis for the noisy Gallager-B decoder in~\cite{Yaz_TC_01} and attain the performance predicted by density evolution for regular LDPC codes. For $b=\left\lfloor {\frac{{{d_v} + 1}}{2}} \right\rfloor$, $l\in \mathbb{Z}\cap[1,d_v-1]$, define four functions:
\begin{eqnarray}
&& \bar{\alpha}(u):=\frac{1-(1-{2u})^{d_c-1}}{2},  \label{alpha} \\
&& \bar{\gamma}(u):=(1-p_{\text{xor}})\bar{\alpha}(u)+p_{\text{xor}}(1-\bar{\alpha}(u)), \label{gamma} \\
&& \Lambda_l(\bar{\gamma}):=\binom{d_v-1}{l} (1-\bar{\gamma})^l \bar{\gamma}^{d_v-1-l},  \label{Lambda}\\
&& \bar{\eta}(\bar{\gamma}):=(1-p_0)\mathop\sum \limits_{l=0}^{d_v-b-1} \Lambda_l(\bar{\gamma})+p_0\mathop\sum \limits_{l=0}^{b-1} \Lambda_l(\bar{\gamma}).{\;\;\;\;\;\;\;\;} \label{Eta}
\end{eqnarray}
Intuitively, $\bar{\gamma}(u)$ denotes the error probability after the XOR-operation at a check node and $\bar{\eta}(u)$ denotes the error probability after the majority-operation at a variable node. These functions are \textcolor{black}{borrowed from~\cite{Yaz_TC_01}} and are crucial for analyzing noisy Gallager-B density evolution. Note that we change the form of functions ${\alpha}(\cdot)$, ${\gamma}(\cdot)$, $\eta(\cdot)$ in~\cite{Yaz_TC_01} into $1-\bar{\alpha}(\cdot)$, $1-\bar{\gamma}(\cdot)$, $1-\bar{\eta}(\cdot)$, in correspondence with the usual goal of analyzing error probability, instead of correctness probability.

We first state a result from~\cite{Yaz_TC_01}, and then simplify the result using an upper bound. Note that the LDPC decoding rule used in~\cite{Yaz_TC_01} is slightly different from ours, as stated in Remark~\ref{Yaz_remark}. We will address this issue in the proof of the upper bound.
\begin{lemma}[\cite{Yaz_TC_01}, pp. 1662, Theorem 1]\label{irregular_DE}
For regular LDPC codes with check node degree $d_c$ and variable node degree $d_v$
\begin{equation}\label{Irr_DE}
\begin{split}
  \noindent p^{(i+1)}=&f(p^{(i)})\\
  :=&p_{\text{maj}}\left[1-\bar{\eta}(\bar{\gamma}(p^{(i)}))\right]
  +(1-p_{\text{maj}})\bar{\eta}\left(\bar{\gamma}(p^{(i)})\right)\\
  =&p_{\text{maj}}+(1-2p_{\text{maj}})\bar{\eta}\left(\bar{\gamma}(p^{(i)})\right),
\end{split}
\end{equation}
where $\bar{\eta}(\cdot)$ is defined in~\eqref{Eta} and $\bar{\gamma}(\cdot)$ is defined in~\eqref{gamma}.
\end{lemma}
The lower bound $p_e^\text{dec}>p_{\text{maj}}$ in~\eqref{Converge} follows from the fact that $p_\text{maj}<\frac{1}{2}$ (see Gate Model I) and $p_e^\text{dec}=p^{(C)}=f(p^{(C-1)})$, where $C$ is the total number of iterations of decoding at each level. In what follows, we upper-bound the RHS of~\eqref{Irr_DE} by upper-bounding the functions $\bar{\gamma}(\cdot)$ and $\bar{\eta}(\cdot)$ defined in~\eqref{gamma} and \eqref{Eta}. The result is shown in Lemma~\ref{Irr_DE_ge_lmm}.
\begin{lemma}\label{Irr_DE_ge_lmm}
For regular LDPC codes with check node degree $d_c$ and variable node degree $d_v$
\begin{equation}\label{Irr_DE_ge}
  \noindent p^{(i+1)}<f_0(p^{(i)}):=p_{\text{maj}}+\bar{\eta}_0(\bar{\gamma}_0(p^{(i)})),
\end{equation}
where
\begin{equation}\label{Eta_0}
  \bar{\eta}_0(\bar{\gamma})=\left( \begin{matrix}
  {{d}_{v}}-1 \\
  \lfloor \frac{d_v+1}{2}\rfloor \\
\end{matrix} \right){\bar{\gamma} }^{\left\lceil\frac{d_v-1}{2} \right\rceil}.
\end{equation}
\begin{equation}\label{gamma_0}
  \bar{\gamma}_0(u)=(d_c-1)u+p_\text{xor}.
\end{equation}
\end{lemma}
\begin{proof}
First, note that the decoding algorithm used in~\cite{Yaz_TC_01} is slightly different from ours in that the tie-breaking rule in~\cite{Yaz_TC_01} is $m_{v\to c}^{(i)}=y_v$, which is different from our rule $m_{v\to c}^{(i)}=z$, where $z$ is a randomly generated bit (see Remark~\ref{Yaz_remark}). It can be shown that, if our updating rule is used, the $\bar{\eta}$ function in the density evolution function~\eqref{Irr_DE} should be changed from~\eqref{Eta} to
\begin{equation}\label{bar_eta_change}
  \bar{\eta}(\bar{\gamma}):=\frac{1}{2}\mathop\sum \limits_{l=0}^{d_v-b-1} \Lambda_l(\bar{\gamma})+\frac{1}{2}\mathop\sum \limits_{l=0}^{b-1} \Lambda_l(\bar{\gamma}).
\end{equation}
When $d_v$ is an even number, $b=\lfloor\frac{d_v+1}{2}\rfloor=\frac{d_v}{2}=d_v-b$. When $d_v$ is an odd number, $b=\lfloor\frac{d_v+1}{2}\rfloor=\frac{d_v+1}{2}=d_v-b+1$. In both cases, we have $d_v-b\le b$. Therefore,~\eqref{bar_eta_change} can be upper bounded by
\begin{equation}
\begin{split}
  \bar{\eta}(\bar{\gamma})&=\frac{1}{2}\mathop\sum \limits_{l=0}^{d_v-b-1} \Lambda_l(\bar{\gamma})+\frac{1}{2}\mathop\sum \limits_{l=0}^{b-1} \Lambda_l(\bar{\gamma})\\
  &\le\frac{1}{2}\mathop\sum \limits_{l=0}^{b-1} \Lambda_l(\bar{\gamma})+\frac{1}{2}\mathop\sum \limits_{l=0}^{b-1} \Lambda_l(\bar{\gamma})=\mathop\sum \limits_{l=0}^{b-1} \Lambda_l(\bar{\gamma}).
\end{split}
\end{equation}
Further
\[\begin{split}
   \sum\limits_{l=0}^{b-1}{{{\Lambda }_{l}}(\bar{\gamma} )}&=\sum\limits_{l=0}^{b-1}{\left( \begin{matrix}
  {{d}_{v}}-1 \\
  l \\
\end{matrix} \right){{\left( 1-\bar{\gamma}  \right)}^{l}}{{\bar{\gamma} }^{{{d}_{v}}-1-l}}} \\
 & ={{\bar{\gamma} }^{{{d}_{v}}-b}}\sum\limits_{l=0}^{b-1}{\left( \begin{matrix}
  {{d}_{v}}-1 \\
  l \\
\end{matrix} \right){{\left( 1-\bar{\gamma}  \right)}^{l}}{{\bar{\gamma} }^{b-1-l}}} \\
 & \le {{\bar{\gamma} }^{{{d}_{v}}-b}}\sum\limits_{l=0}^{b-1}{\left( \begin{matrix}
  {{d}_{v}}-1 \\
  b-1 \\
\end{matrix} \right){{\left( 1-\bar{\gamma}  \right)}^{l}}{{\bar{\gamma} }^{b-1-l}}} \\
 & \le {{\bar{\gamma} }^{{{d}_{v}}-b}}\sum\limits_{l=0}^{b-1}{\left( \begin{matrix}
  {{d}_{v}}-1 \\
  b-1 \\
\end{matrix} \right)\left( \begin{matrix}
  b-1 \\
  l \\
\end{matrix} \right){{\left( 1-\bar{\gamma}  \right)}^{l}}{{\bar{\gamma} }^{b-1-l}}} \\
 & =\left( \begin{matrix}
  {{d}_{v}}-1 \\
  b-1 \\
\end{matrix} \right){{\bar{\gamma} }^{{{d}_{v}}-b}}
\overset{(a)}{=}\left( \begin{matrix}
  {{d}_{v}}-1 \\
  b-1 \\
\end{matrix} \right){\bar{\gamma} }^{\left\lceil\frac{d_v-1}{2} \right\rceil},
\end{split}\]
where step (a) follows from $d_v-b=\left\lceil\frac{d_v-1}{2} \right\rceil$, which can be readily checked by $b=\lfloor \frac{d_v+1}{2}\rfloor$. Therefore,
\begin{equation}\label{eta_bd}
  \bar{\eta}(\bar{\gamma})\le\left( \begin{matrix}
  {{d}_{v}}-1 \\
  \lfloor \frac{d_v-1}{2}\rfloor \\
\end{matrix} \right){\bar{\gamma} }^{\left\lceil\frac{d_v-1}{2} \right\rceil}=\bar{\eta}_0(\bar{\gamma}).
\end{equation}
For the function $\bar{\gamma}(u)$ in~\eqref{gamma}, we upper bound it with the following two inequalities:
\[\bar{\alpha} (u)=\frac{1-{{(1-2u)}^{{{d}_{c}}-1}}}{2}\le \frac{1-\left[ 1-({{d}_{c}}-1)2u \right]}{2}=({{d}_{c}}-1)u,\]
\[\bar{\gamma}(u)=\bar{\alpha}(u)+p_\text{xor}-2\bar{\alpha}(u)p_\text{xor}<\bar{\alpha}(u)+p_\text{xor}.\]
Therefore
\begin{equation}\label{gamma_bd}
  \bar{\gamma}(u)<\bar{\alpha}(u)+p_\text{xor}\le({{d}_{c}}-1)u +p_\text{xor}=\bar{\gamma}_0(u).
\end{equation}
Combining~\eqref{gamma_bd} and~\eqref{eta_bd}
\[\begin{split}
   {{p}^{(i+1)}}&={{p}_{\text{maj}}}\left[ 1-\bar{\eta} (\bar{\gamma} ({{p}^{(i)}})) \right]+(1-{{p}_{\text{maj}}})\bar{\eta} \left( \bar{\gamma} ({{p}^{(i)}}) \right) \\
 & ={{p}_{\text{maj}}}+(1-2{{p}_{\text{maj}}})\bar{\eta} \left( \bar{\gamma} ({{p}^{(i)}}) \right) \\
 & \le {{p}_{\text{maj}}}+\bar{\eta} \left( \bar{\gamma} ({{p}^{(i)}}) \right) \\
 & \le {{p}_{\text{maj}}}+{{\bar{\eta} }_{0}}\left( \bar{\gamma} ({{p}^{(i)}}) \right) \overset{(a)}{\le} {{p}_{\text{maj}}}+{{\bar{\eta} }_{0}}\left( {{\bar{\gamma} }_{0}}({{p}^{(i)}}) \right),
\end{split}\]
where step (a) is due to the fact that ${\bar{\eta} }_{0}(\bar{\gamma})$ is monotonically increasing.
\end{proof}
\subsection{Completing the Proof of Lemma~\ref{lemma7}}
We need to prove that if the bit error probability before decoding is smaller than $p_\text{reg}=(D+1)p_\text{thr}+p_\text{xor}$, the bit error probability after decoding is smaller than $p_\text{maj}+\frac{1}{d_T}p_\text{thr}$. Using Lemma~\ref{Irr_DE_ge_lmm}, we only need to prove
\begin{equation}
  {{p}_{\text{maj}}}+\left( \begin{matrix}
   {{d}_{v}}-1  \\
   \lfloor \frac{{{d}_{v}}-1}{2}\rfloor   \\
\end{matrix} \right){{\left[ ({{d}_{c}}-1)p_\text{reg}+{{p}_{\text{xor}}} \right]}^{\left\lceil \frac{{{d}_{v}}-1}{2}\right\rceil }}<p_\text{maj}+\frac{1}{d_T}p_\text{thr},
\end{equation}
which is equivalent to
\begin{equation}\label{der9}
  \begin{split}
{{\left[ ({{d}_{c}}-1)p_\text{reg}+{{p}_{\text{xor}}} \right]}^{\left\lceil \frac{{{d}_{v}}-1}{2}\right\rceil }}<\frac{1}{d_T}{{\left( \begin{matrix}
   {{d}_{v}}-1  \\
   \lfloor \frac{{{d}_{v}}-1}{2}\rfloor   \\
\end{matrix} \right)}^{-1}}{{p}_{\text{thr}}}.
\end{split}
\end{equation}
We know that
\[\begin{split}
   ({{d}_{c}}-1){{p}_\text{reg}}+{{p}_{\text{xor}}}&\overset{(a)}{=}({{d}_{c}}-1)\left[ {{p}_{\text{xor}}}+(d_T+1){{p}_{\text{thr}}} \right]+{{p}_{\text{xor}}} \\
 & =({{d}_{c}}-1)(d_T+1){{p}_{\text{thr}}}+{{d}_{c}}{{p}_{\text{xor}}} \\
 & \overset{(b)}{\le} ({{d}_{c}}-1)(d_T+1){{p}_{\text{thr}}}+(d_T+1){{p}_{\text{thr}}}\\
 &\le {{d}_{c}}(d_T+1){{p}_{\text{thr}}},
\end{split}\]
where step (a) is from the definition $p_\text{reg}=(d_T+1)p_\text{thr}+p_\text{xor}$ and step (b) is from the second condition in~\eqref{p0satis}. Thus, to prove~\eqref{der9}, it suffices to prove
\[{{\left(  {{d}_{c}}(d_T+1){{p}_{\text{thr}}}\right)}^{\left\lceil \frac{{{d}_{v}}-1}{2}\right\rceil }}<\frac{1}{d_T}{{\left( \begin{matrix}
   {{d}_{v}}-1  \\
   \lfloor \frac{{{d}_{v}}-1}{2}\rfloor   \\
\end{matrix} \right)}^{-1}}{{p}_{\text{thr}}},\]
which is equivalent to
\[{{p}_{\text{thr}}}^{\left\lceil \frac{{{d}_{v}}-3}{2}\right\rceil }<{{\left( \begin{matrix}
   {{d}_{v}}-1  \\
   \lfloor \frac{{{d}_{v}}-1}{2}\rfloor   \\
\end{matrix} \right)}^{-1}}{{d}_{c}}^{-\left\lceil \frac{{{d}_{v}}-1}{2}\right\rceil }d_T^{-1}{{(d_T+1)}^{-\left\lceil \frac{{{d}_{v}}-1}{2}\right\rceil }}.\]
We have defined $d=\left\lceil \frac{{{d}_{v}}-1}{2}\right\rceil ={{d}_{v}}-b$ in Theorem~\ref{Main_thm}. Thus, the above inequality is ensured by
\[{{p}_{\text{thr}}}<{{\left( \begin{matrix}
   {{d}_{v}}-1  \\
   \lfloor \frac{{{d}_{v}}-1}{2}\rfloor   \\
\end{matrix} \right)}^{-\frac{1}{d-1}}}{{d}_{c}}^{-\frac{d}{d-1}}d_T^{-\frac{1}{d-1}}{{(d_T+1)}^{-\frac{d}{d-1}}},\]
which is the first condition in~\eqref{p0satis}.

\section{Proof of Theorem~\ref{low_bound_thm}}\label{Lower_Bound}
\textcolor{black}{Theorem~\ref{low_bound_thm} provides a lower bound on the number of operations by lower-bounding the operations done at the entrance stage of the noisy circuit, i.e., operations that have one of the $L$ inputs $(s_1, s_2,...s_L)$ as an argument.} In order to prove Theorem~\ref{low_bound_thm}, we need the following lemma (stated implicitly in~\cite[Proposition 1]{Pip_TIT_91}) which characterizes the equivalence of a noisy-gate model and a noisy-wire model.
\begin{lemma}
\textcolor{black}{For each unreliable gate from Gate Model I ($D,\epsilon$) with error probability $\epsilon$ and fan-in number $\le D$, its output variable can be stochastically simulated by (equivalent in distribution to) another unreliable gate $\tilde{g}$ that computes the same function but with the following property: each input wire flips the input independently with probability $\epsilon/D$ and the gate has additional output noise independent of input wire noise.}
\end{lemma}
\begin{proof}For an arbitrary unreliable gate
\[y=g(u_1,u_2,...,u_d)\oplus z_g,d\le D,\]
consider another unreliable gate together with noisy wires
\[\begin{split}\tilde{y}&=\tilde{g}(u_1,u_2,...,u_d)\oplus \tilde{z}_g\\
&=g(u_1\oplus w_1,u_2\oplus w_2,...,u_d\oplus w_d)\oplus \tilde{z}_g,d\le D,\end{split}\]
where $w_j$ is the noise on the $j$-th input wire and takes value $1$ with probability $\epsilon/D$. The probability that all $d$ wires convey the correct inputs is $(1-\epsilon/D)^d>1-d\frac{\epsilon}{D}>1-\epsilon$. Therefore, if $\tilde{z}_g$ is $0$ w.p.1, the error of $\tilde{g}$ will be smaller than $\epsilon$. Thus, using standard continuity arguments, we can find a random variable $\tilde{z}_g$ which equals to $1$ w.p. $\epsilon'<\epsilon$, while making $\tilde{y}$ and $y$ equivalent in distribution.
\end{proof}
Based on this lemma, we know that a noisy network defined in Section~\ref{Model} can always be replaced by another network, where each wire has an error probability $\frac{\epsilon}{D}$. Before \textcolor{black}{a specific input $s_k$} enters the noisy circuit, it is always transmitted along the wires connected to the \textcolor{black}{entrance stage} of the gates in the circuit. Because of the assumption that gates after the inputs are noisy, the bit will be `sampled' by the noisy wires. For convenience of analysis, we assume each gate can only be used once so that the number of operations is equal to the number of unreliable gates. Now that each gate only computes once, each noisy wire can only carry information once as well. We assume each $s_k$ is transmitted on $T_k$ distinct wires. Then, the probability that the message on all $T_k$ wires flips is
\begin{equation}
p_k=(\epsilon/D)^{T_k}.
\end{equation}
\textcolor{black}{Therefore, the error probability of the input bit $s_k$ satisfies $P_\text{in}^k>p_k$. Since matrix $\mathbf{A}$ is assumed to have full row rank, if the linear transformation computation is noiseless, even a single input bit error leads to an output block error. Therefore, even when the linear transformation computation is noiseless, the output block error probability $P_e^\text{blk}$ is greater than the input error probability $P_\text{in}^k$. Since the computation is noisy, $P_e^\text{blk}$ is still greater than $P_\text{in}^k$, and hence is greater than $p_k$. Therefore, if $(\epsilon/D)^{T_k}=p_k> p_{\text{tar}}$, the block error probability $P_e^\text{blk}>P_\text{in}^k>p_k>p_{\text{tar}}$, which contradicts with the aim to make the block error probability smaller than $p_{\text{tar}}$}. Thus,
\[p_{\text{tar}}>(\epsilon/D)^{T_k},\]
which means that for any bit $s_k$
\begin{equation}
T_k>\frac{\log{1/p_{\text{tar}}}}{\log{D/\epsilon}}.
\end{equation}
Therefore, the number of wires connected to each input bit must be at least $\frac{\log 1/p_{\text{tar}}}{\log D/\epsilon}$. Since the number of input bits is $L$, the total number of wires connected to all input bits is at least $\frac{L\log 1/p_{\text{tar}}}{\log D/\epsilon}$. Since we are using gates with bounded fan-in smaller than $D$, the number of gates is at least $\frac{L\log 1/p_{\text{tar}}}{D\log D/\epsilon}$, so does the number of operations. Since there are $K$ output bits, the number of operations per output bit $\mathscr{N}_\mathrm{per\text{-}bit}>\frac{L\log 1/p_{\text{tar}}}{KD\log D/\epsilon}$.

\section{Codes that satisfy the requirement (A.3)}\label{worst_case}
The existence of codes that satisfy the requirement (A.3) follows from a result in \cite{burshtein2008error}. We first present the result from \cite{burshtein2008error}.

Define $\beta_0,\beta_1,\beta_2,\beta_3$ respectively as the largest integer less than $d_v/2$, the largest integer less than or equal to $d_v/2$, the smallest integer greater than or equal to $d_v/2$, and the smallest integer greater than $d_v/2$. Create four real parameters $\gamma_{12}$, $\delta_{12}$, $\pi_0$ and $\omega_0$ that satisfy the following inequalities
\begin{equation}\label{Bur5}
(1-\theta)\alpha N\le\gamma_{12}N+\delta_{12}N,
\end{equation}
\begin{equation}\label{Bur6}
0\le \gamma_{12}N\le \alpha N,
\end{equation}
\begin{equation}\label{Bur7}
0\le \pi_0(1-R) N\le \omega_0d_vN\le\alpha d_vN,
\end{equation}
\begin{equation}\label{Bur8}
\begin{split}
\beta_3(\alpha-\gamma_{12})N\le& \omega_0d_vN\\
\le&\min\left(\frac{d'}{d_c}\pi_0 d_vN,\gamma_{12}\beta_1N+d_v(\alpha-\gamma_{12})N\right),
\end{split}
\end{equation}
where $d'$ is the largest odd number which is less than or equal to $d_c$, and
\begin{equation}\label{Bur9}
0\le \delta_{12}N\beta_2 \le(\pi_0-\omega_0)d_vN.
\end{equation}
Define the following polynomials
\begin{equation}\label{Bur10}
  F_0(x)\overset{\Delta}{=}\sum_{j=0}^{\beta_0}\binom{d_v}{j}x^j,
\end{equation}
\begin{equation}\label{Bur11}
  F_1(x)\overset{\Delta}{=}\sum_{j=0}^{\beta_1}\binom{d_v}{j}x^j,
\end{equation}
\begin{equation}\label{Bur12}
  F_2(x)\overset{\Delta}{=}\sum_{j=\beta_2}^{d_v}\binom{d_v}{j}x^j,
\end{equation}
\begin{equation}\label{Bur13}
  F_3(x)\overset{\Delta}{=}\sum_{j=\beta_3}^{d_v}\binom{d_v}{j}x^j,
\end{equation}
\begin{equation}\label{Bur14}
  G_o(x)\overset{\Delta}{=}\sum_{j=1,3,\dots,d'}\binom{d_c}{j}x^j,
\end{equation}
\begin{equation}\label{Bur15}
  G_e(x)\overset{\Delta}{=}\sum_{j=0,2,\dots,d''}^{c}\binom{d_c}{j}x^j,
\end{equation}
where $d''$ is the largest even number less than or equal to $d_c$. Then we define
\begin{equation}\label{Bur16}
\begin{split}
  \psi(\alpha,\gamma_{12},\delta_{12},\pi_0,\omega_0)\overset{\Delta}{=}h(\gamma_{12},\alpha-\gamma_{12},\delta_{12})+(1-R)h(\pi_0)\\
  +t_1+t_2+u_1+u_2-d_vh(\omega_0,\alpha-\omega_0,\pi_0-\omega_0),
\end{split}
\end{equation}
where $h(\cdot)$ is the entropy function defined as
\begin{equation}
\begin{split}
  h(\tau_1,\tau_2,\dots,\tau_i)=&-\sum_{j=1}^i\tau_j\log \tau_j\\
  &-\left(1-\sum_{j=1}^i\tau_j\right)\log\left(1-\sum_{j=1}^i\tau_j\right),
\end{split}
\end{equation}
and
\begin{equation}\label{Bur17}
  t_1=\inf_{x>0}\left\{\gamma_{12}\log F_1+(\alpha-\gamma_{12})\log F_3-\omega_0d_v\log x\right\},
\end{equation}
\begin{equation}\label{Bur18}
\begin{split}
  t_2=\inf_{x>0}\left\{\delta_{12}\log F_2+(1-\alpha-\delta_{12})\log F_0\right.\\
  \left.-(\pi_0-\omega_0)d_v\log x\right\},
\end{split}
\end{equation}
\begin{equation}\label{Bur19}
  u_1=\inf_{x>0}\left\{\pi_0(1-R)\log G_o -\omega_0d_v\log x\right\},
\end{equation}
\begin{equation}\label{Bur20}
  u_2=\inf_{x>0}\left\{(1-\pi_0)(1-R)\log G_e -(\alpha-\omega_0)d_v\log x\right\}.
\end{equation}
The base of all the logarithms is $e$. Then, Theorem 1 of \cite{burshtein2008error} and the last paragraph on page 521 of \cite{burshtein2008error} implies the following result:
\begin{lemma}(\cite[Theorem 1]{burshtein2008error})\label{Bur_lemma}
Consider the random ensemble of ($d_v,d_c$)-regular LDPC codes with $d_v>4$ and block length $N$. Let $\alpha_0$ be the smallest positive root of the function $f(\alpha)$ which is defined by
\begin{equation}\label{falpha}
  f(\alpha)=\max_{\gamma_{12},\delta_{12},\pi_0,\omega_0}\psi(\alpha,\gamma_{12},\delta_{12},\pi_0,\omega_0),
\end{equation}
where the maximization is over all values of $\gamma_{12},\delta_{12},\pi_0,\omega_0$ that satisfy \eqref{Bur5}-\eqref{Bur9}. Then, for any $\bar{\alpha}_0<\alpha_0$, if $N$ is sufficiently large, then except  for almost all codes in this ensemble can correct at least $\theta\bar{\alpha}_0 N$ errors out of any arbitrary $\bar{\alpha}_0 N$ errors using one iteration of the PBF algorithm.
\end{lemma}
\begin{IEEEproof}
Here we briefly summarize the proof in \cite{burshtein2008error}. Denote by $\bar{p}_e(\bar{\alpha}_0 N)$ the fraction of (bad) codes in the ($d_v,d_c$)-regular ensemble that cannot correct a linear fraction $\theta\bar{\alpha}_0N$ of all combinations of $\bar{\alpha}_0N$ errors or less using one iteration of the PBF algorithm. Then, according to (38) in \cite{burshtein2008error}, $\bar{p}_e(\alpha_0 N)$ is upper-bounded by
\begin{equation}\label{Bur_thm}
  \bar{p}_e(\bar{\alpha}_0 N)\le \sum_{\alpha N\le\bar{\alpha}_0N} C(\alpha N)^{11/2}e^{Nf(\alpha)},
\end{equation}
where the summation is over all integer values of $\alpha N\le \bar{\alpha}_0N$, and $C=(2\pi)^{3/2}e^{1/3}\frac{d_v^{9/2}d_c^{3/2}}{\beta_2}$. Therefore, when $\bar{\alpha}_0$ is sufficiently small so that $f(\alpha)<0$ for all $\alpha<\bar{\alpha}_0$, $\bar{p}_e(\bar{\alpha}_0 N)\to 0$ as $N\to\infty$, which means that almost all codes in the ($d_v,d_c$)-regular ensemble can correct $\theta$ fraction of all possible combinations of $\bar{\alpha}_0N$ errors using one iteration of the PBF algorithm.
\end{IEEEproof}
\textcolor{black}{Theorem 1 in \cite{burshtein2008error} was stated for $\theta=0$ and the original constraint corresponding to the constraint \eqref{Bur5} ($(1-\theta)\alpha N\le\gamma_{12}N+\delta_{12}N$) was $\alpha N\le\gamma_{12}N+\delta_{12}N$. In this paper,we use the result for $\theta=\text{constant}>0$. This result can be obtained by directly changing the original constraint $\alpha N\le\gamma_{12}N+\delta_{12}N$ in \cite{burshtein2008error} to the new constraint $(1-\theta)\alpha N\le\gamma_{12}N+\delta_{12}N$ (this direct change is also stated at the bottom of page 521 in \cite{burshtein2008error} after the proof of Theorem 1). A refined bound for \eqref{Bur_thm} can be obtained using (22)(23)(25) and (33) in \cite{burshtein2008error}, which shows
\begin{equation}\label{Bur_thm_1}
\begin{split}
  &\bar{p}_e(\bar{\alpha}_0 N)\le \sum_{\alpha N\le\bar{\alpha}_0N} \left( \sum_{\gamma_{12}N,\delta_{12}N,\pi_0(1-R)N,\omega_0d_vN} \right.\\
  &\left.
  (2\pi Nd_v)^{3/2}e^{1/3}\sqrt{\omega_0(\alpha-\omega_0)(\pi_0-\omega_0)}e^{N\psi(\alpha,\gamma_{12},\delta_{12},\pi_0,\omega_0)}\right),
\end{split}
\end{equation}
where the outer summation is over all integer values of $\alpha N\le \bar{\alpha}_0N$, and the inner summation is over all integer values of $\gamma_{12}N,\delta_{12}N,\pi_0(1-R)N,\omega_0d_vN$ that satisfy \eqref{Bur5} to \eqref{Bur9}. We will use this refined bound to obtain finite-length result in the following example.}
\begin{example}\label{example1}
One example of the parameter choice is $d_v=9$, $d_c=18$ and $\theta=0.15$. In this case, we computed the first positive root of $f(\alpha)=0$ using MATLAB and obtained $\alpha=5.1\cdot 10^{-4}$. This means that using one iteration of the PBF algorithm, we can correct a fraction $\theta=0.15$ of $5.1\cdot10^{-4}\cdot N$ worst-case errors using a $(9,18)$ regular LDPC code when $N$ is sufficiently large. We can also use this result to obtain finite-length bounds (computing an upper bound on the fraction of bad codes using \eqref{Bur_thm_1}). We obtained that at least 4.86\% of $(9,18)$ regular LDPC codes of length $N=50,000$ in the random LDPC ensemble can reduce the number of errors by 15\% using one iteration of the PBF algorithm, when the number of errors is smaller than or equal to 20, which corresponds to the case when $\alpha_0=0.0004$.
\end{example}
The existence of codes that satisfy requirement (A.3) can also be established using Expander LDPC codes. Here, we review some results on expander LDPCs~\cite{Sip_FCS_96}.
\begin{definition}\label{def1} (Expander Graph)
An $(N,P,d_v,\gamma,\alpha)$ bipartite expander is a $d_v$-left-regular bipartite graph $\mathcal{G}(\mathcal{V}_L\cup \mathcal{V}_R,\mathcal{E})$ where $|\mathcal{V}_L|=N$ and $|\mathcal{V}_R|=P$. In this bipartite graph, it holds that $\forall \mathcal{S}\subset \mathcal{V}_L$ with $|\mathcal{S}|\le \gamma N$, $\mathcal{N}(\mathcal{S})\ge \alpha d_v|\mathcal{S}|$, where $\mathcal{N}(\mathcal{S})$ denotes the neighborhood of the set $\mathcal{S}$, i.e., the set of nodes in $\mathcal{V}_{R}$ connected to $\mathcal{S}$.
\end{definition}
An $(N,P,d_v,\gamma,\alpha)$ expander LDPC code is a length-$N$ LDPC code, where the Tanner graph of the code is the corresponding expander graph with $\mathcal{V}_{L}$ corresponding to the set of variable nodes and $\mathcal{V}_{R}$ the parity check nodes. We use $d_c=d_vN/P$ to denote the right-degree of the expander code.
\begin{lemma}\label{expander_fraction}(\cite[Thm11]{Sip_FCS_96})
Using an $(N,P,d_v,\gamma,\frac{3}{4}+\epsilon_e)$ regular expander LDPC code with parity check node degree $d_c=d_vN/P$, one can use one iteration of noiseless PBF algorithm to bring the fraction of errors down from $\alpha$ to $(1-4\epsilon_e)\alpha$ provided that the original corrupted codeword has at most $\gamma(1+4\epsilon_e)/2$ fraction of errors.
\end{lemma}
\begin{example}
The construction of a good Expander code has been investigated for a long time. Constructive approaches for Expander codes can be found in \cite{capalbo2002randomness,guruswami2009unbalanced}. In \cite{burshtein2001expander,feldman2007lp,richardson2008modern}, it is shown that random regular LDPC codes are expanders with high probability when the code length $N\to\infty$. In \cite[Theorem 8.7]{richardson2008modern} it is shown that, suppose $\gamma_\text{max}$ is the positive solution of the equation
\begin{equation}
\begin{split}
  &\frac{d_v-1}{d_v}h_2(\gamma)-\frac{1}{d_c}h_2(\gamma d_c*(3/4+\epsilon_e))\\
  &-\gamma(3/4+\epsilon_e)d_ch_2\left(\frac{1}{(3/4+\epsilon_e)d_c}\right)=0,
\end{split}
\end{equation}
then, for $3/4+\epsilon_e<\frac{d_v-1}{d_v}$ and $\gamma\in (0,\gamma_\text{max})$, a random regular $(d_v,d_c)$ LDPC Tanner graph is a $(d_v,d_c,\gamma,\frac{3}{4}+\epsilon_e)$ \textcolor{black}{expander with probability $1-O(N^{-\beta})$, where $\beta=d_v\left[1-(3/4+\epsilon_e)\right]-1$ is a constant greater than 0 when $3/4+\epsilon_e<\frac{d_v-1}{d_v}$} (which means that all sets of left nodes with cardinality smaller than $\gamma N$ have an expansion factor at least $\frac{3}{4}+\epsilon_e$). For $d_v=16$, $d_c=32$, and $\epsilon_e=0.0375$ (which is equivalent to $4\epsilon_e=0.15$, the same as $\theta=0.15$ in Example~\ref{example1}), we use MATLAB to numerically solve the above equation and obtained $\gamma_\text{max}\approx 4.1*10^{-5}$, which means the fraction of errors $\alpha$ can be as large as $\gamma_\text{max}(1+4\epsilon_e)/2=2.3575\cdot 10^{-5}$.
\end{example}
\section{Proof of Theorem~\ref{thm4}}\label{vs_app}
We tune the energy supply such that
\begin{equation}
  \max(p_\text{and},p_\text{xor},p_\text{maj})\le\lambda=\frac{\theta\alpha_0/2}{\left[d_c(1-R)+1 \right]+2},
\end{equation}
is satisfied for the first $L-L_{vs}$ stages (first phase), which ensures that
\begin{align}
   &{{p}_{\text{and}}}+\left[ d_c(1-R)+1 \right]{{p}_{\text{xor}}}+{{p}_\text{maj}}\le\theta {{\alpha }_{0}}/2,
\end{align}
is satisfied. We tune the energy supply such that
\begin{equation}
  \max (p_{\text{and}}^{(i+1)},p_{\text{xor}}^{(i+1)},p_{\text{maj}}^{(i+1)})\le \lambda^{(i+1)}=\frac{\theta {{\alpha }_{0}}{{\left( 1-\frac{1}{2}\theta  \right)}^{i}}/4}{\left[ {{d}_{c}}(1-R)+1 \right]+2},
\end{equation}
is satisfied for the last $L_{vs}$ stages (second phase), which ensures that
\begin{equation}\label{vs_condition_appendix}
  [{{d}_{c}}(1-R)+1]p_{\text{xor}}^{(i+1)}+p_{\text{maj}}^{(i+1)}+p_\text{and}^{(i+1)}\le \frac{1}{4}\theta\alpha_0{\left(1-\frac{1}{2}\theta\right)^{i}},
\end{equation}
is satisfied (we have mentioned this in~\eqref{vs_condition}). Since this version of ENCODED-V technique with dynamic voltage scaling has the same procedure and constant supply energy during the first $L-L_\text{vs}$ stages (first phase) as the ENCODED-F technique, from Theorem~\ref{expander_probabilistic}, we know that after the first $(L-L_\text{vs})$ stages, the output error fraction is smaller than $\alpha_0$ with probability at least $1-P_{e}^\text{blk}$, where $P_{e}^\text{blk}<3(L-L_\text{vs})\exp\left( -\lambda^*N \right)$ and $\lambda^*$ is defined in~\eqref{expander_error_exp}.

We will prove that, after the $i$-th stage of the remaining $L_\text{vs}$ stages, the error fraction is upper bounded by
\begin{equation}\label{vs_induction}
  \alpha _{\text{PBF}}^{(i)}\le \alpha_0{{(1-\theta/2)}^{i}},
\end{equation}
with high probability. Thus, after $L_\text{vs}$ iterations, we obtain
\begin{equation}\label{vs_induction_2}
  \alpha _{\text{PBF}}^{(L_\text{vs})}\le \alpha_0{{(1-\theta/2)}^{L_\text{vs}}}\le p_\text{tar},
\end{equation}
where the last step can be verified by plugging in~\eqref{vs_A}.

The case for $i=0$ is already true as argued above. Suppose~\eqref{vs_induction} holds for some $i\ge 0$, then, we prove~\eqref{vs_induction} also holds for the $(i+1)$-th stage of the second phase. Note that from~\eqref{vs_condition_appendix}, the probability that the number of new errors introduced during the PBF decoding at the $(i+1)$-th stage, which is $[d_c(1-R)+1]\alpha _{\text{xor}}^{(i+1)}+\alpha _{\text{maj}}^{(i+1)}+\alpha _{\text{and}}^{(i+1)}$, satisfies
\begin{equation}\label{vs_der_1}
\begin{split}
  \Pr &\left( [d_c(1-R)+1]\alpha _{\text{xor}}^{(i+1)}+\alpha _{\text{maj}}^{(i+1)}+\alpha _{\text{and}}^{(i+1)}\right.\\
  &\left.\quad>\frac{1}{2}\alpha_0\theta{{(1-\theta/2)}^{i}} \right)\\
  \overset{(a)}{<}\Pr &\left([d_c(1-R)+1]\alpha _{\text{xor}}^{(i+1)}+\alpha _{\text{maj}}^{(i+1)}+\alpha _{\text{and}}^{(i+1)}\right.\\
  &>[d_c(1-R)+1]p_{\text{xor}}^{(i+1)}+p_{\text{maj}}^{(i+1)}+p_{\text{and}}^{(i+1)}\\
  &\left.\qquad+\frac{1}{4}\alpha_0\theta{{(1-\theta/2)}^{i}}\right)\\
  {<}\Pr &\left( \alpha _{\text{and}}^{(i+1)}>p_{\text{and}}^{(i+1)}+\lambda^{(i+1)} \right)\\
  &+\Pr \left( \alpha _{\text{xor}}^{(i+1)}>p_{\text{xor}}^{(i+1)}+\lambda^{(i+1)} \right)\\
  &+\Pr \left( \alpha _{\text{maj}}^{(i+1)}>p_{\text{maj}}^{(i+1)}+\lambda^{(i+1)} \right)\\
  \overset{(b)}{<}3&\exp \left( -{{\widetilde{\lambda }}^{(i+1)}}N \right),
\end{split}
\end{equation}
where step (a) follows from~\eqref{vs_condition_appendix}, step (c) follows from the large deviation bound in Lemma~\ref{Binomial_deviation} and ${{\widetilde{\lambda }}^{(i+1)}}$ is defined in~\eqref{vs_condition_2}. Therefore, with probability at least $1-3\exp \left( -{{\widetilde{\lambda }}^{(i+1)}}N \right)$,
\begin{equation}\label{vs_condition_alpha}
  \begin{split}
   \alpha _{\text{PBF}}^{(i+1)}&\le \alpha _{\text{PBF}}^{(i)}(1-\theta)+[{{d}_{c}}(1-R)+1]\alpha _{\text{xor}}^{(i+1)}\\
   &\quad+\alpha _{\text{maj}}^{(i+1)}+\alpha _{\text{and}}^{(i+1)} \\
 & \overset{(a)}{\le} \alpha_0{{(1-\theta/2)}^{i}}(1-\theta)+\frac{1}{2}\alpha_0\theta{{(1-\theta/2)}^{i}}\\
 &=\alpha_0(1-\theta/2)^{i+1},
\end{split}
\end{equation}
where step (a) can be obtained by combining~\eqref{vs_condition_alpha} and~\eqref{vs_induction}. Now that we have proved \eqref{vs_induction} for the $(i+1)$-th stage, we can carry out the math induction for all $i$ that satisfies $1\le i\le L_\text{vs}$. If \eqref{vs_induction} holds for all $i$, the final error fraction is smaller than $p_\text{tar}$. Thus, the overall probability that the final error fraction is greater than $p_\text{tar}$ is upper bounded by the summation of $3(L-L_\text{vs})\exp\left( -\lambda^*N \right)$ in the first $L-L_\text{vs}$ stages and the RHS of \eqref{vs_der_1} for the last $L_\text{vs}$ stages, which is
\begin{equation}\label{vs_error_prob_app}
  P_{e}^\text{blk}<3(L-L_\text{vs})\exp\left( -\lambda^*N \right)+3\mathop\sum\limits_{i=1}^{L_\text{vs}}\exp \left( -{{\widetilde{\lambda }}^{(i+1)}}N \right).
\end{equation}
Thus,~\eqref{vs_error_prob} is proved.

Finally, we compute the overall energy consumption. The energy consumed in the $i$-th stage can be written as
\begin{equation}\label{energy_i}
  {{E}_{i}}=N\epsilon_{\text{and}}^{-1}\left( p_{\text{and}}^{(i)} \right)+N\epsilon_{\text{maj}}^{-1}\left( p_{\text{maj}}^{(i)} \right)+(N+P)\epsilon_{\text{xor}}^{-1}\left( p_{\text{xor}}^{(i)} \right).
\end{equation}
By summing over all stages both in the first phase and the second phase and normalizing by the number of outputs $K$, the total energy consumption per output bit can be written as in~\eqref{vs_energy}.
\section{Proof of Corollary~\ref{vs_coro}}\label{vs_app_2}
We choose
\begin{equation}\label{vs_1}
\begin{split}
  p_\text{and}=p_\text{xor}=p_\text{maj}=\lambda=\frac{\theta\alpha_0/2}{\left[d_c(1-R)+1 \right]+2},
\end{split}
\end{equation}
in the first $L-L_\text{vs}$ stages and
\begin{equation}\label{vs_2}
\begin{split}
  p_{\text{and}}^{(i+1)}=p_{\text{xor}}^{(i+1)}=p_{\text{maj}}^{(i+1)}= \lambda^{(i+1)}=\frac{\theta {{\alpha }_{0}}{{\left( 1-\frac{1}{2}\theta  \right)}^{i}}/4}{\left[ {{d}_{c}}(1-R)+1 \right]+2},
\end{split}
\end{equation}
in the $i$-th stage of the last $L_\text{vs}$ stages (defined in~\eqref{vs_A}).

By plugging in~\eqref{vs_1}, \eqref{vs_2} and ${{L}_{\text{vs}}}=\left\lceil \frac{\log \frac{2}{{{p}_{\text{tar}}}}+\log {{\alpha }_{0}}}{\log \frac{1}{1-\frac{1}{2}\theta }} \right\rceil $ into the error probability expression~\eqref{vs_error_prob}, we know that the ENCODED-V technique has output error fraction smaller than $\alpha_0{{(1-\theta/2)}^{L_\text{vs}}}\le \frac{1}{2}p_\text{tar}$ with probability at least $1-P_e^\text{blk}$, where $P_e^\text{blk}$ satisfies
\begin{equation}
\begin{split}
  P_{e}^\text{blk}<&3(L-L_\text{vs})\exp\left(-\lambda^*N \right)
  +3L_\text{vs}\exp(-\widetilde{\lambda}^{(L_\text{vs}+1)}N),
\end{split}
\end{equation}
where ${{\widetilde{\lambda }}^{(i+1)}}=D(2{{\lambda }^{(i+1)}}\|{{\lambda }^{(i+1)}})=\left( 2\log 2-1 \right){{\lambda }^{(i+1)}}+\mathcal{O}({{({{\lambda }^{(i+1)}})}^{2}})$ and ${{\lambda }^{*}}=D(2\lambda\|\lambda)=\left( 2\log 2-1 \right)\lambda+\mathcal{O}({{\lambda}^{2}})$. Since $\lambda^{(L_\text{vs}+1)}=\frac{\theta {{\alpha }_{0}}{{\left( 1-\frac{1}{2}\theta  \right)}^{L_\text{vs}}}/4}{\left[ {{d}_{c}}(1-R)+1 \right]+2}$ and $\alpha_0{{(1-\theta/2)}^{L_\text{vs}-1}}> \frac{1}{2}p_\text{tar}$, we have $\lambda^{(L_\text{vs}+1)}>\frac{\theta {p_\text{tar}{\left( 1-\frac{1}{2}\theta  \right)}}/8}{\left[ {{d}_{c}}(1-R)+1 \right]+2}$. Therefore,
\begin{equation}
  P_{e}^\text{blk}<3L\exp(-\theta^* N),
\end{equation}
where ${{\theta }^{*}}=\min \left\{\lambda^*,D\left(2\frac{\theta {p_\text{tar}{\left( 1-\frac{1}{2}\theta  \right)}}/4}{\left[ {{d}_{c}}(1-R)+1 \right]+2}\left\|\frac{\theta {p_\text{tar}{\left( 1-\frac{1}{2}\theta  \right)}}/4}{\left[ {{d}_{c}}(1-R)+1 \right]+2}\right.\right)\right\}$.

Denote the output error fraction by $\delta_e^\text{frac}$, which is a random variable supported on $[0,1]$. We know that $\Pr(\delta_e^\text{frac}>\frac{1}{2}p_\text{tar})<P_e^\text{blk}$. Thus, the output bit error probability is upper bounded by
\begin{equation}
\begin{split}
  \mathbb{E}[\delta_e^\text{frac}]&<\Pr(\delta_e^\text{frac}>\frac{1}{2}p_\text{tar})\mathbb{E}\left[\delta_e^\text{frac}|\delta_e^\text{frac}>\frac{1}{2}p_\text{tar}\right]\\
  &+\Pr(\delta_e^\text{frac}\le\frac{1}{2}p_\text{tar}))\mathbb{E}\left[\delta_e^\text{frac}|\delta_e^\text{frac}\le \frac{1}{2}p_\text{tar}\right]\\
  &<P_e^\text{blk}+\frac{1}{2}p_\text{tar}.
\end{split}
\end{equation}
When $N>\frac{1}{{{\theta }^{*}}}\log \left( \frac{6L}{{{p}_{\text{tar}}}} \right)$, $P_e^\text{blk}<\frac{1}{2}p_\text{tar}$, and hence the output bit error probability $\mathbb{E}[\delta_e^\text{frac}]$ satisfies $\mathbb{E}[\delta_e^\text{frac}]<p_\text{tar}$.

In this corollary, we only examine the case when $\epsilon_\text{and}(u)=\epsilon_\text{xor}(u)=\epsilon_\text{maj}(u)=\epsilon(u)$ (either polynomial decay or exponential decay). We also choose the same gate error probabilities $p_\text{and}=p_\text{xor}=p_\text{maj}=\lambda$ and $p_{\text{and}}^{(i+1)}=p_{\text{xor}}^{(i+1)}=p_{\text{maj}}^{(i+1)}= \lambda^{(i+1)}$ for different types of unreliable gates (see~\eqref{vs_1} and \eqref{vs_2}). Therefore, the energy consumption~\eqref{vs_energy} is simplified to
\begin{equation}
  {{E}_{\text{per-bit}}}\le\frac{3N+P}{K}\left( L-{{L}_{\text{vs}}} \right){{\epsilon }^{-1}}\left( \lambda  \right)\text{+}\frac{3N+P}{K}\sum\limits_{i=1}^{{{L}_{\text{vs}}}}{{{\epsilon }^{-1}}\left( \lambda^{(i)}  \right)},
\end{equation}
where ${{L}_{\text{vs}}}=\left\lceil \frac{\log \frac{2}{{{p}_{\text{tar}}}}+\log {{\alpha }_{0}}}{\log \frac{1}{1-\frac{1}{2}\theta }} \right\rceil$.

When the energy-reliability tradeoff function $\epsilon_\text{and}(u)=\epsilon_\text{xor}(u)=\epsilon_\text{maj}(u)=(\frac{1}{u})^{c},c>0$, the total energy consumption per bit
\begin{equation}
\begin{split}
  {{E}_{\text{per-bit}}}&\le\frac{3N+P}{K}{{\lambda }^{-\frac{1}{c}}}\left[ \left( L-{{L}_{\text{vs}}} \right)+\frac{{{\left( 1-\frac{1}{2}\theta  \right)}^{-\frac{{{L}_{\text{vs}}}}{c}}}-1}{{{\left( 1-\frac{1}{2}\theta  \right)}^{-\frac{1}{c}}}-1} \right]\\
  &\le\frac{3N+P}{K}{{\lambda }^{-\frac{1}{c}}}\left[ \left( L-{{L}_{\text{vs}}} \right)\right.\\
  &\qquad\qquad\left.+\frac{{{\left( 1-\frac{1}{2}\theta  \right)}^{-\frac{1}{c}}}{{\left( \frac{{{p}_{\text{tar}}}}{2{{\alpha }_{0}}} \right)}^{-\frac{1}{c}}}-1}{{{\left( 1-\frac{1}{2}\theta  \right)}^{-\frac{1}{c}}}-1} \right]\\
  &=\Theta \left( \frac{N}{K}\max\left\{L,\left(\frac{1}{p_\text{tar}}\right)^{\frac{1}{c}}\right\} \right).
\end{split}
\end{equation}

When the energy-reliability tradeoff function $\epsilon_\text{and}(u)=\epsilon_\text{xor}(u)=\epsilon_\text{maj}(u)=\exp(-c u),c>0$, the total energy consumption per bit
\begin{equation}
\begin{split}
  {{E}_{\text{per-bit}}}&=\frac{3N+P}{cK}\left( L\log \frac{1}{\lambda }+\frac{1}{2}{{L}_{\text{vs}}}\left( {{L}_{\text{vs}}}+1 \right)\log \frac{1}{1-\frac{1}{2}\theta } \right)\\
  &=\Theta\left(\frac{N}{K}\max\left\{L,\log^2\frac{1}{p_\text{tar}}\right\}\right).
\end{split}
\end{equation}

When the energy-reliability tradeoff function $\epsilon_\text{and}(u)=\epsilon_\text{xor}(u)=\epsilon_\text{maj}(u)=\exp(-c \sqrt{u}),c>0$, the total energy consumption per bit
\begin{equation}
\begin{split}
  & {{E}_{\text{per-bit}}}=\frac{3N+P}{K}\left( L-{{L}_{\text{vs}}} \right){{\left( \frac{1}{c}\log \frac{1}{\lambda } \right)}^{2}}\\
  &\qquad+\frac{3N+P}{K}\sum\limits_{i=1}^{{{L}_{\text{vs}}}}{{{\left( \frac{1}{c}\log \frac{1}{\lambda }+\frac{i-1}{c}\log \frac{1}{1-\frac{1}{2}\theta } \right)}^{2}}} \\
 & =\frac{3N+P}{K}\left[L{{\left( \frac{1}{c}\log \frac{1}{\lambda } \right)}^{2}}\right.\\
 &\qquad\qquad+\frac{2}{c}\log \frac{1}{\lambda }\cdot \frac{1}{c}\log \frac{1}{1-\frac{1}{2}\theta }\frac{({{L}_{\text{vs}}}-1){{L}_{\text{vs}}}}{2}\\
 &\qquad\qquad\left.+\frac{1}{{{c}^{2}}}{{\log }^{2}}\frac{1}{1-\frac{1}{2}\theta }\frac{1}{6}({{L}_{\text{vs}}}-1){{L}_{\text{vs}}}(2{{L}_{\text{vs}}}-1)\right] \\
 & =\Theta\left(\frac{N}{K}\max\left\{L,\log^3\frac{1}{p_\text{tar}}\right\}\right).
\end{split}
\end{equation}

\bibliographystyle{ieeetr}
\bibliography{rough}

\begin{thebibliography}{100}

\bibitem{yang2016computing}
Y.~Yang, P.~Grover, and S.~Kar, ``Computing linear transforms with unreliable
  components,'' in {\em Proc. IEEE Int. Symp. Inf. Theory}, pp.~1934--1938,
  IEEE, 2016.

\bibitem{Bor_Micro_05}
S.~Borkar, ``Designing reliable systems from unreliable components: the
  challenges of transistor variability and degradation,'' {\em IEEE Micro},
  vol.~25, no.~6, pp.~10--16, 2005.

\bibitem{Shan_DTC_08}
N.~R. Shanbhag, S.~Mitra, G.~de~Veciana, M.~Orshansky, R.~Marculescu,
  J.~Roychowdhury, D.~Jones, and J.~M. Rabaey, ``The search for alternative
  computational paradigms,'' {\em IEEE Des. Test. Comput}, vol.~25, no.~4,
  pp.~334--343, 2008.

\bibitem{haque2010hard}
I.~S. Haque and V.~S. Pande, ``Hard data on soft errors: A large-scale
  assessment of real-world error rates in {GPGPU},'' in {\em Proc. IEEE/ACM
  Int. Conf. Cluster Cloud Grid Comput.}, pp.~691--696, IEEE Computer Society,
  2010.

\bibitem{Dennard}
R.~H. Dennard, V.~L. Rideout, E.~Bassous, and A.~R. Leblanc, ``Design of
  ion-implanted {MOSFET}'s with very small physical dimensions,'' {\em IEEE J.
  Solid-State Circuits}, vol.~9, no.~5, pp.~256--268, 1974.

\bibitem{Nat_CUP_14}
R.~Nathanael and T.-J.~K. Liu, {\em CMOS and Beyond: Logic Switches for
  Terascale Integrated Circuits}, ch.~11 Mechanical switches.
\newblock Cambridge University Press, 2014.

\bibitem{Pil_ACM_01}
P.~Pillai and K.~G. Shin, ``Real-time dynamic voltage scaling for low-power
  embedded operating systems,'' {\em ACM SIGOPS Operating Systems Review},
  vol.~35, no.~5, pp.~89--102, 2001.

\bibitem{Zhao_TR_07}
C.~Zhao, X.~Bai, and S.~Dey, ``Evaluating transient error effects in digital
  nanometer circuits,'' {\em IEEE Trans. Rel}, vol.~56, no.~3, pp.~381--391,
  2007.

\bibitem{Mir_Spec_12}
M.~Miranda, ``The threat of semiconductor variability: As transistors shrink,
  the problem of chip variability grows.''
  \url{http://spectrum.ieee.org/semiconductors/design/the-threat-of-semiconductor-variability},
  June 2012.

\bibitem{Shu_Nature_13}
M.~M. Shulaker, G.~Hills, N.~Patil, H.~Wei, H.~Chen, H.-S.~P. Wong, and
  S.~Mitra, ``Carbon nanotube computer,'' {\em Nature}, vol.~501, no.~7468,
  pp.~526--530, 2013.

\bibitem{Pat_TCAD_08}
N.~Patil, J.~Deng, A.~Lin, H.-S. Wong, and S.~Mitra, ``Design methods for
  misaligned and mispositioned carbon-nanotube immune circuits,'' {\em IEEE
  Trans. Comput.-Aided Design Integr. Circuits Syst.}, vol.~27, no.~10,
  pp.~1725--1736, 2008.

\bibitem{Sha_Bel_48}
C.~Shannon, ``A mathematical theory of communication,'' {\em Bell Syst. Tech.
  J.}, vol.~27, pp.~623--656, Oct 1948.

\bibitem{Neu_Aut_56}
J.~V. Neumann, ``Probabilistic logics and the synthesis of reliable organisms
  from unreliable components,'' {\em Automata studies}, vol.~34, pp.~43--98,
  1956.

\bibitem{Abda_JSSC_13}
R.~A. Abdallah and N.~R. Shanbhag, ``An energy-efficient ecg processor in 45-nm
  {CMOS} using statistical error compensation,'' {\em IEEE J. Solid-State
  Circuits}, vol.~48, no.~11, pp.~2882--2893, 2013.

\bibitem{Han_VLSI_14}
J.~Han, E.~Leung, L.~Liu, and F.~Lombardi, ``A fault-tolerant technique using
  quadded logic and quadded transistors,'' {\em IEEE Trans. Very Large Scale
  Integr. (VLSI) Syst.}, vol.~PP, no.~99, pp.~1--1, 2014.

\bibitem{Han_DTC_05}
J.~Han, J.~Gao, P.~Jonker, Y.~Qi, and J.~Fortes, ``Toward hardware-redundant,
  fault-tolerant logic for nanoelectronics,'' {\em IEEE Des. Test. Comput},
  vol.~22, pp.~328--339, July 2005.

\bibitem{Kim_ACMMi_03}
D.~Ernst, N.~S. Kim, S.~Das, S.~Pant, R.~Rao, T.~Pham, C.~Ziesler, D.~Blaauw,
  T.~Austin, K.~Flautner, and T.~Mudge, ``Razor: A low-power pipeline based on
  circuit-level timing speculation,'' in {\em Proc. 36th Int'l Symp.
  Microarchitecture}, pp.~7--18, 2003.

\bibitem{Cho_TCAD_12}
H.~Cho, L.~Leem, and S.~Mitra, ``Ersa: Error resilient system architecture for
  probabilistic applications,'' {\em IEEE Trans. Comput.-Aided Design Integr.
  Circuits Syst.}, vol.~31, no.~4, pp.~546--558, 2012.

\bibitem{Gal_TIT_62}
R.~G. Gallager, ``Low-density parity-check codes,'' {\em IRE Trans. Inf.
  Theory}, vol.~8, pp.~21--28, January 1962.

\bibitem{Sip_FCS_96}
M.~Sipser and D.~A. Spielman, ``Expander codes,'' {\em IEEE Trans. Inf.
  Theory}, vol.~42, no.~6, 1996.

\bibitem{Tay_Bel_68}
M.~G. Taylor, ``Reliable information storage in memories designed from
  unreliable components,'' {\em Bell Syst. Tech. J.}, vol.~47, no.~10,
  pp.~2299--2337, 1968.

\bibitem{Kuz_PIT_73}
A.~V. Kuznetsov, ``Information storage in a memory assembled from unreliable
  components,'' {\em Probl. Inf. Transm.}, vol.~9, no.~3, pp.~100--114, 1973.

\bibitem{Chi_ITW_07}
S.~K. Chilappagari and B.~Vasic, ``Fault tolerant memories based on expander
  graphs,'' in {\em Proc. IEEE Inf. Theory Workshop}, 2007.

\bibitem{Gun_IAC_08}
K.~Gunnam, G.~Choi, M.~Yeary, S.~Yang, and Y.~Lee, ``Next generation iterative
  {LDPC} solutions for magnetic recording storage,'' in {\em Proc. Asilomar
  Conf. Signal, Syst. Comput.}, pp.~1148--1152, Oct 2008.

\bibitem{huang2015acoco}
C.-H. Huang, Y.~Li, and L.~Dolecek, ``{ACOCO}: {A}daptive coding for
  approximate computing on faulty memories,'' {\em IEEE Trans. Commun.},
  vol.~63, no.~12, pp.~4615--4628, 2015.

\bibitem{Var_TIT_11}
L.~Varshney, ``Performance of {LDPC} codes under faulty iterative decoding,''
  {\em IEEE Trans. Inf. Theory}, vol.~57, pp.~4427--4444, July 2011.

\bibitem{Yaz_TC_01}
S.~M.~S. Tabatabaei, H.~Cho, and L.~Dolecek, ``Gallager {B} decoder on noisy
  hardware,'' {\em IEEE Trans. Commun.}, vol.~61, pp.~1660--1673, May 2013.

\bibitem{Huang_TC_14}
C.-H. Huang, Y.~Li, and L.~Dolecek., ``Gallager {B} {LDPC} decoder with
  transient and permanent errors,'' {\em IEEE Trans. Commun.}, vol.~62,
  pp.~15--28, January 2014.

\bibitem{ref_1}
A.~Balatsoukas-Stimming and A.~Burg, ``Density evolution for min-sum decoding
  of {LDPC} codes under unreliable message storage,'' {\em IEEE Commun. Lett.},
  vol.~18, no.~5, pp.~849--852, 2014.

\bibitem{ref_2}
E.~D. C.~Kameni~Ngassa, V.~Savin and D.~Declercq, ``Density evolution and
  functional threshold for the noisy min-sum decoder,'' {\em IEEE Trans.
  Commun.}, vol.~63, no.~5, pp.~1497--1509, 2015.

\bibitem{Ric_TIT_01_2}
T.~Richardson and R.~Urbanke, ``The capacity of low-density parity-check codes
  under message-passing decoding,'' {\em IEEE Trans. Inf. Theory}, vol.~47,
  pp.~599--618, Feb 2001.

\bibitem{Len_TIT_01}
M.~Lentmaier, D.~Truhachev, K.~Zigangirov, and D.~Costello, ``An analysis of
  the block error probability performance of iterative decoding,'' {\em IEEE
  Trans. Inf. Theory}, vol.~51, pp.~3834--3855, Nov 2005.

\bibitem{Had_TIT_05}
C.~Hadjicostis and G.~C. Verghese, ``Coding approaches to fault tolerance in
  linear dynamic systems,'' {\em IEEE Trans. Inf. Theory}, vol.~51,
  pp.~210--228, Jan 2005.

\bibitem{Spi_FCS_96}
D.~A. Spielman, ``Highly fault-tolerant parallel computation,'' in {\em Proc.
  Symp. Foundations Comput. Sci.}, pp.~154--163, Oct 1996.

\bibitem{Rac_ISIT_08}
E.~Rachlin and J.~E. Savage, ``A framework for coded computation,'' in {\em
  Proc. IEEE Int. Symp. Inf. Theory}, pp.~2342--2346, July 2008.

\bibitem{Bert_CAD_05}
D.~Bertozzi, L.~Benini, and G.~De~Micheli, ``Error control schemes for on-chip
  communication links: the energy-reliability tradeoff,'' {\em IEEE Trans.
  Comput.-Aided Design Integr. Circuits Syst.}, vol.~24, pp.~818--831, June
  2005.

\bibitem{Simon_ITW_11}
F.~Simon, ``On the capacity of noisy computations,'' in {\em Proc. IEEE Inf.
  Theory Workshop}, pp.~185--189, Oct 2011.

\bibitem{Simon_ITW_10}
F.~Simon, ``Capacity of a noisy function,'' in {\em Proc. IEEE Inf. Theory
  Workshop}, pp.~1--5, Aug 2010.

\bibitem{Koch1}
T.~Koch, A.~Lapidoth, and P.~Sotiriadis, ``Channels that heat up,'' {\em IEEE
  Trans. Inf. Theory}, vol.~55, pp.~3594--3612, Aug 2009.

\bibitem{Gro_JSAC_11}
P.~Grover, K.~Woyach, and A.~Sahai, ``Towards a communication-theoretic
  understanding of system-level power consumption,'' {\em IEEE J. Sel. Areas
  Commun}, vol.~29, pp.~1744--1755, September 2011.

\bibitem{Gro_ISIT_12}
P.~Grover, A.~Goldsmith, and A.~Sahai, ``Fundamental limits on the power
  consumption of encoding and decoding,'' in {\em Proc. IEEE Int. Symp. Inf.
  Theory}, pp.~2716--2720, IEEE, 2012.

\bibitem{Gro_TIT_15}
P.~Grover, ``Information friction and its implications on minimum energy
  required for communication,'' {\em IEEE Trans. Inf. Theory}, vol.~61,
  pp.~895--907, Feb 2015.

\bibitem{Blake1}
C.~Blake and F.~R. Kschischang, ``Energy of decoding algorithms,'' in {\em
  Canadian Workshop on Inf. Theory}, pp.~1--5, June 2013.

\bibitem{Blake2}
C.~G. Blake and F.~R. Kschischang, ``Energy consumption of {VLSI} decoders,''
  {\em IEEE Trans. Inf. Theory}, vol.~61, pp.~3185--3198, June 2015.

\bibitem{Gro_ISIT_14}
P.~Grover, ``Is `shannon-capacity of noisy computing' zero?,'' in {\em Proc.
  IEEE Int. Symp. Inf. Theory}, pp.~2854--2858, June 2014.

\bibitem{Kar_TDSC_04}
T.~Karnik and P.~Hazucha, ``Characterization of soft errors caused by single
  event upsets in {CMOS} processes,'' {\em IEEE Trans. Dependable Secure
  Comput.}, vol.~1, pp.~128--143, April 2004.

\bibitem{Dob_PPI_77}
R.~L.~V. Dobrushin and S.~I. Ortyukov, ``Upper bound on the redundancy of
  self-correcting arrangements of unreliable functional elements,'' {\em Probl.
  Inf. Transm.}, vol.~13, no.~3, pp.~56--76, 1977.

\bibitem{Pip_FOC_85}
N.~Pippenger, ``On networks of noisy gates,'' in {\em Proc. Symp. Foundations
  Comput. Sci.}, pp.~30--38, IEEE, 1985.

\bibitem{Pip_TIT_91}
N.~Pippenger, G.~Stamoulis, and J.~Tsitsiklis, ``On a lower bound for the
  redundancy of reliable networks with noisy gates,'' {\em IEEE Trans. Inf.
  Theory}, vol.~37, pp.~639--643, May 1991.

\bibitem{Had_TAC_03}
C.~N. Hadjicostis, ``Nonconcurrent error detection and correction in
  fault-tolerant linear finite-state machines,'' {\em IEEE Trans. Autom.
  Control}, vol.~48, no.~12, pp.~2133--2140, 2003.

\bibitem{laurenciu2016error}
N.~C. Laurenciu, T.~Gupta, V.~Savin, and S.~Cotofana, ``Error correction code
  protected data processing units,'' in {\em Proc ACM/IEEE NANOARCH},
  pp.~37--42, IEEE, 2016.

\bibitem{Yang_All_14}
Y.~Yang, P.~Grover, and S.~Kar, ``Can a noisy encoder be used to communicate
  reliably?,'' in {\em Proc. Allerton Conf. on Commun., Control and Comput.},
  pp.~659--666, Sept 2014.

\bibitem{dupraz2016practical}
E.~Dupraz, V.~Savin, S.~K. Grandhi, E.~Popovici, and D.~Declercq, ``Practical
  {LDPC} encoders robust to hardware errors,'' in {\em 2016 IEEE International
  Conference on Communications (ICC),}, pp.~1--6, IEEE, 2016.

\bibitem{Hac_ISIT_13}
J.~Hachem, I.-H. Wang, C.~Fragouli, and S.~Diggavi, ``Coding with encoding
  uncertainty,'' in {\em Proc. IEEE Int. Symp. Inf. Theory}, pp.~276--280, July
  2013.

\bibitem{Huang_TC_84}
K.-H. Huang and J.~Abraham, ``Algorithm-based fault tolerance for matrix
  operations,'' {\em IEEE Trans. Comput.}, vol.~C-33, pp.~518--528, June 1984.

\bibitem{Wang_TC_94}
S.-J. Wang and N.~K. Jha, ``Algorithm-based fault tolerance for fft networks,''
  {\em IEEE Trans. Comput.}, vol.~43, pp.~849--854, Jul 1994.

\bibitem{Ding_ISPA_11}
C.~Ding, C.~Karlsson, H.~Liu, T.~Davies, and Z.~Chen, ``Matrix multiplication
  on {GPU}s with on-line fault tolerance,'' in {\em Proc. IEEE Int'l Symp.
  Parallel Distrib. Process. Appl.}, pp.~311--317, May 2011.

\bibitem{Anf_TC_88}
C.~Anfinson and F.~Luk, ``A linear algebraic model of algorithm-based fault
  tolerance,'' {\em IEEE Trans. Comput.}, vol.~37, pp.~1599--1604, Dec 1988.

\bibitem{Chen_HPCN_09}
Z.~Chen, ``Optimal real number codes for fault tolerant matrix operations,'' in
  {\em Proc. Conf. High Performance Computing Networking Storage and Analysis},
  SC '09, (New York, NY, USA), pp.~29:1--29:10, ACM, 2009.

\bibitem{Radha_Asil_13}
C.~Radhakrishnan and A.~C. Singer, ``Recursive least squares filtering under
  stochastic computational errors,'' in {\em Proc. Asilomar Conf. Signal, Syst.
  Comput.}, pp.~1529--1532, 2013.

\bibitem{Bowman_ADAC_07}
K.~Bowman, J.~Tschanz, C.~Wilkerson, T.~K. S.-L.~Lu, V.~De, and S.~Borkar,
  ``Circuit techniques for dynamic variation tolerance,'' in {\em Proc. 46th
  ACM/IEEE DAC}, pp.~4--7, 2007.

\bibitem{Choi_TSP_07}
J.~W. Choi, B.~Shim, A.~Singer, and N.~I. Cho, ``Low-power filtering via
  minimum power soft error cancellation,'' {\em IEEE Trans. Signal Process.},
  vol.~55, pp.~5084--5096, Oct 2007.

\bibitem{Had_CDC_01}
C.~N. Hadjicostis, ``Non-concurrent error detection and correction in
  discrete-time {LTI} dynamic systems,'' in {\em Proc. 40th IEEE Conf. Decision
  and Control}, vol.~2, pp.~1899--1904, 2001.

\bibitem{nahlus2014energy}
I.~Nahlus, E.~P. Kim, N.~R. Shanbhag, and D.~Blaauw, ``Energy-efficient dot
  product computation using a switched analog circuit architecture,'' in {\em
  Proc. Int. Symp. Low Power Electron. Des.}, pp.~315--318, ACM, 2014.

\bibitem{Goor_DTC_93}
A.~van~de Goor, ``Using march tests to test {SRAM}s,'' {\em IEEE Des. Test.
  Comput}, vol.~10, pp.~8--14, March 1993.

\bibitem{Kim_TC_98}
I.~Kim, Y.~Zorian, G.~Komoriya, H.~Pham, F.~P. Higgins, and J.~L. Lewandowski,
  ``Built in self repair for embedded high density {SRAM},'' in {\em Proc. IEEE
  Int'l Test Conf.}, pp.~1112--1119, Oct 1998.

\bibitem{AhlswedeAVCs}
R.~Ahlswede and J.~Wolfowitz, ``The capacity of a channel with arbitrarily
  varying channel probability functions and binary output alphabet,'' {\em Z
  WAHRSCHEINLICHKEIT}, vol.~15, no.~3, pp.~186--194, 1970.

\bibitem{MartonExcessDistortion}
K.~Marton, ``Error exponent for source coding with a fidelity criterion,'' {\em
  IEEE Trans. Inf. Theory}, vol.~20, pp.~197--199, March 1974.

\bibitem{knuth}
D.~E. Knuth, {\em The Art of Computer Programming, volume 1: Fundamental
  Algorithms Addison-Wesley}.
\newblock Addison-Wesley Professional, 1997.

\bibitem{burshtein2008error}
D.~Burshtein, ``On the error correction of regular {LDPC} codes using the
  flipping algorithm,'' {\em IEEE Trans. Inf. Theory}, vol.~54, no.~2,
  pp.~517--530, 2008.

\bibitem{ref_3}
E.~E. X.-Y.~Hu and D.~M. Arnold, ``Regular and irregular progressive
  edge-growth tanner graphs,'' {\em IEEE Trans. Inf. Theory}, vol.~51, no.~1,
  pp.~386--398, 2005.

\bibitem{Ber_ICC_93}
C.~Berrou and A.~Glavieux, ``Near optimum error correcting coding and decoding:
  turbo-codes,'' {\em IEEE Trans. Commun.}, vol.~44, pp.~1261--1271, Oct 1996.

\bibitem{Abb_TCom_07}
A.~Abbasfar, D.~Divsalar, and K.~Yao, ``Accumulate-repeat-accumulate codes,''
  {\em IEEE Trans. Commun.}, vol.~55, pp.~692--702, April 2007.

\bibitem{Cost_PHE_04}
S.~Lin and D.~Costello, {\em Error control coding}.
\newblock 2004.

\bibitem{Vas_ITW_06}
B.~Vasic and S.~K. Chilappagari, ``An information theoretical framework for
  analysis and design of nanoscale fault-tolerant memories based on low-density
  parity-check codes,'' {\em IEEE Trans. Circuits Syst. I, Reg. Papers},
  vol.~54, no.~11, pp.~2438--2446, 2007.

\bibitem{du2013network}
J.~Du and Y.-C. Wu, ``Network-wide distributed carrier frequency offsets
  estimation and compensation via belief propagation,'' {\em IEEE Trans. Signal
  Process}, vol.~61, no.~23, pp.~5868--5877, 2013.

\bibitem{Ric_TIT_01}
T.~Richardson and R.~Urbanke, ``Efficient encoding of low-density parity-check
  codes,'' {\em IEEE Trans. Inf. Theory}, vol.~47, pp.~638--656, Feb 2001.

\bibitem{Che_AMS_52}
H.~Chernoff, ``A measure of asymptotic efficiency for tests of a hypothesis
  based on the sum of observations,'' {\em Ann. Math. Stat.}, pp.~493--507,
  1952.

\bibitem{Kou_TIT_01}
Y.~Kou, S.~Lin, and M.~Fossorier, ``Low-density parity-check codes based on
  finite geometries: a rediscovery and new results,'' {\em IEEE Trans. Inf.
  Theory}, vol.~47, pp.~2711--2736, Nov 2001.

\bibitem{Mou_SPM_04}
J.~Moura, J.~Lu, and H.~Zhang, ``Structured low-density parity-check codes,''
  {\em IEEE Signal Process. Mag.}, vol.~21, pp.~42--55, Jan 2004.

\bibitem{Tho_INP_03}
J.~Thorpe, ``Low-density parity-check (ldpc) codes constructed from
  protographs,'' {\em IPN progress report}, vol.~42, no.~154, pp.~42--154,
  2003.

\bibitem{Fos_TIT_04}
M.~Fossorier, ``Quasicyclic low-density parity-check codes from circulant
  permutation matrices,'' {\em IEEE Trans. Inf. Theory}, vol.~50,
  pp.~1788--1793, Aug 2004.

\bibitem{yang2015computing}
Y.~Yang, P.~Grover, and S.~Kar, ``Computing linear transformations with
  unreliable components,'' {\em arXiv:1506.07234}, 2015.

\bibitem{Ernst_Micro_03}
D.~Ernst, N.~S. Kim, S.~Das, S.~Pant, R.~Rao, T.~Pham, C.~Ziesler, D.~Blaauw,
  T.~Austin, K.~Flautner, and T.~Mudge, ``Razor: a low-power pipeline based on
  circuit-level timing speculation,'' in {\em Proc. 36th Int'l Symp.
  Microarchitecture}, pp.~7--18, Dec 2003.

\bibitem{Patil_spin}
A.~D. Patil, S.~Manipatruni, D.~Nikonov, I.~A. Young, and N.~R. Shanbhag,
  ``Shannon-inspired statistical computing to enable spintronics,'' {\em arXiv
  preprint arXiv:1702.06119}, 2017.

\bibitem{butler2012switching}
W.~H. Butler, T.~Mewes, C.~K.~A. Mewes, P.~B. Visscher, W.~H. Rippard, S.~E.
  Russek, and R.~Heindl, ``Switching distributions for perpendicular
  spin-torque devices within the macrospin approximation,'' {\em IEEE Trans.
  Magn.}, vol.~48, no.~12, pp.~4684--4700, 2012.

\bibitem{kim2015spin}
J.~Kim, A.~Paul, P.~A. Crowell, S.~J. Koester, S.~S. Sapatnekar, J.-P. Wang,
  and C.~H. Kim, ``Spin-based computing: device concepts, current status, and a
  case study on a high-performance microprocessor,'' {\em Proc. IEEE},
  vol.~103, no.~1, pp.~106--130, 2015.

\bibitem{manipatruni2012modeling}
S.~Manipatruni, D.~E. Nikonov, and I.~A. Young, ``Modeling and design of
  spintronic integrated circuits,'' {\em IEEE Trans. Circuits Syst. I, Reg.
  Papers}, vol.~59, no.~12, pp.~2801--2814, 2012.

\bibitem{Chi_ITW_06}
S.~Chilappagari, M.~Ivkovic, and B.~Vasic, ``Analysis of one step majority
  logic decoders constructed from faulty gates,'' in {\em Proc. IEEE Int. Symp.
  Inf. Theory}, pp.~469--473, July 2006.

\bibitem{KarthikSips}
K.~Ganesan, P.~Grover, and J.~Rabaey, ``The power cost of over-designing
  codes,'' in {\em Proceedings of IEEE Workshop on Signal Processing Systems
  (SiPS)}, pp.~128--133, Oct 2011.

\bibitem{Kud_ISIT_12}
S.~Kudekar, T.~Richardson, and R.~Urbanke, ``Spatially coupled ensembles
  universally achieve capacity under belief propagation,'' {\em IEEE Trans.
  Inf. Theory}, vol.~59, pp.~7761--7813, Dec 2013.

\bibitem{lee2016speeding}
K.~Lee, M.~Lam, R.~Pedarsani, D.~Papailiopoulos, and K.~Ramchandran, ``Speeding
  up distributed machine learning using codes,'' in {\em Proc. IEEE Int. Symp.
  Inf. Theory}, pp.~1143--1147, IEEE, 2016.

\bibitem{YaoqingISTC}
Y.~Yang, P.~Grover, and S.~Kar, ``Fault-tolerant parallel linear filtering
  using compressive sensing,'' in {\em Proc. Int. Symp. Turbo Codes \&
  Iterative Information Processing}, pp.~201--205, IEEE, 2016.

\bibitem{dutta2016short}
S.~Dutta, V.~Cadambe, and P.~Grover, ``Short-dot: Computing large linear
  transforms distributedly using coded short dot products,'' in {\em Advances
  In Neural Information Processing Systems}, pp.~2092--2100, 2016.

\bibitem{duttaISIT2017}
S.~Dutta, V.~Cadambe, and P.~Grover, ``Coded convolution can provide
  arbitrarily large gains in successfully computing before a deadline,'' in
  {\em IEEE International Symposium on Information Theory (ISIT)}, July 2017.

\bibitem{PedarsaniHeterogeneous}
A.~Reisizadehmobarakeh, S.~Prakash, R.~Pedarsani, and S.~Avestimehr, ``Coded
  computation over heterogeneous clusters,'' in {\em 2017 Proc. Workshop Inf.
  Theory Appl.}, IEEE, 2017.

\bibitem{cappello2014toward}
F.~Cappello, A.~Geist, W.~Gropp, S.~Kale, B.~Kramer, and M.~Snir, ``Toward
  exascale resilience: 2014 update,'' {\em Supercomputing frontiers and
  innovations}, vol.~1, no.~1, pp.~5--28, 2014.

\bibitem{kaminsky2016big}
A.~Kaminsky, ``{BIG} {CPU}, {BIG} {DATA}: Solving the world's toughest
  computational problems with parallel computing,'' 2016.

\bibitem{van1997summa}
R.~A. Van De~Geijn and J.~Watts, ``Summa: Scalable universal matrix
  multiplication algorithm,'' {\em Concurrency-Practice and Experience},
  vol.~9, no.~4, pp.~255--274, 1997.

\bibitem{yang2016fault}
Y.~Yang, P.~Grover, and S.~Kar, ``Fault-tolerant distributed logistic
  regression using unreliable components,'' in {\em Proc. Allerton Conf. on
  Commun., Control and Comput.}, pp.~940--947, IEEE, 2016.

\bibitem{zhang2012verification}
F.~Zhang and H.~D. Pfister, ``Verification decoding of high-rate {LDPC} codes
  with applications in compressed sensing,'' {\em IEEE Trans. Inf. Theory},
  vol.~58, no.~8, pp.~5042--5058, 2012.

\bibitem{capalbo2002randomness}
M.~Capalbo, O.~Reingold, S.~Vadhan, and A.~Wigderson, ``Randomness conductors
  and constant-degree lossless expanders,'' in {\em Proc. 34th ACM Symp. Theory
  Comput.}, pp.~659--668, ACM, 2002.

\bibitem{guruswami2009unbalanced}
V.~Guruswami, C.~Umans, and S.~Vadhan, ``Unbalanced expanders and randomness
  extractors from {P}arvaresh-{V}ardy codes,'' {\em J. ACM}, vol.~56, no.~4,
  p.~20, 2009.

\bibitem{burshtein2001expander}
D.~Burshtein and G.~Miller, ``Expander graph arguments for message-passing
  algorithms,'' {\em IEEE Trans. Inf. Theory}, vol.~47, no.~2, pp.~782--790,
  2001.

\bibitem{feldman2007lp}
J.~Feldman, T.~Malkin, R.~A. Servedio, C.~Stein, and M.~J. Wainwright, ``{LP}
  decoding corrects a constant fraction of errors,'' {\em IEEE Trans. Inf.
  Theory}, vol.~53, no.~1, pp.~82--89, 2007.

\bibitem{richardson2008modern}
T.~Richardson and R.~Urbanke, {\em Modern coding theory}.
\newblock Cambridge University Press, 2008.

\end{thebibliography}

\end{document}